\def\llncs{0}
\def\fullpage{1}
\def\anonymous{0}
\def\authnote{0}
\def\notxfont{0}
\def\submission{0}
\def\cameraready{0}
\def\noclassic{0}
\def\arxiv{1}
\ifnum\submission=1
\def\anonymous{1}
\def\llncs{1}
\else
\fi

\ifnum\submission=1
\def\llncs{1}
\def\authnote{0}
\def\anonymous{0}
\else
\fi

\ifnum\anonymous=1
\def\llncs{1}
\def\authnote{0}
\else
\fi

\ifnum\llncs=1
	\documentclass[envcountsect,a4paper,runningheads,10pt]{llncs}
\else
	\documentclass[letterpaper,hmargin=1.05in,vmargin=1.05in]{article}
			\ifnum\fullpage=1
		\usepackage{fullpage}
		\fi
\fi

\usepackage[%
  colorlinks=true,
  citecolor=blue,
  pagebackref=true
]{hyperref}

\usepackage{amsmath, amsfonts, amssymb, mathtools,amscd}

\usepackage{amsthm}

\usepackage{lmodern}
\usepackage[T1]{fontenc}
\usepackage[utf8]{inputenc}

\usepackage{arydshln} 
\usepackage{url}
\usepackage{ifthen}
\usepackage{bm}
\usepackage{multirow}
\usepackage[dvips]{graphicx}
\usepackage[usenames]{color}
\usepackage{xcolor,colortbl} 
\usepackage{threeparttable}
\usepackage{comment}
\usepackage{paralist,verbatim}
\usepackage{cases}
\usepackage{booktabs}
\usepackage{braket}
\usepackage{cancel} 
\usepackage{ascmac} 
\usepackage{framed}
\usepackage{authblk}
\usepackage{pifont}
\usepackage{autonum}
\definecolor{darkblue}{rgb}{0,0,0.6}
\definecolor{darkgreen}{rgb}{0,0.5,0}
\definecolor{maroon}{rgb}{0.5,0.1,0.1}
\definecolor{dpurple}{rgb}{0.2,0,0.65}

\usepackage[capitalise,noabbrev]{cleveref}
\usepackage[absolute]{textpos}
\usepackage[final]{microtype}
\usepackage[absolute]{textpos}
\usepackage{everypage}

\newtheoremstyle{thicktheorem}%
{\topsep}
{\topsep}
{\itshape}{}%
{\bfseries}%
{.}
{ }%
{\thmname{#1}\thmnumber{ #2}%
		\thmnote{ (#3)}%
}

\newtheoremstyle{remark}
{\topsep}
{\topsep}
	{}
	{}
	{}
	{.}
	{ }
	{\textit{\thmname{#1}}\thmnumber{ #2}
			\thmnote{ (#3)}%
	}

\ifnum\llncs=0
	\theoremstyle{thicktheorem}
	\newtheorem{theorem}{Theorem}[section]
	\newtheorem{lemma}[theorem]{Lemma}
	\newtheorem{corollary}[theorem]{Corollary}
	\newtheorem{proposition}[theorem]{Proposition}
	\newtheorem{definition}[theorem]{Definition}
	\newtheorem{game}[theorem]{Game}

	\theoremstyle{remark}
	
	\newtheorem{remark}[theorem]{Remark}

\else
\fi

	\crefname{theorem}{Theorem}{Theorems}
	\crefname{assumption}{Assumption}{Assumptions}
	\crefname{construction}{Construction}{Constructions}
	\crefname{corollary}{Corollary}{Corollaries}
	\crefname{conjecture}{Conjecture}{Conjectures}
	\crefname{definition}{Definition}{Definitions}
	\crefname{exmaple}{Example}{Examples}
	\crefname{experiment}{Experiment}{Experiments}
	\crefname{counterexample}{Counterexample}{Counterexamples}
	\crefname{lemma}{Lemma}{Lemmata}
	\crefname{observation}{Observation}{Observations}
	\crefname{proposition}{Proposition}{Propositions}
	\crefname{remark}{Remark}{Remarks}
	\crefname{claim}{Claim}{Claims}
	\crefname{fact}{Fact}{Facts}
	\crefname{note}{Note}{Notes}

\ifnum\llncs=1
 \crefname{appendix}{App.}{Appendices}
 \crefname{section}{Sec.}{Sections}
\else
\fi

\ifnum\llncs=1
\renewcommand*{\backref}[1]{}
\else
	\renewcommand*{\backref}[1]{(Cited on page~#1.)}
	\ifnum\notxfont=1
	\else
		\usepackage{newtxtext}
	\fi
\fi

\ifnum\authnote=0  
\newcommand{\mor}[1]{}
\newcommand{\taiga}[1]{}
\newcommand{\ryo}[1]{}
\newcommand{\takashi}[1]{}

\newcommand{\revise}[1]{#1}
\else
\newcommand{\mor}[1]{$\ll$\textsf{\color{red} Tomoyuki: { #1}}$\gg$}
\newcommand{\taiga}[1]{$\ll$\textsf{\color{magenta} Taiga: { #1}}$\gg$}
\newcommand{\takashi}[1]{$\ll$\textsf{\color{orange} Takashi: { #1}}$\gg$}
\newcommand{\ryo}[1]{$\ll$\textsf{\color{darkgreen} Ryo: { #1}}$\gg$}

\newcommand{\revise}[1]{{\color{purple}#1}} 
\fi

\newcommand{\Delete}{\algo{Del}}
\newcommand{\cert}{\keys{cert}}
\newcommand{\skcd}{\mathsf{skcd}}
\newcommand{\pkcd}{\mathsf{pkcd}}
\newcommand{\cccd}{\mathsf{cccd}}
\newcommand{\sen}{\mathcal{S}}
\newcommand{\rec}{\mathcal{R}}
\newcommand{\ctmsg}{\ct_{\mathsf{msg}}}

\newcommand{\OW}{\mathsf{OW}}
\newcommand{\ow}{\mathsf{ow}}

\newcommand{\NCE}{\mathsf{NCE}}
\newcommand{\nce}{\mathsf{nce}}
\newcommand{\Fake}{\algo{Fake}}
\newcommand{\Reveal}{\algo{Reveal}}
\newcommand{\tlsk}{\widetilde{\keys{sk}}}
\newcommand{\tldk}{\widetilde{\keys{dk}}}
\newcommand{\tlmsk}{\widetilde{\keys{msk}}}
\newcommand{\FakeCT}{\algo{FakeCT}}
\newcommand{\FakeSK}{\algo{FakeSK}}
\newcommand{\FakeSetup}{\algo{FakeSetup}}

\newcommand{\iO}{i\cO}
\newcommand{\sfD}{\mathsf{D}}

\newcommand{\NIZK}{\algo{NIZK}}
\newcommand{\nizk}{\mathsf{nizk}}
\newcommand{\Prove}{\algo{Prove}}
\newcommand{\Sim}{\algo{Sim}}

\newcommand{\sfinvalid}{\mathsf{Invalid}}
\newcommand{\NP}{\compclass{NP}}
\newcommand{\Lang}{\mathcal{L}}
\newcommand{\Rela}{\mathcal{R}}

\newcommand{\GenF}{\algo{Gen}_{\mathcal{F}}}
\newcommand{\GenG}{\algo{Gen}_{\mathcal{G}}}
\newcommand{\InvF}{\algo{Inv}_{\mathcal{F}}}
\newcommand{\InvG}{\algo{Inv}_{\mathcal{G}}}
\newcommand{\Chk}{\algo{Chk}}
\newcommand{\Samp}{\algo{Samp}}
\newcommand{\Supp}{\mathrm{Supp}}
\newcommand{\fk}{\mathsf{k}}

\newcommand{\Bad}{\mathtt{Bad}}

\newcommand{\WE}{\mathsf{WE}}
\newcommand{\we}{\mathsf{we}}

\newcommand{\OSS}{\mathsf{OSS}}
\newcommand{\oss}{\mathsf{oss}}

\newcommand{\chosen}{\leftarrow}

\newcommand{\lrun}{\leftarrow}

\newcommand{\la}{\leftarrow}
\newcommand{\ra}{\rightarrow}

\newcommand{\seteq}{\coloneqq}

\newcommand{\concat}{\|}

\newcommand{\setbracket}[1]{\{#1\}}
\newcommand{\setbk}[1]{\{#1\}}
\newcommand{\abs}[1]{|#1|}

\newcommand{\cA}{\mathcal{A}}
\newcommand{\cB}{\mathcal{B}}
\newcommand{\cC}{\mathcal{C}}
\newcommand{\cD}{\mathcal{D}}
\newcommand{\cE}{\mathcal{E}}
\newcommand{\cF}{\mathcal{F}}
\newcommand{\cG}{\mathcal{G}}

\newcommand{\cK}{\mathcal{K}}

\newcommand{\cM}{\mathcal{M}}

\newcommand{\cO}{\mathcal{O}}

\newcommand{\cR}{\mathcal{R}}
\newcommand{\cS}{\mathcal{S}}

\newcommand{\cX}{\mathcal{X}}
\newcommand{\cY}{\mathcal{Y}}

\def\tlpi{\widetilde{\pi}}

\def\makeuppercase#1{
\expandafter\newcommand\csname tl#1\endcsname{\widetilde{#1}}
}

\def\makelowercase#1{
\expandafter\newcommand\csname tl#1\endcsname{\widetilde{#1}}
}

\newcommand{\N}{\mathbb{N}}

\newcommand{\R}{\mathbb{R}}

\newcommand{\Ms}{\mathcal{M}}

\newcommand{\Ks}{\mathcal{K}}

\newcommand{\Xs}{\mathcal{X}}
\newcommand{\Ps}{\mathcal{P}}

\newcommand{\secp}{\lambda}

\newcommand{\crs}{\mathsf{crs}}
\newcommand{\tlcrs}{\widetilde{\mathsf{crs}}}
\newcommand{\aux}{\mathsf{aux}}

\newcommand{\A}{\entity{A}}

\newcommand{\adva}[2]{\mathsf{Adv}_{#1}^{\mathsf{#2}}}
\newcommand{\advb}[3]{\mathsf{Adv}_{#1}^{\mathsf{#2} \mbox{-} \mathsf{#3}}}
\newcommand{\advc}[4]{\mathsf{Adv}_{#1}^{\mathsf{#2} \mbox{-} \mathsf{#3} \mbox{-} \mathsf{#4}}}
\newcommand{\advd}[5]{\mathsf{Adv}_{#1}^{\mathsf{#2} \mbox{-} \mathsf{#3} \mbox{-} \mathsf{#4} \mbox{-} \mathsf{#5}}}

\newcommand{\expa}[3]{\mathsf{Exp}_{#1}^{ \mathsf{#2} \mbox{-} \mathsf{#3}}}
\newcommand{\expb}[4]{\mathsf{Exp}_{#1}^{ \mathsf{#2} \mbox{-} \mathsf{#3} \mbox{-} \mathsf{#4}}}
\newcommand{\wtexpb}[4]{\widetilde{\mathsf{Exp}}_{#1}^{ \mathsf{#2} \mbox{-} \mathsf{#3} \mbox{-} \mathsf{#4}}}
\newcommand{\expc}[4]{\mathsf{Exp}_{#1}^{ \mathsf{#2} \mbox{-} \mathsf{#3} \mbox{-} \mathsf{#4}}}

\newcommand{\sfhyb}[2]{\mathsf{Hyb}_{#1}^{#2}}
\newcommand{\hybi}[1]{\mathsf{Hyb}_{#1}}
\newcommand{\hybij}[2]{\mathsf{Hyb}_{#1}^{#2}}
\newcommand{\revealsk}{\mathtt{Reveal}_{\sk}}

\newcommand{\sfnull}{\mathsf{null}}

\newcommand*{\sk}{\keys{sk}}
\newcommand*{\pk}{\keys{pk}}
\newcommand*{\msk}{\keys{msk}}

\newcommand{\ct}{\keys{CT}}

\newcommand*{\SK}{\keys{SK}}

\newcommand*{\dk}{\keys{dk}}
\newcommand*{\vk}{\keys{vk}}
\newcommand*{\ek}{\keys{ek}}

\newcommand*{\td}{\keys{td}}

\newcommand*{\tlct}{\widetilde{\ct}}

\newcommand{\CT}{\keys{CT}}

\newcommand*{\msg}{\keys{msg}}

\newcommand*{\keys}[1]{\mathsf{#1}}

\newcommand*{\algo}[1]{\ensuremath{\mathsf{#1}}}

\newcommand*{\entity}[1]{\mathcal{#1}}

\newcommand{\compclass}[1]{\textbf{\textrm{#1}}}

\newenvironment{boxfig}[2]{\begin{figure}[#1]\fbox{\begin{minipage}{0.97\linewidth}
                        \vspace{0.2em}
                        \makebox[0.025\linewidth]{}
                        \begin{minipage}{0.95\linewidth}
            {{
                        #2 }}
                        \end{minipage}
                        \vspace{0.2em}
                        \end{minipage}}}{\end{figure}}

\newcommand{\pprotocol}[4]{
\begin{boxfig}{h!}{\footnotesize 
\centering{\textbf{#1}}
    #4
\vspace{0.2em} } \caption{\label{#3} #2}
\end{boxfig}
}

\newcommand{\protocol}[4]{
\pprotocol{#1}{#2}{#3}{#4} }

\newcommand{\bit}{\{0,1\}}

\newcommand{\Setup}{\algo{Setup}}

\newcommand{\Gen}{\algo{Gen}}
\newcommand{\KeyGen}{\algo{KeyGen}}
\newcommand{\keygen}{\algo{KeyGen}}

\newcommand{\Enc}{\algo{Enc}}
\newcommand{\Dec}{\algo{Dec}}

\newcommand{\Sign}{\algo{Sign}}
\newcommand{\Vrfy}{\algo{Vrfy}}

\newcommand{\vrfy}{\algo{Vrfy}}

\newcommand\SKE{\algo{SKE}}

\newcommand{\ske}{\mathsf{ske}}

\newcommand{\oske}{\mathsf{oske}}

\newcommand\PKE{\algo{PKE}}

\newcommand{\pke}{\mathsf{pke}}

\newcommand\ABE{\algo{ABE}}

\newcommand{\abe}{\mathsf{abe}}

\newcommand{\negl}{{\mathsf{negl}}}

\newcommand{\poly}{{\mathrm{poly}}}

\newcommand{\zo}[1]{\{0,1\}^{#1}}

\newcommand{\xor}{\oplus}

\newcommand{\Ppoly}{\compclass{P}/\compclass{poly}}

\usepackage{fancyhdr}

\ifnum\llncs=1
\let\oldvec\vec
\let\vec\oldvec
\makeatletter
\renewcommand*\l@author[2]{}
\renewcommand*\l@title[2]{}
\makeatletter

\else
\fi    

\theoremstyle{remark}

\title{
\textbf{Quantum Encryption with Certified Deletion, Revisited:\\ Public Key, Attribute-Based, and Classical Communication}
\thanks{This is a major update version of the paper by Nishimaki and Yamakawa~\cite{EPRINT:NisYam21} with many new results.}
}

\begin{document}

\ifnum\anonymous=1
\author{\empty}\institute{\empty}
\else
%
%
\ifnum\llncs=1
\author{
	Taiga Hiroka\inst{1} \and Tomoyuki Morimae\inst{1,2} \and Ryo Nishimaki\inst{3} \and Takashi Yamakawa\inst{3}
}
\institute{
	Yukawa Institute for Theoretical Physics, Kyoto University, Japan \and PRESTO, JST, Japan \and NTT Secure Platform Laboratories
}
\else
%
%
\ifnum\noclassic=1
\author[1]{Ryo Nishimaki}
\author[1]{Takashi Yamakawa}
\affil[1]{{\small NTT Secure Platform Laboratories, Tokyo, Japan}\authorcr{\small \{ryo.nishimaki.zk,takashi.yamakawa.ga\}@hco.ntt.co.jp}}
\else
\author[1]{Taiga Hiroka}
\author[1,2]{\hskip 1em Tomoyuki Morimae}
\author[3]{\hskip 1em Ryo Nishimaki}
\author[3]{\hskip 1em Takashi Yamakawa}
\affil[1]{{\small Yukawa Institute for Theoretical Physics, Kyoto University, Japan}\authorcr{\small \{taiga.hiroka,tomoyuki.morimae\}@yukawa.kyoto-u.ac.jp}}
\affil[2]{{\small PRESTO, JST, Japan}}
\affil[3]{{\small NTT Secure Platform Laboratories, Tokyo, Japan}\authorcr{\small \{ryo.nishimaki.zk,takashi.yamakawa.ga\}@hco.ntt.co.jp}}
\fi
\renewcommand\Authands{, }
\fi 
\fi

\ifnum\llncs=1
\date{}
\else
\date{\today}
\fi
\maketitle
\thispagestyle{fancy}
\rhead{YITP-21-40}

\begin{abstract}
Broadbent and Islam (TCC '20) proposed a quantum cryptographic primitive called \emph{quantum encryption with certified deletion}.
In this primitive, a receiver in possession of a quantum ciphertext can generate a classical certificate that the encrypted message is deleted.
Although their construction is information-theoretically secure, it is limited to the setting of one-time symmetric key encryption (SKE), where a sender and receiver have to share a common key in advance and the key can be used only once. Moreover, the sender has to generate a quantum state and send it to the receiver over a quantum channel in their construction.
\revise{Although deletion certificates are privately verifiable, which means a verification key for a certificate has to be kept secret, in the definition by Broadbent and Islam, we can also consider public verifiability.}

In this work, we present various constructions of encryption with certified deletion.
\begin{itemize}
\item Quantum communication case: We achieve (reusable-key) public key encryption (PKE) and attribute-based encryption (ABE) with certified deletion.
Our PKE scheme with certified deletion is constructed assuming the existence of IND-CPA secure PKE, and our ABE scheme with certified deletion is constructed assuming the existence of indistinguishability obfuscation and one-way function. \revise{These two schemes are privately verifiable.}
\item Classical communication case: We also achieve PKE  with certified deletion that uses only classical communication. 
We give two schemes, a privately verifiable one and a publicly verifiable one. The former is constructed assuming the LWE assumption in the quantum random oracle model. The latter is constructed assuming the existence of one-shot signatures and extractable witness encryption.
\end{itemize}
\end{abstract}

\ifnum\cameraready=1
\else
\ifnum\llncs=1
\else
\newpage
  \setcounter{tocdepth}{2}      
  \setcounter{secnumdepth}{2}   
  \setcounter{page}{0}          
  \tableofcontents
  \thispagestyle{empty}
  \clearpage
\fi
\fi

\section{Introduction}\label{sec:intro}
The no-cloning theorem,  which means that an unknown quantum state cannot be copied in general, is one of the most fundamental principles in quantum physics. 
As any classical information can be trivially copied, this indicates a fundamental difference between classical and quantum information.
The no-cloning theorem has been the basis of many quantum cryptographic protocols, including quantum money \cite{Wiesner83} and quantum key distribution~\cite{BB84}.  

Broadbent and Islam \cite{TCC:BroIsl20} used the principle to construct \emph{quantum encryption with certified deletion}. 
In this primitive,  
a sender encrypts a classical message to generate a quantum ciphertext. 
A receiver in possession of the quantum ciphertext and a classical decryption key can either decrypt the ciphertext or ``delete" the encrypted message by generating a classical certificate. 
After generating a valid certificate of deletion, 
no adversary can recover the message \emph{even if the decryption key is given}.\footnote{We note that if the adversary is given the decryption key before the deletion, it can decrypt the ciphertext to obtain the message and keep it even after the deletion, but such an ``attack" is unavoidable.}
We remark that this functionality is classically impossible to achieve since one can copy a classical ciphertext and keep it so that s/he can decrypt it at any later time.
They prove the security of their construction without relying on any computational assumption, 
which ensures information-theoretical security.
Although they achieved the exciting new functionality, their construction is limited to the one-time symmetric key encryption (SKE) setting. In one-time SKE, a sender and receiver have to share a common key in advance, and the key can be used only once.

A possible application scenario of quantum encryption with certified deletion is the following. 
A user uploads encrypted data on a quantum cloud server. Whenever the user wishes to delete the data, the cloud generates a deletion certificate and sends it to the user.
After the user verifies the validity of the certificate, s/he is convinced that the data cannot be recovered even if the decryption key is accidentally leaked later.
Such quantum encryption could prevent data retention and help to implement the right to be forgotten~\cite{GDPR16}.
In this scenario, one-time SKE is quite inconvenient. By the one-time restriction, the user has to locally keep as many decryption keys as the number of encrypted data in the cloud, in which case there seems to be no advantage of uploading the data to the cloud server:
If the user has such large storage, 
s/he could have just locally kept the messages rather than uploading encryption of them to the cloud. 
Also, in some cases, a party other than the decryptor may want to upload data to the cloud. 
This usage would be possible if we can extend the quantum encryption with certified deletion to public key encryption (PKE). Remark that the one-time restriction is automatically resolved for PKE by a simple hybrid argument.
Even more flexibly, a single encrypted data on the cloud may be supposed to be decrypted by multiple users according to some access control policy.    
Such an access control has been realized by attribute-based encryption (ABE)~\cite{EC:SahWat05,CCS:GPSW06} in classical cryptography.
Thus, it would be useful if we have ABE with certified deletion. Our first question in this work is:
\begin{center}
    \emph{Can we achieve PKE and ABE with certified deletion?}
\end{center}

Moreover, a sender needs to send quantum states (random BB84 states~\cite{BB84}) over a quantum channel in the construction by Broadbent and Islam~\cite{TCC:BroIsl20}.
Although generating and sending random BB84 states are not so difficult tasks (and in fact they are
already possible with current technologies), 
a classical sender and communication over only a classical channel are of course much easier.
Besides, communicating over a classical channel is desirable in the application scenario above since many parties want to upload data to a cloud.
In addition to these practical motivations, furthermore, achieving classical channel certified deletion is also an interesting theoretical research direction given the fact that many quantum cryptographic protocols have been ``dequantized" recently~\cite{FOCS:Mahadev18a,AC:CCKW19,CoRR:RadSat19,STOC:AGKZ20,EPRINT:KitNisYam20}.
Thus, our second question in this work is:
\begin{center}
    \emph{Can we achieve PKE with certified deletion, a classical sender, and classical communication?}
\end{center}

In the definition by Broadbent and Islam~\cite{TCC:BroIsl20}, a verification key for a deletion certificate must be kept secret (privately verifiable). If the verification key is revealed, the security is no longer guaranteed in their scheme.
We can also consider public verifiability, which means the security holds even if a verification key is revealed to adversaries.
Broadbent and Islam left the following question as an open problem:
\begin{center}
\emph{Is publicly verifiable encryption with certified deletion possible?}
\end{center} 

\subsection{Our Result}
We solve the three questions above affirmatively in this work.

\paragraph{PKE and ABE with certified deletion and quantum communication.}
We present formal definitions of PKE and ABE with certified deletion, and present constructions of them:
\begin{itemize}
    \item We construct a PKE scheme with certified deletion assuming the existence of (classical) IND-CPA secure PKE.
    We also observe that essentially the same construction gives a reusable SKE scheme with certified deletion if we use IND-CPA secure SKE, which exists under the existence of one-way function (OWF), instead of PKE. 
    \item We construct a (public-key) ABE scheme with certified deletion assuming the existence of indistinguishability obfuscation (IO)~\cite{JACM:BGIRSVY12} and OWF.
    This construction satisfies the collusion resistance and adaptive security, i.e., it is secure against adversaries that adaptively select a target attribute and obtain arbitrarily many decryption keys.
\end{itemize}
We note that our constructions rely on computational assumptions and thus not information-theoretically secure, unlike the construction in \cite{TCC:BroIsl20}.
This is unavoidable since even plain PKE or ABE cannot be information-theoretically secure. We also note that the constructions above are privately verifiable as the definition of one-time SKE by Broadbent and Islam~\cite{TCC:BroIsl20}.

Our main technical insight is that we can combine the one-time secure SKE with certified deletion of \cite{TCC:BroIsl20} and plain PKE to construct PKE with certified deletion by a simple hybrid encryption technique
if the latter satisfies \emph{receiver non-committing} (RNC) security~\cite{STOC:CFGN96,EC:JarLys00,TCC:CanHalKat05}.
Since it is known that PKE/SKE with RNC security can be constructed from any IND-CPA secure PKE/SKE \cite{TCC:CanHalKat05,C:KNTY19}, our first result follows. 

For the second result, we first give a suitable definition of RNC security for  ABE that suffices for our purpose. 
Then we construct an ABE scheme with RNC security based on the existence of IO and OWF. 
By combining this with one-time SKE with certified deletion by hybrid encryption, we obtain an ABE scheme with certified deletion.

\paragraph{PKE with certified deletion, a classical sender, and classical communication.}
We also present formal definitions of PKE with certified deletion and classical communication, and present two constructions:
\begin{itemize}
\item
We construct a PKE scheme with privately verifiable certified deletion and classical communication in the quantum random oracle model (QROM)~\cite{AC:BDFLSZ11}. 
Our construction is secure under the LWE assumption in the QROM.
\item
We construct a PKE scheme with publicly verifiable certified deletion and classical communication.
Our construction uses one-shot signatures~\cite{STOC:AGKZ20} and extractable witness encryption~\cite{STOC:GGSW13,C:GKPVZ13}.
This solves the open problem by Broadbent and Islam~\cite{TCC:BroIsl20}. 
\end{itemize}
In both constructions, a sender is a classical algorithm, but needs to interact with a receiver during ciphertext generation.

In the classical communication case, an encryption algorithm must be interactive even if we consider computationally bounded adversaries (and even in the QROM). The reason is that a malicious QPT receiver can generate two copies of a quantum ciphertext from classical messages sent from a sender,
and one is used for generating a deletion certificate and the other is used for decryption.

Moreover, both constructions rely on computational assumptions and thus not information-theoretically secure, unlike the construction by Broadbent and Islam~\cite{TCC:BroIsl20}.
This is unavoidable even if an encryption algorithm is interactive (and even in the QROM). The reason is that a computationally unbounded malicious receiver can classically simulate its honest behavior to get a
classical description of the quantum ciphertext. 

For the first construction, we use a new property of noisy trapdoor claw-free (NTCF) functions, \emph{the cut-and-choose adaptive hardcore property} (\cref{lem:cut_and_choose_adaptive_hardcore}), which we introduce in this work. We prove that the cut-and-choose adaptive hardcore property is reduced to the adaptive hardcore bit property~\cite{FOCS:BCMVV18}
and injective invariance~\cite{FOCS:Mahadev18a}.  Those properties hold under the LWE assumption~\cite{FOCS:BCMVV18,FOCS:Mahadev18a}. This new technique is of independent interest.
The idea of the second construction is to encrypt a plaintext by witness encryption so that
a valid witness is a one-shot signature for bit $0$.  
We use a valid one-shot signature for bit $1$ as a deletion certificate. The one-shot property of one-shot signatures prevents decryption of witness encryption after issuing a valid deletion certificate.
Georgiou and Zhandry~\cite{EPRINT:GeoZha20} used a similar combination of one-shot signatures and witness encryption to construct unclonable decryption keys.

\subsection{Related work}\label{sec:related_work}
Before the work by Broadbent and Islam~\cite{TCC:BroIsl20}, Fu and Miller~\cite{FM18} and Coiteux-Roy and Wolf~\cite{CRW19} also studied the concept of certifying deletion of information in different settings. (See~\cite{TCC:BroIsl20} for the comparison with these works.)

The quantum encryption scheme with certified deletion by Broadbent and Islam~\cite{TCC:BroIsl20} is based on Wiesner's conjugate coding, which is the backbone of quantum money~\cite{Wiesner83} and quantum key distribution~\cite{BB84}.
A similar idea has been used in many constructions in quantum cryptography that include (but not limited to) revocable quantum timed-release encryption~\cite{JACM:Unruh15}, uncloneable quantum encryption~\cite{TQC:BroLor20}, single-decryptor encryption~\cite{EPRINT:GeoZha20}, and copy protection/secure software leasing~\cite{CMP20}. 
Among them, revocable quantum timed-release encryption is conceptually similar to quantum encryption with certified deletion. In this primitive, a receiver can decrypt a quantum ciphertext only after spending a certain amount of time $T$.
The receiver can also choose to return the ciphertext before the time $T$ is over, in which case it is ensured that the message can no longer be recovered. As observed by Broadbent and Islam~\cite{TCC:BroIsl20}, an essential difference from quantum encryption with certified deletion is that the revocable quantum timed-release encryption does not have a mechanism to generate a \emph{classical} certificate of deletion. 
Moreover, the construction by Unruh~\cite{JACM:Unruh15}  heavily relies on the random oracle heuristic~\cite{C:BelRog97,AC:BDFLSZ11}, and there is no known construction without random oracles.

Kundu and Tan~\cite{KunduTan} constructed (one-time symmetric key) quantum encryption with certified deletion with the device-independent security, i.e., the security holds even if quantum devices are untrusted. Moreover, they show that their construction satisfies composable security.

The notion of NTCF functions was first introduced by Brakerski et al.~\cite{FOCS:BCMVV18}, and further extended to construct a classical verification of quantum computing by Mahadev~\cite{FOCS:Mahadev18a}. (See also a related primitive so-called QFactory~\cite{AC:CCKW19}.)
The adaptive hardcore bit property of NTCF functions was also used for semi-quantum money~\cite{CoRR:RadSat19} and secure software leasing with classical communication~\cite{EPRINT:KitNisYam20}.

Ananth and Kaleoglu concurrently and independently present reusable secret key and public key uncloneable encryption schemes~\cite{EPRINT:AnaKal21}.
Uncloneable encryption~\cite{TQC:BroLor20} is related to but different from quatum encryption with certified deletion. Uncloneable encryption prevents adversaries from creating multiple ciphertexts whose plaintext is the same as that of the original ciphertext.
Their constructions are based on a similar idea to one of our main ideas. Specifically, their construction is obtained by combining one-time secret key uncloneable encryption and standard SKE/PKE with the ``fake-key property", which is similar to the RNC security.


\subsection{Technical Overview Part I: Quantum Communication Case}\label{sec:technical_overview_quantum}
We provide an overview of how to achieve PKE and ABE with certified deletion using quantum communication in this section.
To explain our idea, we introduce the definition of PKE with certified deletion.
\paragraph{Definition of quantum encryption with certified deletion.}
A PKE with certified deletion consists of the following algorithms.
\begin{description}
\item [$\keygen(1^\secp) \ra (\pk,\sk)$:] This is a key generation algorithm that generates a pair of public and secret keys.
\item[$\Enc(\pk,m)\ra (\vk,\CT)$:] This is an encryption algorithm that generates a ciphertext of plaintext and a verification key for this ciphertext.
\item[$\Dec(\sk,\CT) \ra m^\prime$:] This is a decryption algorithm that decrypts a ciphertext.
\item[$\Delete(\CT)\ra \cert$:] This is a deletion algorithm that generates a certificate to guarantee that the ciphertext $\CT$ was deleted.
\item[$\Vrfy(\vk,\cert)\ra \top$ or $\bot$:] This is a verification algorithm that checks the validity of a certificate $\cert$ by using a verification key. As correctness, we require that this algorithm returns $\top$ (i.e., it accepts) if $\cert$ was honestly generated by $\Delete(\CT)$ and $(\vk,\CT)$ was honestly generated by $\Enc$.
\end{description}
Roughly speaking, certified deletion security requires that no quantum polynomial time (QPT) adversary given $\pk$ and $\CT$ can obtain any information about the plaintext in $\CT$ \emph{even if $\sk$ is given after a valid certificate $\cert \lrun \Delete(\CT)$ is generated}.
The difference between PKE and reusable SKE with certified deletion is that, in reusable SKE, $\keygen$ outputs only $\sk$.
In the one-time SKE case by Broadbent and Islam~\cite{TCC:BroIsl20}, $\Enc$ does not output $\vk$ and $\Vrfy$ uses $\sk$ instead of $\vk$.

\paragraph{Our idea for PKE.}
We use the construction of one-time SKE with certified deletion by Broadbent and Islam~\cite{TCC:BroIsl20}. However, we do not need to know the detail of the SKE scheme since we use it in a black-box way in our PKE scheme.
What we need to understand about the SKE scheme are the following abstracted properties:
(1) A secret key and a plaintext are classical strings.
(2) A ciphertext is a quantum state.
(3) The encryption algorithm does not output a verification key since the verification key is equal to the secret key.
(4) It satisfies the verification correctness and certified deletion security explained above.

Our idea is to convert the SKE with certified deletion scheme into a PKE with certified deletion scheme by combining with a standard PKE scheme (standard hybrid encryption technique). This conversion is possible since a secret key of the SKE scheme is a classical string.
Let $\PKE.(\keygen,\Enc,\Dec)$ and $\SKE.(\keygen,\Enc,\Dec,\Delete,\Vrfy)$ be normal PKE and one-time SKE with certified deletion schemes, respectively. Our PKE with certified deletion scheme is described as follows.
\begin{description}
\item [$\keygen(1^\secp)$:] This outputs $(\pke.\pk,\pke.\sk)\lrun \PKE.\keygen(1^\secp)$.
\item[$\Enc(\pk,m)$:] This generates $\ske.\sk \lrun \SKE.\keygen(1^\secp)$, $\ske.\CT \lrun \SKE.\Enc(\ske.\sk,m)$,  
and $\pke.\CT\lrun \PKE.\Enc(\pke.\pk,\allowbreak \ske.\sk)$, and outputs $\vk \seteq \ske.\sk$ and $\CT \seteq (\ske.\CT,\pke.\CT)$.
\item[$\Dec(\sk,\CT)$:] This computes $\ske.\sk^\prime \lrun \PKE.\Dec(\pke.\sk,\pke.\CT)$ and $m^\prime \lrun \SKE.\Dec(\ske.\sk^\prime,\ske.\CT)$, and outputs $m^\prime$.
\item[$\Delete(\CT)$:] This generates and outputs $\cert \lrun \SKE.\Delete(\ske.\CT)$.
\item[$\Vrfy(\vk,\cert)$:] This outputs the output of $\SKE.\Vrfy(\ske.\sk,\cert)$ (note that $\vk =\ske.\sk$).
\end{description}
At first glance, this naive idea seems to work since even if $\pke.\sk$ is given to an adversary after a valid $\cert$ is generated, $\ske.\CT$ does not leak information about the plaintext by certified deletion security of the SKE scheme. Note that PKE is used to encrypt $\ske.\sk$ (not $m$). One-time SKE is sufficient since $\ske.\sk$ is freshly generated in $\Enc$.
The proof outline is as follows. First, we use IND-CPA security of normal PKE to erase information about $\ske.\sk$. Then, we use the one-time certified deletion security of $\SKE$.
Unfortunately, we do not know how to prove the first step above because we must give $\pke.\sk$ to an adversary in a security reduction.
In the first step, we need to show that if a distinguisher detects that $\PKE.\Enc(\pke.\pk, \ske.\sk)$ is changed to $\PKE.\Enc(\pke.\pk,0^{\abs{\ske.\sk}})$, we can break IND-CPA security of the normal PKE. However, to run the distinguisher, we need to give $\pke.\sk$ to the distinguisher after it sends a valid certificate for deletion. The reduction has no way to give $\pke.\sk$ to the distinguisher since the reduction is trying to break the PKE scheme!

To solve this problem, we use RNC encryption (RNCE)~\cite{EC:JarLys00,TCC:CanHalKat05}. RNCE consists of algorithms $(\keygen,\Enc,\Dec,\Fake,\Reveal)$. The key generation algorithm outputs not only a key pair $(\pk,\sk)$ but also an auxiliary trapdoor information $\aux$. The fake ciphertext generation algorithm $\Fake(\pk,\sk,\aux)$ can generate a fake ciphertext $\tlct$ that does not include information about a plaintext. The reveal algorithm $\Reveal(\pk,\sk,\aux,\tlct,m)$ can generate a fake secret key that decrypts $\tlct$ to $m$. The RNC security notion requires that $(\tlct=\Fake(\pk,\sk,\aux),\Reveal(\pk,\sk,\aux,\tlct,m))$ is computationally indistinguishable from $(\Enc(\pk,m),\sk)$.

RNCE perfectly fits the scenario of certified deletion. We use an RNCE scheme $\NCE.(\keygen,\Enc,\Dec,\Fake,\Reveal)$ instead of a normal PKE in the PKE with certified deletion scheme above.
To erase $\ske.\sk$, we use the RNC security. We change $\NCE.\Enc(\nce.\pk,\ske.\sk)$ and $\nce.\sk$ into $\nce.\tlct =\NCE.\Fake(\nce.\pk,\nce.\sk,\nce.\aux)$ and $\NCE.\Reveal(\nce.\pk,\nce.\sk,\nce.\aux,\nce.\tlct,\ske.\sk)$, respectively. Thus, as long as $\ske.\sk$ is given after a valid certification is generated, we can simulate the secret key of the PKE with certified deletion scheme. Using RNCE
solves the problem above since the reduction obtains both a target ciphertext and a secret key (real or fake) in the RNC security game. To complete the security proof, we use the certified deletion security of $\SKE$. Here, the point is that the reduction can simulate a secret key by $\Reveal$ since the reduction is given $\ske.\sk$ after a valid certificate is sent in the certified deletion security game.

If we use secret key RNCE instead of public key RNCE, we can achieve reusable SKE with certified deletion via the design idea above.
Secret/public key RNCE
can be constructed from IND-CPA SKE/PKE, respectively~\cite{TCC:CanHalKat05,C:KNTY19}, and SKE with certified deletion exists unconditionally~\cite{TCC:BroIsl20}. Thus, we can achieve PKE (resp. reusable SKE) with certified deletion from IND-CPA PKE (resp. OWFs).

Note that the RNCE technique above is the fundamental technique in this work.\footnote{Ananth and Kaleoglu concurrently and independently present essentially the same technique in the uncloneable encryption setting~\cite{EPRINT:AnaKal21}.} We use this technique both in the quantum communication case and in the classical communication case.

\paragraph{Our idea for ABE.}
We can extend the idea for PKE to the ABE setting. In this work, we focus on key-policy ABE, where a policy (resp. attribute) is embedded in a secret key (resp. ciphertext). The crucial tool is (receiver) non-committing ABE (NCABE), which we introduce in this work.

Although the definition of NCABE is basically a natural extension of that of RNCE,
we describe algorithms of NCABE for clarity. It helps readers who are not familiar with normal ABE. The first four algorithms below are algorithms of normal ABE.
\begin{description}
\item[$\Setup(1^\secp)\ra (\pk,\msk)$:] This is a setup algorithm that generates a public key and a master secret key.
\item[$\keygen(\msk,P)\ra \sk_P$:] This is a key generation algorithm that generates a secret key for a policy $P$.
\item[$\Enc(\pk,X,m)\ra \CT_X$:] This is an encryption algorithm that generates a ciphertext of $m$ under an attribute $X$.
\item[$\Dec(\sk_P,\CT_X)\ra m^\prime$ or $\bot$:] This is a decryption algorithm that decrypts $\CT_X$ if $P(X)=\top$. If $P(X)=\bot$, it outputs $\bot$.
\item[$\FakeSetup(1^\secp)\ra (\pk,\aux)$:] This is a fake setup algorithm that generates a public key and a trapdoor auxiliary information $\aux$.   
\item[$\FakeCT(\pk,\aux,X)\ra \tlct_X$:] This is a fake ciphertext generation algorithm that generates a fake ciphertext $\tlct_X$ under an attribute $X$. 
\item[$\FakeSK(\pk,\aux,P)\ra \tlsk_P$:] This is a fake key generation algorithm that generates a fake secret key $\tlsk_P$ for $P$.
\item[$\Reveal(\pk,\aux,\tlct,m)\ra \tlmsk$:] This is a reveal algorithm that generates a fake master secret key $\tlmsk$.
\end{description}
Roughly speaking, the NCABE security notion requires that the fake public key, master secret key, ciphertext, and secret keys are computationally indistinguishable from the normal public key, master key, ciphertext, and secret keys.
It is easy to see that the hybrid encryption approach works in the ABE setting as well. Thus, the goal is achieving an NCABE scheme.

Our NCABE construction follows the RNCE construction based on IND-CPA PKE~\cite{TCC:CanHalKat05,C:KNTY19}. However, the crucial difference between the PKE and ABE settings is that, in the ABE setting, adversaries are given many secret keys for queried policies (that is, we consider collusion-resistance). There is an obstacle to achieving collusion resistance because secret keys for policies depend on a master secret key.
Note that adversaries can send secret key queries \emph{both before and after} the target ciphertext is given.

First, we explain the RNCE scheme from PKE. Although we explain the $1$-bit plaintext case, it is easy to extend to the multi-bit case. The idea is the simple double encryption technique by Naor and Yung~\cite{STOC:NaoYun90}, but we do not need non-interactive zero-knowledge (NIZK). We generate two key pairs $(\pk_0,\sk_0)$ and $(\pk_1,\sk_1)$ and set $\pk\seteq (\pk_0,\pk_1)$, $\sk\seteq \sk_z$, and $\aux=(\sk_0,\sk_1,z^\ast)$ where $z,z^\ast\chosen\zo{}$. A ciphertext consists of $\Enc(\pk_0,b)$ and $\Enc(\pk_1,b)$. We can decrypt the ciphertext by using $\sk_z$.
A fake ciphertext $\tlct$ is $(\Enc(\pk_{z^\ast},0),\Enc(\pk_{1-z^\ast},1))$.
To generate a fake secret key for a plaintext $m^\ast$, the reveal algorithm outputs $\sk_{z^\ast \xor m^\ast}$. It is easy to see decrypting $\tlct$ by $\sk_{z^\ast \xor m^\ast}$ yields $m^\ast$.

Our NCABE is based on the idea above. That is, we use two key pairs $(\pk_0,\msk_0)$ and $(\pk_1,\msk_1)$ of a normal ABE scheme $\ABE.(\Setup,\keygen,\Enc,\Dec)$, and a ciphertext consists of $(\ABE.\Enc(\pk_0,X,b),\ABE.\Enc(\pk_1,X,b))$ where $X$ is an attribute. Our reveal algorithm outputs $\msk_{z^\ast \xor m^\ast}$ for a plaintext $m^\ast$ as in the PKE case. The problem is a secret key for a policy $P$. A naive idea is that a key generation algorithm outputs $\sk_P \lrun \ABE.\keygen(\msk_z,P)$ where $z \chosen \zo{}$ is chosen in the setup algorithm, and a fake key generation algorithm outputs $\tlsk_P \lrun \ABE.\keygen(\msk_{z^\ast \xor m^\ast},P)$. However, this apparently does not work since $\tlsk_P$ depends on $m^\ast$. Unless $\tlsk_P$ is independent of $m^\ast$, we cannot use NCABE to achieve ABE with certified deletion because $\ske.\sk$ of SKE with certified deletion is sent \emph{after} a valid certification is generated ($\ske.\sk$ would be a plaintext of ABE in the hybrid encryption). To make a fake key generation be independent of $m^\ast$, we need to hide which master secret key is used to generate a secret key for $P$. If a secret key leaks information about which secret key (extracted from $\msk_0$ or $\msk_1$) is used, we cannot adaptively select a fake master secret key in the reveal algorithm.

IO helps us to overcome this hurdle. Our idea is as follows. A key generation algorithm outputs an obfuscated circuit of a circuit $\sfD[\sk_z]$ that takes a ciphertext $(\abe.\CT_0,\abe.\CT_1)\seteq (\ABE.\Enc(\pk_0,X,b),\ABE.\Enc(\pk_1,X,b))$ and outputs $\ABE.\Dec(\sk_z,\abe.\CT_z)$ where $z\chosen \zo{}$ and $\sk_z \lrun \ABE.\keygen(\msk_z,P)$ is hard-coded in $\sfD$. A fake key generation algorithm outputs an obfuscated circuit of a circuit $\sfD_0[\sk_0]$ that takes $(\abe.\CT_0,\abe.\CT_1)$ and outputs $\ABE.\Dec(\sk_0,\abe.\CT_0)$ where $\sk_0 \lrun \ABE.\keygen(\msk_0,P)$ is hard-coded in $\sfD_0$. Note that the fake secret key cannot be used to decrypt a fake ciphertext $(\abe.\CT_{z^\ast},\abe.\CT_{1-z^\ast})\seteq (\ABE.\Enc(\pk_{z^\ast},X,0),\ABE.\Enc(\pk_{1-z^\ast},X,1))$ where $z^\ast \chosen \zo{}$ since $P(X)=\bot$ must hold by the requirement on ABE security. Since the decryption circuits $\sfD$ and $\sfD_0$ are obfuscated, adversaries have no idea about which secret key ($\sk_0$ or $\sk_1$) is used for decryption. This idea is inspired by the functional encryption (FE) scheme by Garg et al.~\cite{SICOMP:GGHRSW16}.

The final issue is that adversaries can detect whether a secret key is real or fake if they use an invalid ciphertext $(\ABE.\Enc(\pk_0,b),\ABE.\Enc(\pk_1,1-b))$ as an input to the obfuscated circuits. To prevent this attack, we use statistically sound NIZK to check the consistency of double encryption as the FE scheme by Garg et al.~\cite{SICOMP:GGHRSW16}. By the statistical soundness of NIZK, we can guarantee that the obfuscated decryption circuit does not accept invalid ciphertexts, and $\sfD$ and $\sfD_0$ are functionally equivalent. Note that a secret key for policy $P$ outputs $\bot$ for the target ciphertext since a target attribute $X^\ast$ in the target ciphertext satisfies $P(X)=\bot$. We do not need the simulation-soundness, unlike the FE scheme by Garg et al. due to the following reason. In the FE scheme, plain PKE schemes are used for the double encryption technique and a secret key $\sk_0$ or $\sk_1$ is hard-coded in a functional decryption key. Before we use PKE security under $\pk_b$, we need to switch decryption from by $\sk_b$ to $\sk_{1-b}$ by IO security. During this phase, we need to use a fake simulated proof of NIZK. Thus, the simulation-soundness is required. However, in our ABE setting, a secret key for $P$ (not the master secret keys $\msk_0,\msk_1$) is hard-coded in $\sfD$ (or $\sfD_0$) above. Thanks to the ABE key oracle, $\sk_0$ and $\sk_1$ for $P$ are always available in reductions. We can first use IO security to switch from $\sfD$ to $\sfD_0$. After that, we change a real NIZK proof into a fake one. Thus, our NCABE scheme does not need the simulation-soundness. 
This observation enables us to achieve the adaptive security rather than the selective security unlike the FE scheme by Garg et al.\footnote{In the initial version of this work~\cite{EPRINT:NisYam21}, we achieve only the selective security because we use statistical simulation-sound NIZK as the FE scheme by Garg et al.~\cite{SICOMP:GGHRSW16}. We improve the result.} 
See~\cref{sec:NCABE_from_IO} for the detail. Thus, we can achieve NCABE from IO and OWFs since adaptively secure standard ABE can be constructed from IO and OWFs.


\subsection{Technical Overview Part II: Classical Communication Case}\label{sec:technical_overview_classical}
We provide an overview of how to achieve 
privately verifiable and publicly verifiable
PKE with certified deletion using classical communication in this section.
We note that both of them rely on interactive encryption algorithms. 

\paragraph{Privately verifiable construction.}
For realizing a privately verifiable construction with classical communication, we rely on \emph{NTCF functions}~\cite{FOCS:BCMVV18,FOCS:Mahadev18a}.
In this overview, we consider an ideal version, noise-free claw-free permutations for simplicity. 
A trapdoor claw-free permutation is $f:\bit \times \bit^w \rightarrow \bit^w$ such that 
\revise{
(1) $f(0,\cdot)$ and $f(1,\cdot)$ are permutations over $\bit^w$, (2) given the description of $f$, it is hard to find $x_0$ and $x_1$ such that $f(0,x_0)=f(1,x_1)$, and (3) there is a trapdoor $\td$ that enables one to efficiently find $x_0$ and $x_1$ such that $f(0,x_0)=f(1,x_1)=y$ for any $y$. 
}
In addition, the existing work showed that (a noisy version of) it satisfies a property called the \emph{adaptive hardcore bit property} under the LWE assumption 
~\cite{FOCS:BCMVV18}. 
To explain this, suppose that one 
generates the state $\sum_{b,x}\ket{b}|x\rangle|f(b,x)\rangle$, and measures the third register in the computational basis to get a result $y$.
Then the first and second registers collapse to the state $\frac{1}{\sqrt{2}}\left(\ket{0}\ket{x_0}+\ket{1}\ket{x_1}\right)$ with $f(0,x_0)=f(1,x_1)=y$.
If one measures the state in the computational basis, the measurement outcome is $(0,x_0)$ or $(1,x_1)$.
If, on the other hand, one measures the state in the Hadamard basis, the measurement outcome is $(e,d)$ 
such that $e=d\cdot(x_0\oplus x_1)$.
The adaptive hardcore bit property roughly means that
once one gets $(0,x_0)$ or $(1,x_1)$, it cannot output $(e,d)$ such that 
$d\neq 0$ and 
$e=d\cdot(x_0\oplus x_1)$ with probability better than $1/2+\negl(\secp)$.
Note that this is a tight bound since $e=d\cdot(x_0\oplus x_1)$ holds with probability $1/2$ if we randomly choose $e$.    
Existing works showed that this property can be amplified by parallel repetition~\cite{CoRR:RadSat19,EPRINT:KitNisYam20}: 
Specifically, let 
$(0,x_{i,0})$ and $(1,x_{i,1})$ be the preimages of $y_i$ under $f_i$ for $i\in[n]$  where  $n=\omega(\log \secp)$. 
Then once one gets a sequence $\{b_i,x_{i,b_i}\}_{i\in[n]}$ for some $b_1\concat...\concat b_n \in \bit^n$,   
it can get a sequence $\{e_i,d_i\}_{i\in [n]}$ such that 
$d_i\neq 0$ and 
$e_i=d_i\cdot(x_{i,0}\oplus x_{i,1})$ only with negligible probability.

We use this property to construct an encryption scheme with certified deletion. A natural idea would be as follows:
The sender sends $n$ functions $\{f_i\}_{i\in [n]}$ to the receiver,   the receiver generates  $\{y_i\}_{i\in [n]}$ along with states $\{\frac{1}{\sqrt{2}}\left(\ket{0}\ket{x_{i,0}}+\ket{1}\ket{x_{i,1}}\right)\}_{i\in [n]}$ as above and sends $\{y_i\}_{i\in [n]}$ to the sender, and 
the sender sends receiver a ciphertext $\ct$  decryptable only when $\{b_i,x_{i,b_i}\}_{i\in[n]}$ for some $b_1\concat...\concat b_n \in \bit^n$ is available.
We discuss how to implement such a ciphertext later. 
We use $\{e_i,d_i\}_{i\in [n]}$ such that 
$e_i=d_i\cdot (x_{i,0}\oplus x_{i,1})$ as a deletion certificate. 
The receiver can decrypt the ciphertext  by measuring the states in the computational basis, and once it outputs a valid deletion
certificate, it must ``forget" preimages by the amplified adaptive hardcore property and thus cannot decrypt the ciphertext.   
This idea can be implemented by a straightforward manner if we generate $\ct$ by (extractable) witness encryption~\cite{STOC:GGSW13,C:GKPVZ13} under the corresponding $\NP$ language. 
However, since witness encryption is a strong assumption, we want to avoid this. 
Indeed, we can find the following candidate construction using a hash function $H$ modeled as a random oracle. 
We set the ciphertext as $\ct\seteq \{\ct_{i,b}\}_{i\in[n],b\in\bit}$ 
where $\{m_i\}_{i\in[n]}$ is an $n$-out-of-$n$ secret sharing  of the message $m$ and 
$\ct_{i,b}\seteq m_i\oplus H(b\concat x_{i,b})$. 
The intuition is that an adversary has to get  $m_i$ for all $i\in[n]$ to get $m$  and it has to know $(0,x_{i,0})$ or $(1,x_{i,1})$ to know $m_i$.
Therefore, it seems that any adversary that gets any information of $m$ can be used to extract a sequence $\{b_i,x_{i,b_i}\}_{i\in[n]}$ for some $b_1\concat...\concat b_n \in \bit^n$.
If this is shown, then it is straightforward to prove that the adversary can get no information of $m$ once it submits a valid deletion certificate by the amplified adaptive hardcore property as explained above. 
However, turning this intuition into a formal proof seems difficult. 
A common technique to extract information from adversary's  random oracle queries is the one-way to hiding lemma~\cite{JACM:Unruh15,C:AmbHamUnr19}, which roughly claims that if the adversary distinguishes $H(X)$ from random, then we would get $X$ with non-negligible probability by measuring a randomly chosen query. 
Here, a problem is that we have to extract $n$ strings $\{b_i,x_{i,b_i}\}_{i\in[n]}$ simultaneously.
On the other hand, the extraction by the one-way to hiding lemma disturbs adversary's state by a  measurement, and thus we cannot use this technique sequentially.\footnote{A recent work by Coladangelo, Majenz, and Poremba \cite{CMP20} studied what is called ``simultaneous one-way to hiding lemma", but their setting is different from ours and their lemma cannot be used in our setting.}

The difficulty above comes from the fact that the sender cannot know which of $(0,x_{i,0})$ and $(1,x_{i,1})$ the receiver will get, and thus it has to send a ciphertext that can be decrypted in either case.
To resolve this issue, we rely on the injective invariance, which roughly says that there is an injective function $g$ that is computationally indistinguishable from $f$~\cite{FOCS:Mahadev18a}.  
First, suppose that we just use $g$ instead of $f$ in the above idea.  
Since $g$ is injective, there is a unique preimage $(b_i,x_i)$ of $y_i$, in which case the sender knows that the receiver will get $\{(b_i,x_i)\}_{i\in[n]}$ by the standard basis measurement.
In this case, the aforementioned problem can be easily resolved by setting $\ct\seteq m\oplus H(b_1\concat x_1\concat...\concat b_n\concat x_n)$ as the ciphertext.
In this case, it is easy to prove that 
we can extract $\{b_i,x_i\}_{i\in[n]}$ if an adversary obtains some information of $m$
by applying the standard one-way to hiding lemma. 
However, an obvious problem is that  the deletion certificate no longer works for $g$ since 
the receiver's state collapses to a classical state after the measurement of $\{y_i\}_{i\in [n]}$ and thus 
the Hadamard basis measurement results in just uniform bits.

Our idea is to take advantages of both of them. 
Specifically, the sender sends  functions $\{\eta_i\}_{i\in[n]}$, where
$\eta_i$ is the $g$-type function for $i\in S$ and it is the $f$-type function for $i\in [n]\setminus S$ with a certain set $S\subset[n]$.
The receiver generates a set of states each of which is a superposition of two preimages of a $f$-type function or
a state encoding the unique preimage of a $g$-type function.
The preimages of $g$-type functions are used for encryption/decryption, 
and 
the Hadamard measurement results  are used for deletion certificate, whose validity is only checked on positions where $f$-type functions are used. 
We also include an encryption of the description of the subset $S$ in the ciphertext so that a legitimate receiver can know which position should be used in the decryption. 
More precisely,  we set $\ct\seteq (
\Enc(S), 
m\oplus H(\{b_i,x_i\}_{i\in [S]}))$
where $\Enc$ is a PKE scheme with the RNC security.\footnote{We require $\Enc$ to satisfy the RNC security due to a similar reason to that in~\cref{sec:technical_overview_quantum}, which we omit to explain here.}\footnote{In the actual construction, there is an additional component that is needed for preventing an adversary from decrypting the ciphertext \emph{before} outputting a valid deletion certificate without the decryption key.
This is just a security as standard PKE and can be added easily.
Thus, we omit this and focus on the security \emph{after} outputting a valid deletion certificate.
}   A deletion certificate $\{e_i,d_i\}_{i\in[n]}$ is valid if we have 
$d_i\neq 0$ and 
$e_i=d_i\cdot (x_{i,0}\oplus x_{i,1})$ for all $i\in [n]\setminus S$.
For the security proof of this construction, 
the amplified adaptive hardcore property cannot be directly used, 
because it is a property about $f$-type functions whereas the above construction mixes
$f$-type functions and $g$-type functions, and what we want to have is the mutually-exclusive property between
preimages of $g$-type functions and
deletion certificates of $f$-type functions.
To solve the problem, we introduce a new property  which we call 
{\it the cut-and-choose adaptive hardcore property} (\cref{lem:cut_and_choose_adaptive_hardcore}).
The cut-and-choose adaptive hardcore property intuitively means that once the receiver issues a deletion certificate $\{e_i,d_i\}_{i\in [n]}$ that is valid for all $i\in [n]\setminus S$ before knowing $S$, 
it can no longer generate correct preimages $\{b_i,x_i\}_{i\in [S]}$ even if it receives $S$ later. 
Intuitively, this holds because the only way to obtain such $\{e_i,d_i\}_{i\in [n]}$ before knowing $S$ would be to measure the states in the Hadamard basis for all $i\in[n]$, in which case the receiver should forget all preimages. 
We show that the cut-and-choose adaptive hardcore property can be reduced to the adaptive hardcore
bit property and injective invariance.
The new property we show itself is of independent interest, and we believe it will be useful in many other applications
of quantum cryptography.

Because the only known construction of NTCF functions~\cite{FOCS:BCMVV18,FOCS:Mahadev18a}
assumes the LWE assumption, our construction of the PKE with privately verifiable certified deletion with classical communication
is also  based on the LWE assumption and our security proof is done in the QROM.
We note that the construction only achieves private verification because verification of  deletion certificates requires both of two preimages of $f$-type functions, which cannot be made public.

\paragraph{Publicly verifiable construction.}
The above construction is not publicly verifiable because the verification of the validity of $(e_i,d_i)$  requires both preimages $x_{i,0}$ and $x_{i,1}$, which cannot be made public.
One might notice that the validity check of the preimage  can be done publicly,
and might suggest the following construction: preimages  are used for deletion certificate, and Hadamard measurement outcomes $\{e_i,d_i\}_{i\in [n]}$ are used as the decryption key of the encryption. Because a valid $\{e_i,d_i\}_{i\in [n]}$  is a witness of an $\NP$ statement,
we could use  (extractable) witness encryption~\cite{STOC:GGSW13,C:GKPVZ13} to ensure that a receiver can decrypt the message only if it knows a valid $\{e_i,d_i\}_{i\in [n]}$.
This idea, however, does not work, because the statement of the witness encryption contains private information (i.e., preimages), and witness encryption ensures nothing about privacy of the statement under which a message is encrypted.

Our idea to solve the problem is 
to use the one-shot signature~\cite{STOC:AGKZ20}.
Roughly speaking, one-shot signature (with a message space $\bit$) enables one to generate a classical public key $\pk$ along with a quantum secret key $\sk$, which can be used to generate either of  a signature $\sigma_0$ for message $0$ or $\sigma_1$ for message $1$, but not both. 
We note that a signature can be verified publicly. 

We combine one-shot signatures with  extractable witness encryption.\footnote{We note that a combination of one-shot signatures and extractable witness encryption appeared in the work of Georgiou and Zhandry \cite{EPRINT:GeoZha20} in a related but different context.} The encryption $\Enc(m)$ of a message $m$ 
in our construction is a ciphertext of witness encryption of message
$m$ under the statement corresponding to the verification of one-shot signature for  message $0$. 
The deletion certificate is, on the other hand, a one-shot signature for message $1$.
Once a valid signature of $1$ is issued, a valid signature of $0$, which is a decryption key of our witness encryption, 
is no longer possible to generate due to the security of the one-shot signature. 
This intuitively ensures the certified deletion security of our construction.
Because signatures are publicly verifiable,
the verification of our construction is also publicly verifiable.  
In the actual construction, in order to prevent an adversary from decrypting the ciphertext before issuing the deletion certificate, we add an additional layer of encryption, for which we use RNCE   due to the similar reason to that in \cref{sec:technical_overview_quantum}. 

Unfortunately, the only known construction of the one-shot signature needs classical oracles. It is an open question whether
we can construct a PKE with publicly verifiable certified deletion with classical communication
based on only standard assumptions such as the LWE assumption.

\section{Preliminaries}\label{sec:prelim}

\subsection{Notations and Mathematical Tools}\label{sec:notation}
We introduce basic notations and mathematical tools used in this paper.

In this paper,  $x \chosen X$ denotes selecting an element from a finite set $X$ uniformly at random, and $y \gets \algo{A}(x)$ denotes assigning to $y$ the output of a probabilistic or deterministic algorithm $\algo{A}$ on an input $x$. When we explicitly show that $\algo{A}$ uses randomness $r$, we write $y \gets \algo{A}(x;r)$. When $D$ is a distribution, $x \chosen D$ denotes sampling an element from $D$. Let $[\ell]$ denote the set of integers $\{1, \cdots, \ell \}$, $\secp$ denote a security parameter, and $y \seteq z$ denote that $y$ is set, defined, or substituted by $z$. For a string $s \in \zo{\ell}$, $s[i]$ denotes $i$-th bit of $s$.
QPT stands for quantum polynomial time. 
PPT stands for (classical) probabilistic polynomial time.
For a subset $S\subseteq W$ of a set $W$, $\overline{S}$ is the complement of $S$, i.e., $\overline{S}:=W\setminus S$.

A function $f: \N \ra \R$ is a negligible function if for any constant $c$, there exists $\secp_0 \in \N$ such that for any $\secp>\secp_0$, $f(\secp) < \secp^{-c}$. We write $f(\secp) \leq \negl(\secp)$ to denote $f(\secp)$ being a negligible function.
A function $g: \N \ra \R$ is a noticeable function if there exist constants $c$ and $\secp_0$ such that for any $\secp \ge \secp_0$, $g(\secp) \ge \secp^{-c}$.
The trace distance between two states $\rho$ and $\sigma$ is given by $\parallel \rho-\sigma 
\parallel_{tr}$, where 
$\parallel A\parallel_{tr}\seteq \mathrm{Tr}\sqrt{A^{\dagger}A}$
is the trace norm.
We call a function $f$ a density on $X$ if $f:X\rightarrow[0,1]$ such that $\sum_{x\in X}f(x)=1$.
For two densities $f_0$ and $f_1$ over the same finite domain $X$, the Hellinger distance between $f_0$ and $f_1$ is ${\bf H}^2(f_0,f_1)\seteq 1-\sum_{x\in X}\sqrt{f_0(x)f_1(x)}$.

\subsection{Cryptographic Tools}\label{sec:basic_crypto}

In this section, we review cryptographic tools used in this paper.

\paragraph{Public key encryption.}

\begin{definition}[Public Key Encryption (Syntax)]\label{definition:public key}
A public key encryption scheme $\Sigma=(\keygen,\Enc,\Dec)$ is a triple of PPT algorithms, a key generation algorithm $\keygen$,
an encryption algorithm $\Enc$ and a decryption algorithm $\Dec$, with plaintext space $\cM$.
\begin{description}
    \item[$\keygen(1^\secp)\ra (\pk,\sk)$:] The key generation algorithm takes as input the security parameter $1^\secp$ and outputs a public key $\pk$ and a secret key $\sk$.
    \item[$\Enc(\pk,m) \ra \ct$:] The encryption algorithm takes as input $\pk$ and a plaintext $m \in \Ms$, and outputs a ciphertext $\ct$.
    \item[$\Dec(\sk,\ct) \ra m^\prime \mbox{ or } \bot$:] The decryption algorithm takes as input $\sk$ and $\ct$, and outputs a plaintext $m^\prime$ or $\bot$.
\end{description}
\end{definition}

\begin{definition}[Correctness for PKE]\label{def:pke_correctness}
For any $\secp\in \N$, $m\in\Ms$,
\begin{align}
\Pr\left[
\Dec(\sk,\ct)\ne m
\ \middle |
\begin{array}{ll}
(\pk,\sk)\lrun \keygen(1^\secp)\\
\ct \lrun \Enc(\pk,m)
\end{array}
\right] 
\le\negl(\secp).
\end{align}
\end{definition}

\begin{definition}[OW-CPA security]\label{definition:OW-CPA}
Let $\Sigma =(\keygen,\Enc,\Dec)$ be a PKE scheme. For QPT adversaries $\cA$, we define the following security experiment $\expa{\Sigma,\cA}{ow}{cpa}(\secp)$.
\begin{enumerate}
    \item The challenger generates $(\pk,\sk)\lrun \keygen(1^{\secp})$,
    chooses $m\lrun \Ms$, computes $\ct\lrun \Enc(\pk,m)$,  
    and sends $(\pk,\ct)$ to $\cA$.
    \item $\cA$ outputs $m'$. The experiment outputs $1$ if $m'=m$ and otherwise $0$.
\end{enumerate}
We say that the $\Sigma$ is OW-CPA  secure if for any QPT $\cA$, it holds that
\begin{align}
\advb{\Sigma,\cA}{ow}{cpa}(\secp) \seteq \Pr[\expa{\Sigma,\cA}{ow}{cpa}(\secp)=1]\leq \negl(\secp).
\end{align}
Note that we assume $1/\abs{\Ms}$ is negligible.
\end{definition}

\begin{definition}[IND-CPA security]\label{definition:IND-CPA}
Let $\Sigma =(\keygen,\Enc,\Dec)$ be a PKE scheme. For QPT adversaries $\cA$, we define the following security experiment $\expa{\Sigma,\cA}{ind}{cpa}(\secp,b)$.

\begin{enumerate}
    \item The challenger generates $(\pk,\sk)\lrun \keygen(1^{\secp})$, and sends $\pk$ to $\cA$.
    \item $\cA$ sends $(m_0,m_1)\in\cM^2$ to the challenger.
    \item The challenger computes $\ct_b \lrun \Enc(\pk,m_b)$, and sends $\ct_b$ to $\cA$.
    \item $\cA$ outputs $b'\in\bit$. This is the output of the experiment.
\end{enumerate}
We say that the $\Sigma$ is IND-CPA secure if for any QPT $\cA$, it holds that
\begin{align}
\advb{\Sigma,\cA}{ind}{cpa}(\secp) \seteq \abs{\Pr[\expa{\Sigma,\cA}{ind}{cpa}(\secp,0)=1]  - \Pr[\expa{\Sigma,\cA}{ind}{cpa}(\secp,1)=1]} \leq \negl(\secp).
\end{align}
\end{definition}

It is well-known that the IND-CPA security implies the OW-CPA security. 
There are many IND-CPA secure PKE schemes against QPT adversaries under standard cryptographic assumptions.
A famous one is Regev PKE scheme, which is IND-CPA secure if the learning with errors (LWE) assumption holds against QPT adversaries~\cite{JACM:Regev09}. See the references for the LWE assumption and constructions of post-quantum secure PKE~\cite{JACM:Regev09,STOC:GenPeiVai08}.

\begin{definition}[Indistinguishability Obfuscator~\cite{JACM:BGIRSVY12}]\label{def:io}
A PPT algorithm $\iO$ is an IO for a circuit class $\{\cC_\secp\}_{\secp \in \N}$ if it satisfies the following two conditions.

\begin{description}
\item[Functionality:] For any security parameter $\secp \in \N$, circuit $C \in \cC_\secp$, and input $x$, we have that
\begin{align}
\Pr[C'(x)=C(x) \mid C' \lrun \iO(C)] = 1\enspace.
\end{align}

\item[Indistinguishability:] For any QPT distinguisher $\cD$ and
for any pair of circuits $C_0, C_1 \in \cC_\secp$ such that for any input $x$, $C_0(x) = C_1(x)$ and $\abs{C_0}=\abs{C_1}$,
it holds that
 \begin{align}
\adva{\iO,\cD}{io}(\secp) \seteq \abs{
 \Pr\left[\cD(\iO(C_0))= 1\right] -
 \Pr\left[\cD(\iO(C_1))= 1\right]
 } \leq \negl(\secp)\enspace.
 \end{align}
\end{description}
\end{definition}
There exist candidate constructions of IO against QPT adversaries~\cite{STOC:GayPas21,EC:WeeWic21,EPRINT:BDGM20b}.

\paragraph{Attribute-based encryption.}
We review the notion of (key-policy)  attribute-based encryption (ABE)~\cite{EC:SahWat05,CCS:GPSW06}.
\begin{definition}[Attribute-Based Encryption (Syntax)]\label{def:abe_syntax}
An ABE scheme is a tuple of PPT algorithms $(\Setup,\keygen,\Enc,\Dec)$ with plaintext space $\Ms$, attribute space $\Xs$, and
 policy space $\Ps$.
\begin{description}
    \item[$\Setup(1^\secp)\ra (\pk,\msk)$:] The setup algorithm takes as input the security parameter $1^\secp$ and outputs a public key $\pk$ and a master secret key $\msk$.
    \item[$\keygen (\msk,P) \ra \sk_P$:] The key generation algorithm takes as $\msk$ and a policy $P \in \Ps$, and outputs a secret key $\sk_P$.
    \item[$\Enc(\pk,X,m) \ra \ct_X$:] The encryption algorithm takes as input $\pk$, an attribute $X\in\Xs$, and a plaintext $m \in \Ms$, and outputs a ciphertext $\ct_X$.
    \item[$\Dec(\sk_P,\ct_X) \ra m^\prime \mbox{ or } \bot$:] The decryption algorithm takes as input $\sk_P$ and $\ct_X$, and outputs a plaintext $m^\prime$ or $\bot$.
\end{description}
\end{definition}

\begin{definition}[Perfect Correctness for ABE]\label{def:abe_correctness}
For any $\secp\in \N$, $m\in\Ms$, $P\in\Ps$, and $X\in\Xs$ such that $P(X)=\top$,
\begin{align}
\Pr\left[
\Dec(\sk_P,\ct_X)\ne m
\ \middle |
\begin{array}{ll}
(\pk,\msk)\lrun \Setup(1^\secp)\\
\sk_P \lrun \keygen(\msk,P)\\
\ct_X \lrun \Enc(\pk,X,m)
\end{array}
\right] 
=0.
\end{align}
\begin{remark}
Though we allow negligible decryption error for other primitives in this paper, we require perfect correctness for ABE.
This is needed for the security proof of the non-committing ABE scheme in~\cref{sec:NCABE_from_IO}.  
\end{remark}
\end{definition}

\begin{definition}[(Adaptive) IND-CPA Security for ABE]\label{def:abe_ind-cpa}
Let $\Sigma=(\Setup, \keygen, \Enc, \Dec)$ be an ABE scheme with plaintext space $\Ms$, attribute space $\Xs$, and policy space $\Ps$.
We consider the following security experiment $\expa{\Sigma,\cA}{ind}{cpa}(\secp,b)$.

\begin{enumerate}
    \item The challenger computes $(\pk,\msk) \lrun \Setup(1^\secp)$ and sends $\pk$ to $\cA$.
    \item $\cA$ sends a query $P\in \Ps$ to the challenger and it returns $\sk_P \lrun \keygen(\msk,P)$ to $\cA$. This process can be repeated polynomially many times.
    \item $\cA$ sends 
    $X^\ast \in \Xs$ and 
    $(m_0,m_1) \in \Ms^2$ to the challenger where $X^\ast$ must satisfy $P(X^\ast)=\bot$ for all key queries $P$ sent so far.
    \item The challenger computes $\ct_b \la \Enc(\pk,X^\ast,m_b)$ and sends $\ct_b$ to $\cA$.
    \item Again, $\cA$ can send key queries $P$ that must satisfy $P(X^\ast)=\bot$.
    \item Finally, $\cA$ outputs $b'\in \bit$.
\end{enumerate}
We say that the $\Sigma$ is IND-CPA secure if for any QPT adversary $\cA$, it holds that
\begin{align}
\advb{\Sigma,\cA}{ind}{cpa}(\secp)\seteq \abs{\Pr[ \expa{\Sigma,\cA}{ind}{cpa}(\secp, 0)=1] - \Pr[ \expa{\Sigma,\cA}{ind}{cpa}(\secp, 1)=1] }\leq \negl(\secp).
\end{align}
\end{definition}
If $\Xs =\zo{\ell}$ where $\ell$ is some polynomial and $\Ps$ consists of all polynomial-sized Boolean circuits, we say ABE for circuits in this paper.

It is known that \emph{selectively secure} ABE for circuits exists under the LWE assumption, which can be upgraded into adaptively secure one by complexity leveraging if we assume subexponential hardness of the LWE problem ~\cite{JACM:GorVaiWee15}. 
Alternatively, if there exist IO and OWFs, there exists adaptively secure functional encryption for $\Ppoly$~\cite{C:Waters15,C:ABSV15}, which can be trivially downgraded into (adaptively) IND-CPA secure ABE for circuits.  
\begin{theorem}\label{thm:ABE_circuits_from_LWE_or_IO}
If one of the following holds, there exists (adaptively) IND-CPA secure (key-policy) ABE scheme for circuits against QPT adversaries.
\begin{itemize}
    \item the LWE problem is subexponentially hard against QPT adversaries~\cite{JACM:GorVaiWee15}.
    \item there exist IO and OWFs secure against QPT adversaries~\cite{C:Waters15,C:ABSV15}.
\end{itemize}
\end{theorem}

\paragraph{Encryption with certified deletion.} Broadbent and Islam introduced the notion of encryption with certified deletion~\cite{TCC:BroIsl20}.
Their notion is for secret key encryption (SKE).
They consider a setting where a secret key is used only once (that is, one time SKE).
Although it is easy to extend the definition to reusable secret key setting, we describe the definition for the one-time setting in this section.
We provide a definition that is accommodated to the reusable setting in~\cref{sec:reusable_SKE_cd}.  

\begin{definition}[One-Time SKE with Certified Deletion (Syntax)]\label{def:sk_cert_del}
A one-time secret key encryption scheme with certified deletion is a tuple of QPT algorithms $(\keygen,\Enc,\Dec,\Delete,\Vrfy)$ with plaintext space $\Ms$ and key space $\Ks$.
\begin{description}
    \item[$\keygen (1^\secp) \ra \sk$:] The key generation algorithm takes as input the security parameter $1^\secp$ and outputs a secret key $\sk \in \Ks$.
    \item[$\Enc(\sk,m) \ra \ct$:] The encryption algorithm takes as input $\sk$ and a plaintext $m\in\Ms$ and outputs a ciphertext $\ct$.
    \item[$\Dec(\sk,\ct) \ra m^\prime \mbox{ or }\bot$:] The decryption algorithm takes as input $\sk$ and $\ct$ and outputs a plaintext $m^\prime \in \Ms$ or $\bot$.
    \item[$\Delete(\ct) \ra \cert$:] The deletion algorithm takes as input $\ct$ and outputs a certification $\cert$.
    \item[$\Vrfy(\sk,\cert)\ra \top \mbox{ or }\bot$:] The verification algorithm takes $\sk$ and $\cert$ and outputs $\top$ or $\bot$.
\end{description}
\end{definition}

\begin{definition}[Correctness for One-Time SKE with Certified Deletion]\label{def:sk_cd_correctness}
There are two types of correctness. One is decryption correctness and the other is verification correctness.
\begin{description}
\item[Decryption correctness:] For any $\secp\in \N$, $m\in\Ms$, 
\begin{align}
\Pr\left[
\Dec(\sk,\ct)\ne m
\ \middle |
\begin{array}{ll}
\sk\lrun \keygen(1^\secp)\\
\ct \lrun \Enc(\sk,m)
\end{array}
\right] 
\le\negl(\secp).
\end{align}

\item[Verification correctness:] For any $\secp\in \N$, $m\in\Ms$, 
\begin{align}
\Pr\left[
\Vrfy(\sk,\cert)=\bot
\ \middle |
\begin{array}{ll}
\sk\lrun \keygen(1^\secp)\\
\ct \lrun \Enc(\sk,m)\\
\cert \lrun \Delete(\ct)
\end{array}
\right] 
\le\negl(\secp).
\end{align}
\end{description}
\end{definition}

\begin{definition}[Certified Deletion Security for One-Time SKE]\label{def:sk_certified_del}
Let $\Sigma=(\keygen, \Enc, \Dec, \Delete, \Vrfy)$ be a secret key encryption with certified deletion.
We consider the following security experiment $\expb{\Sigma,\cA}{otsk}{cert}{del}(\secp,b)$.

\begin{enumerate}
    \item The challenger computes $\sk \la \keygen(1^\secp)$.
    \item $\cA$ sends $(m_0,m_1)\in\cM^2$ to the challenger.
    \item The challenger computes $\ct_b \la \Enc(\sk,m_b)$ and sends $\ct_b$ to $\cA$.
    \item $\cA$ sends $\cert$ to the challenger.
    \item The challenger computes $\Vrfy(\sk,\cert)$. If the output is $\bot$, the challenger sends $\bot$ to $\cA$.
    If the output is $\top$, the challenger sends $\sk$ to $\cA$. 
    \item $\cA$ outputs $b'\in \bit$.
\end{enumerate}
We say that the $\Sigma$ is OT-CD secure if for any QPT $\cA$, it holds that
\begin{align}
\advc{\Sigma,\cA}{otsk}{cert}{del}(\secp)\seteq \abs{\Pr[ \expb{\Sigma,\cA}{otsk}{cert}{del}(\secp, 0)=1] - \Pr[ \expb{\Sigma,\cA}{otsk}{cert}{del}(\secp, 1)=1] }\leq \negl(\secp).
\end{align}
\end{definition}
We sometimes call it one-time SKE with certified deletion if it satisfies OT-CD security.

\begin{remark}
\cref{def:sk_certified_del} intuitively means that once the valid certificate is issued, 
decrypting the ciphertext becomes impossible.
One might think that it would be also possible to define the inverse: once the chiphertext is decrypted, the valid certificate can no longer be issued.
This property is, however, impossible to achieve due to the decryption correctness (\cref{def:sk_cd_correctness}). In fact, if the quantum decryption algorithm $\Dec$ on a quantum ciphertext $\ct$ succeeds with probability at least $1-\negl(\secp)$, then the gentle measurement lemma guarantees that $\ct$ is only negligibly disturbed, from which the valid certificate can be issued.
\end{remark}

We emphasize that in the existing construction of SKE with certified deletion, a secret key is a classical string though a ciphertext must be a quantum state. Broadbent and Islam prove the following theorem.
\begin{theorem}[\cite{TCC:BroIsl20}]\label{thm:ske_cert_del_no_assumption}
There exists OT-CD secure SKE with certified deletion with $\Ms= \bit^{\ell_m}$ and $\Ks= \bit^{\ell_k}$ where $\ell_m$ and $\ell_k$ are some polynomials, unconditionally.
\end{theorem}

\paragraph{Receiver non-committing encryption.}
We introduce the notion of (public key) receiver non-committing encryption (RNCE)~\cite{STOC:CFGN96,EC:JarLys00,TCC:CanHalKat05}, which is used in~\cref{sec:const_pk_cd_from_sk,sec:PKE_cd_cc_construction,sec:pk_pv_cd_cc_construction}. We sometimes simply write NCE to mean RNCE since we consider only RNCE in this paper. See~\cref{sec:reusable_SKE_cd} for the definition of secret key NCE.
\begin{definition}[RNCE (Syntax)]\label{def:nce_syntax}
An NCE scheme is a tuple of PPT algorithms $(\keygen,\Enc,\Dec,\Fake,\Reveal)$ with plaintext space $\Ms$.
\begin{description}
    \item [$\keygen(1^\secp)\ra (\pk,\sk,\aux)$:] The key generation algorithm takes as input the security 
    parameter $1^\secp$ and outputs a key pair $(\pk,\sk)$ and an auxiliary information $\aux$.
    \item [$\Enc(\pk,m)\ra \ct$:] The encryption algorithm takes as input $\pk$ and a plaintext $m\in\cM$ and outputs a ciphertext $\ct$.
    \item [$\Dec(\sk,\ct)\ra m^\prime \mbox{ or }\bot$:] The decryption algorithm takes as input $\sk$ and $\ct$ and outputs a plaintext $m^\prime$ or $\bot$.
    \item [$\Fake(\pk,\sk,\aux)\ra \tlct$:] The fake encryption algorithm takes $\pk$, $\sk$ and $\aux$, and outputs a fake ciphertext $\tlct$.
    \item [$\Reveal(\pk,\sk,\aux,\tlct,m)\ra \tlsk $:] The reveal algorithm takes $\pk,\sk,\aux,\tlct$ and $m$, and outputs a fake secret key $\tlsk$.
\end{description}  
\end{definition}

Correctness is the same as that of PKE.
\begin{definition}[Receiver Non-Committing (RNC) Security]\label{def:nce_security}
An NCE scheme is RNC secure if it satisfies the following.
Let $\Sigma=(\keygen, \Enc, \Dec, \Fake,\Reveal)$ be an NCE scheme.
We consider the following security experiment $\expa{\Sigma,\cA}{rec}{nc}(\secp,b)$.

\begin{enumerate}
    \item The challenger computes $(\pk,\sk,\aux) \lrun \keygen(1^\secp)$ and sends $\pk$ to $\cA$.
    \item $\cA$ sends a query $m \in \Ms$ to the challenger.
    \item The challenger does the following.
    \begin{itemize}
    \item If $b =0$, the challenger generates $\ct \lrun \Enc(\pk,m)$ and returns $(\ct,\sk)$ to $\cA$.
    \item If $b=1$, the challenger generates $\tlct \lrun \Fake(\pk,\sk,\aux)$ and $\tlsk \lrun \Reveal(\pk,\sk,\aux,\tlct,m)$ and returns $(\tlct,\tlsk)$ to $\cA$.
    \end{itemize}
    \item $\cA$ outputs $b'\in \bit$.
\end{enumerate}
Let $\advb{\Sigma,\cA}{rec}{nc}(\secp)$ be the advantage of the experiment above.
We say that the $\Sigma$ is RNC secure if for any QPT adversary, it holds that
\begin{align}
\advb{\Sigma,\cA}{rec}{nc}(\secp)\seteq \abs{\Pr[ \expa{\Sigma,\cA}{rec}{nc}(\secp, 0)=1] - \Pr[ \expa{\Sigma,\cA}{rec}{nc}(\secp, 1)=1] }\leq \negl(\secp).
\end{align}
\end{definition}

\begin{theorem}[{\cite{C:KNTY19}}]\label{thm:indcpa-pke_to_rnc-pke}
If there exists an IND-CPA secure SKE/PKE scheme (against QPT adversaries), there exists an RNC secure secret/public key NCE scheme (against QPT adversaries) with plaintext space $\bit^{\ell}$, where $\ell$ is some polynomial, respectively.
\end{theorem}
Note that Kitagawa, Nishimaki, Tanaka, and Yamakawa~\cite{C:KNTY19} prove the theorem above for the SKE case, but it is easy to extend their theorem to the PKE setting.
We also note that the core idea of Kitagawa et al. is based on the observation by Canetti, Halevi, and Katz~\cite{TCC:CanHalKat05}.

\paragraph{Non-interactive zero-knowledge.}
We review non-interactive zero-knowledge (NIZK) which is used in~\cref{sec:NCABE_from_IO}.
\begin{definition}[Non-Interactive Zero-Knowledge Proofs (Syntax)] \label{def:NIZKP}
A non-interactive zero-knowledge (NIZK) proof for an $\NP$ language $\Lang$ consists of PPT algorithms $(\Setup, \Prove, \Vrfy)$. 

\begin{description}
    \item[$\Setup(1^\secp) \ra \crs$:] The setup algorithm takes as input the security parameter $1^\secp$ and outputs a common reference string $\crs$.

    \item[$\Prove(\crs, x, w) \ra \pi$:] The prover's algorithm takes as input a common reference string $\crs$, a statement $x$, and a witness $w$ and outputs a proof $\pi$.
    
    \item[$\Vrfy (\crs, x, \pi) \ra \top \mbox{ or }\bot$:] The verifier's algorithm takes as input a common reference string $\crs$, a statement $x$, and a proof $\pi$ and outputs $\top$ to indicate acceptance of the proof and $\bot$ otherwise.
\end{description}

A non-interactive proof must satisfy the following requirements.
\begin{description}
\item[Completeness:] For all $\secp\in\N$ and all pairs $(x, w) \in  \Rela_\Lang$, where $\Rela_\Lang$ is the witness relation corresponding to $\Lang$, we have
        \begin{align}
            \Pr[\Vrfy(\crs, x, \pi) = \top \mid \crs \lrun \Setup(1^\secp), \pi  \lrun \Prove(\crs, x, w)] = 1.
        \end{align}
\item[Statistical Soundness:] For all unbounded time adversaries $\cA$, if we run $\crs \lrun \Setup(1^\secp)$, then we have
        \begin{align}
            \Pr [x \not \in \Lang \land \Vrfy(\crs, x, \pi) = \top  \mid (x, \pi)  \lrun \cA(1^\secp, \crs)] \leq \negl(\secp).
        \end{align}  
\item[(Computational) Zero-Knowledge:]
If there exists a PPT simulator $\Sim = (\Sim_1,\Sim_2)$ such that for all QPT adversaries $\cA$ and for all $(x,w)\in\Rela_\Lang$, we have 
        \begin{align}
            \left| \Pr\left[
            \cA(1^\secp,   \crs,x,\pi) = 1
            \ \middle |
\begin{array}{ll}
 \crs \lrun \Setup(1^\secp),\\
\pi\lrun\Prove(\crs,x,w)
\end{array}
 \right] - 
           \Pr\left[
            \cA(1^\secp,   \tlcrs,x,\pi) = 1
            \ \middle |
\begin{array}{ll}
 (\tlcrs,\td)\lrun \Sim_1(1^\secp,x),\\
\pi\lrun \Sim_2(\tlcrs,\td,x)
\end{array}
\right]
                \right| \leq \negl(\secp).
\end{align}
\end{description}
\end{definition}

\begin{theorem}\label{thm:NIZK_from_LWE_or_IO}
If one of the following holds, then there exists computational NIZK proof for $\NP$ against QPT adversaries.
\begin{itemize}
\item the LWE assumption holds against QPT adversaries~\cite{C:PeiShi19}.
\item there exist IO and OWFs~\cite{TCC:BitPan15}.
\end{itemize}
\end{theorem}

\paragraph{One-shot signature.} We review the one-shot signature scheme that will be used in 
\cref{sec:pk_pv_cd_cc_construction}.
The one-shot signature scheme was introduced 
and a construction relative to a classical oracle was given in \cite{STOC:AGKZ20}.

\begin{definition}[One-Shot Signature (Syntax)]\label{definition:OSS}
A one-shot signature scheme is a tuple of QPT  algorithms $(\Setup,\keygen,\allowbreak\Sign,\Vrfy)$ with
the following syntax:
\begin{description}
\item[$\Setup(1^\lambda)\ra \crs$:] The setup algorithm takes as input a security parameter $1^\secp$,
and outputs a classical common reference string $\crs$.
    \item[$\keygen(\crs)\ra (\pk,\sk)$:] The key generation algorithm takes as input 
a common reference string $\crs$, and outputs a classical public key $\pk$ and a
quantum secret key $\sk$.
    \item[$\Sign(\sk,m) \ra \sigma$:] The signing algorithm takes as input the secret key $\sk$ and
    a message $m$, and outputs a classical signature $\sigma$.
    \item[$\Vrfy(\crs,\pk,\sigma,m) \ra \top \mbox{ or } \bot$:] The verification algorithm takes as input 
    the common reference string $\crs$, 
    the public key $\pk$, the signature $\sigma$, and the message $m$, and outputs 
    $\top$ or $\bot$.
\end{description}
\end{definition}

\begin{definition}[Correctness for One-Shot Signature]
We say that a one-shot signature scheme is correct if
\begin{align}
{\rm Pr}[\Vrfy(\crs,\pk,\Sign(\sk,m),m)=\top|
(\pk,\sk)\leftarrow \keygen(\crs),\crs\leftarrow\Setup(1^\secp)  
]\ge1-\negl(\secp)
\end{align}
for any message $m$. 
\end{definition}

\begin{definition}[Security for One-Shot Signature]
We say that a one-shot signature scheme is secure if
for any QPT adversary $\cA$,
\begin{align}
{\rm Pr}[
m_0\neq m_1 \wedge 
\Vrfy(\crs,\pk,\sigma_0,m_0)=\top \wedge \Vrfy(\crs,\pk,\sigma_1,m_1)=\top]\le \negl(\secp)
\end{align}
is satisfied
over $\crs\leftarrow \Setup(1^\secp)$
and $(\pk,m_0,\sigma_0,m_1,\sigma_1)\leftarrow \cA(\crs)$.
\end{definition}

\paragraph{Witness encryption.} We review the notion of witness encryption, which will be used
in \cref{sec:pk_pv_cd_cc_construction}.
The witness encryption was introduced in \cite{STOC:GGSW13}. 
The extractable witness encryption was introduced in \cite{C:GKPVZ13}.

\begin{definition}[Witness Encryption for NP (Syntax)]\label{definition:WE}
A witness encryption scheme for an $\NP$ language $\Lang$ 
is a pair of algorithms $(\Enc,\Dec)$ with the following syntax:
\begin{description}
    \item[$\Enc(1^\lambda,x,m) \ra \ct$:] The encryption algorithm takes as input the security parameter $1^\lambda$, a statement $x$, and a message $m$, and outputs a ciphertext $\ct$.
    \item[$\Dec(\ct,w) \ra m \mbox{ or }\bot$:] The decryption algorithm takes as input 
    the ciphertext $\ct$ and a witness $w$, and
    outputs $m$ or $\bot$.
\end{description}
\end{definition}

\begin{definition}[Correctness for Witness Encryption]
We say that a witness encryption scheme is correct if
the following holds with overwhelming probability over the randomness of 
$\Enc$ and $\Dec$. For any $(x,w)\in R_\Lang$, where $R_\Lang$ is the corresponding witness relation, and any message $m$, it holds that
$\Dec(\Enc(1^\secp,x,m),w)=m$. 
\end{definition}

\begin{definition}[Security for Witness Encryption]
We say that a witness encryption scheme is secure if
for any instance $x\notin \Lang$ and any two messages $m_0$ and $m_1$, it holds that
for any QPT adversary $\cA$, 
\begin{align}
|\Pr[\cA(\Enc(1^\secp,x,m_0))=1]-\Pr[\cA(\Enc(1^\secp,x,m_1)) = 1]| \le \negl(\secp).
\end{align}
\end{definition}

\begin{definition}[Extractable Security for Witness Encryption]
For any QPT adversary $\cA$, polynomial $p$ and messages $m_0$ and $m_1$, there is a 
QPT extractor $\cE$
and polynomial $q$ such that for any mixed state $\aux$ potentially entangled with an
external register, if
$|\Pr[\cA(\aux,\Enc(1^\secp,x,m_0))=1]-\Pr[\cA(\aux,\Enc(1^\secp,x,m_1)) = 1]| \ge \frac{1}{p(\secp)}$
then
$\Pr[(x,\cE(1^\secp,x,\aux))\in R_\Lang] \ge\frac{1}{q(\secp)}$.
\end{definition}

\ifnum\noclassic=1
\else

\subsection{Noisy Trapdoor Claw-Free Functions}\label{sec:NTCF}
We define noisy trapdoor claw-free function family and its injective invariance following \cite{FOCS:BCMVV18,FOCS:Mahadev18a}. 
\begin{definition}[NTCF Family]\label{def:NTCF}
Let $\mathcal{X}$, $\mathcal{Y}$ be finite sets, $\mathcal{D_Y}$ the set of probability distribution 
over $\mathcal{Y}$, and $\mathcal{K_F}$ a finite set of keys. A family of functions 
\begin{align}
\mathcal{F}=\{f_{\fk,b}:\mathcal{X}\rightarrow\mathcal{D_Y} \}_{\fk\in\mathcal{K_F},b\in\{0,1\}}
\end{align}
is a noisy trapdoor claw-free function (NTCF) family if the following holds.\\

\begin{itemize}
\item{\bf Efficient Function Generation:}
 There exists a PPT algorithm $\GenF$
 which takes the security parameter $1^\secp$ as input and 
 outputs  a key $\fk\in\mathcal{K_F}$ and a trapdoor $\td$.
 
 \item{\bf Trapdoor Injective Pair:} For all keys $\fk\in\mathcal{K_F}$, the following holds.
 
 \begin{enumerate}
 \item Trapdoor: For all $b\in\{0,1\}$ and $x\neq x' \in\mathcal{X}$, 
 $\Supp(f_{\fk,b}(x))\cap \Supp(f_{\fk,b}(x'))=\emptyset$.
 In addition, there exists an efficient deterministic algorithm $\InvF$ 
 such that for all $b\in\{0,1\}$, $x\in\mathcal{X}$ and $y\in\Supp(f_{\fk,b}(x)),\InvF(\td,b,y)=x$.
 \item Injective pair: There exists a perfect matching relation
  $\mathcal{R}_{\fk}\subseteq \mathcal{X}\times\mathcal{X}$
such that $f_{\fk,0}(x_0)=f_{\fk,1}(x_1)$ if and only if  $(x_0,x_1)\in\mathcal{R}_{\fk}$.
 \end{enumerate}

\item{\bf Efficient Range Superposition:} For all keys $\fk\in\mathcal{K_F}$ and 
$b\in\{0,1\}$, there exists a function $f'_{\fk,b}:\mathcal{X}\rightarrow\mathcal{D_Y}$
such that the following holds
\begin{enumerate}
\item For all $(x_0,x_1)\in \mathcal{R}_{\fk}$ and
 $y\in\Supp(f'_{\fk,b}(x_b))$, $\InvF(\td,b,y)=x_b$ and
  $\InvF(\td,b\oplus1,y)=x_{b\oplus1}$.
\item
There exists an efficient deterministic algorithm $\Chk_{\mathcal{F}}$ 
that takes as input $\fk,b\in\{0,1\}$, $x\in\mathcal{X}$, and $y\in\mathcal{Y}$ and 
outputs 1 if $y\in\Supp(f'_{\fk,b}(x))$ and 0 otherwise. Note that this algorithm does not take the trapdoor $\td$ as input.
\item
For all $\fk\in\mathcal{K}_{\mathcal{F}}$ and $b\in\{0,1\}$, 
\begin{align}
\mathbb{E}_{x\leftarrow\mathcal{X}}[ {\bf H}^2(f_{\fk,b}(x),f'_{\fk,b}(x))]\leq \negl(\lambda).
\end{align}
Here {\bf H}$^2$ is the Hellinger distance(See Section~\ref{sec:notation}). In addition, there exists a QPT algorithm 
$\Samp_{\mathcal{F}}$ that takes as input $\fk$ and $b\in\{0,1\}$ and prepares the quantum state 
\begin{align}
|\psi '\rangle= \frac{1}{\sqrt{|\mathcal{X}|}}\sum_{x\in\mathcal{X},
y\in\mathcal{Y}}\sqrt{(f'_{\fk,b}(x))(y)}|x\rangle|y\rangle.
\end{align}
This property immediately means that 
\begin{align}
\parallel |\psi\rangle\langle\psi|- 
|\psi'\rangle\langle\psi'|\parallel_{tr}\leq \negl(\lambda)
\end{align}
where $|\psi\rangle=\frac{1}{\sqrt{|\mathcal{X}|}}
\sum_{x\in\mathcal{X},y\in\mathcal{Y}}\sqrt{(f_{\fk,b}(x))(y)}|x\rangle|y\rangle$.
 \end{enumerate}
\item{\bf Adaptive Hardcore Bit:} For all keys $\fk\in\cK_\cF$, the following holds. For some integer $w$ that is a polynomially bounded function of $\secp$,
\begin{enumerate}
\item For all $b\in\bit$ and $x\in\mathcal{X}$, there exists a set $G_{\fk,b,x}\subseteq \zo{w}$ such that $\Pr_{d\chosen \zo{w}}[d \notin G_{\fk,b,x}] \le \negl(\secp)$. In addition, there exists a PPT algorithm that checks for membership in $G_{\fk,b,x}$ given $\fk,b,x$, and $\td$.
\item  There is an efficiently computable injection $J: \cX \ra \zo{w}$ such that $J$ can be inverted efficiently on its range, and such that the following holds. Let
\begin{align}
 H_\fk  & \seteq \setbracket{(b,x_b,d, d\cdot(J(x_0)\xor J(x_1))) \mid b\in\zo{},(x_0,x_1)\in \cR_\fk, d \in G_{\fk,0,x_0} \cap G_{\fk,1,x_1}},\\
 \overline{H}_\fk & \seteq \setbracket{(b,x_b,d, e)\mid (b,x,d,e\xor 1)\in H_\fk},
\end{align}
then for any QPT $\A$, it holds that
\begin{align}
\abs{\Pr_{(\fk,\td)\gets \GenF (1^\secp)}[\A(\fk)\in H_\fk] - \Pr_{(\fk,\td)\gets \GenF (1^\secp)}[\A(\fk)\in \overline{H}_\fk]} \le \negl(\secp).
\end{align}
\end{enumerate}
 \end{itemize}
\end{definition}

It is known that we can amplify the adaptive hardcore property by parallel repetition in the following sense. 
\begin{lemma}[Amplified Adaptive Hardcore Property]\cite{CoRR:RadSat19,EPRINT:KitNisYam20}\label{lem:amplified_adaptive_hardcore}
Any NTCF family $\mathcal{F}$ satisfies the amplified adaptive hardcore property,
which means that for any QPT adversary $\cA$ and $n=\omega(\log\lambda)$, 
 \begin{align}
\Pr\left[
 \begin{array}{ll}
 \forall i\in[n]~
 \Chk_{\mathcal{F}}(\fk_i,b_i,x_i,y_i)=1,\\ 
d_i\in  G_{\fk_i,0,x_{i,0}} \cap G_{\fk_i,1,x_{i,1}},\\
e_i=d_i\cdot (J(x_{i,0})\oplus J(x_{i,1}))
 \end{array}
 \middle | 
 \begin{array}{ll}
 (\fk_i,\td_i)\gets \GenF(1^\secp) \text{~for~} i\in[n] \\
 \{(b_i,x_i,y_i,d_i,e_i)\}_{i\in[n]}\gets \A(\{\fk_i\}_{i\in[n]})\\
 x_{i,\beta} \gets \InvF (\td_i,\beta,y_i)\text{~for~}(i,\beta)\in[n]\times \bit
 \end{array}
 \right]\leq \negl(\secp).
 \end{align}
\end{lemma}
 We call the procedures in the right half of the above probability the \emph{amplified adaptive hardcore game} and we
 say that $\A$ wins the game if the conditions in the left half of the probability are satisfied.
 By using this terminology, the above lemma says that no QPT adversary wins the amplified adaptive hardcore game with non-negligible probability.  
\begin{remark}
A similar lemma is presented in \cite{EPRINT:KitNisYam20} with the difference that the first condition $\Chk_{\mathcal{F}}(\fk_i,b_i,x_i,y_i)=1$ is replaced with  $x_i=x_{i,b_i}$. 
Since $\Chk_{\mathcal{F}}(\fk_i,b_i,x_i,y_i)=1$ implies $x_i=x_{i,b_i}$ by the first and second items of efficient range superposition property, the lemma in the above form also follows.
\end{remark}

Next, we define the injective invariance for an NTCF family.
For defining this, we first define a trapdoor injective function family.
\begin{definition}[Trapdoor Injective Function Family]\label{def:TIF}
Let $\mathcal{X}$, $\mathcal{Y}$ be finite sets, $\mathcal{D_Y}$ the set of probability distribution 
over $\mathcal{Y}$, and $\mathcal{K_G}$ a finite set of keys. A family of functions 
\begin{align}
\mathcal{G}=\{g_{\fk,b}:\mathcal{X}\rightarrow\mathcal{D_Y} \}_{\fk\in\mathcal{K_G},b\in\{0,1\}}
\end{align}
is a trapdoor injective function family if the following holds.\\

\begin{itemize}
\item{\bf Efficient Function Generation:}
 There exists a PPT algorithm $\GenG$
 which 
 takes the security parameter $1^\secp$ as input and 
 outputs a key $\fk\in\mathcal{K_G}$ and a trapdoor $\td$.
 
 \item{\bf Disjoint Trapdoor Injective Pair:} For all keys $\fk\in\mathcal{K_G}$,
 for all $b,b'\in \bit$, and $x,x'\in \mathcal{X}$, 
 if $(b,x)\neq (b',x')$, 
 $\Supp(g_{\fk,b}(x))\cap \Supp(g_{\fk,b'}(x'))=\emptyset$. 
 Moreover, there exists an efficient deterministic algorithm $\InvG$ 
 such that for all  
 $b\in \bit$, 
 $x\in\mathcal{X}$ and $y\in\Supp(g_{\fk,b}(x)),\InvG(\td,y)=(b,x)$.

\item{\bf Efficient Range Superposition:} For all keys $\fk\in\mathcal{K_G}$ and 
$b\in\{0,1\}$, 
\begin{enumerate}
\item
There exists an efficient deterministic algorithm $\Chk_{\mathcal{G}}$ 
that takes as input $\fk,b\in\{0,1\}$, $x\in\mathcal{X}$, and $y\in\mathcal{Y}$ and 
outputs 1 if $y\in\Supp(g_{\fk,b}(x))$ and 0 otherwise. Note that this algorithm does not take the trapdoor $\td$ as input.
\item
There exists a QPT algorithm 
$\Samp_{\mathcal{G}}$ that takes as input $\fk$ and $b\in\{0,1\}$ and outputs the quantum state 
\begin{align}
 \frac{1}{\sqrt{|\mathcal{X}|}}\sum_{x\in\mathcal{X},
y\in\mathcal{Y}}\sqrt{(g_{\fk,b}(x))(y)}|x\rangle|y\rangle.
\end{align}
 \end{enumerate}
 \end{itemize}
\end{definition}

\begin{definition}[Injective Invariance]\label{def:injective_invariance}
We say that a NTCF family $\mathcal{F}$ is injective invariant if there exists a trapdoor injective family $\mathcal{G}$ such that:
\begin{enumerate}
    \item The algorithms $\Chk_\mathcal{F}$ and $\Samp_{\mathcal{F}}$ are the same as the algorithms $\Chk_\mathcal{G}$ and $\Samp_{\mathcal{G}}$. 
    In this case, we simply write $\Chk$ and $\Samp$ to mean them. 
    \item For all QPT algorithm $\cA$, we have 
    \begin{align}
    \abs{\Pr[\cA(\fk)=1:(\fk,\td)\lrun \GenF(1^\secp)]-\Pr[\cA(\fk)=1:(\fk,\td)\lrun \GenG(1^\secp)]}\leq  \negl(\secp).
    \end{align}
\end{enumerate}
\end{definition}

\begin{lemma}[\cite{FOCS:Mahadev18a}]
If the LWE assumption holds against QPT adversaries, there exists an injective invariant NTCF family.
\takashi{May be better to mention parameter for LWE.}
\end{lemma}

\subsection{Quantum Random Oracle Model}\label{sec:QROM}
In \cref{sec:classical_com}, we rely on the quantum random oracle model (QROM)~\cite{AC:BDFLSZ11}. 
In the QROM, a uniformly random function with certain domain and range is chosen at the beginning, and quantum access to this function is given for all parties including an adversary.
Zhandry  showed that quantum access to random function can be efficiently simulatable by using so called the compressed oracle technique~\cite{C:Zhandry19}.

We review the one-way to hiding lemma \cite{JACM:Unruh15,C:AmbHamUnr19}, which is often useful when analyzing schemes in the QROM. 
The following form of the lemma is based on \cite{C:AmbHamUnr19}. 
\begin{lemma}[{One-Way to Hiding Lemma~\cite{C:AmbHamUnr19}}]
\label{lem:o2h}
Let $S \subseteq \mathcal{X}$ be random.
Let $G,H \colon \mathcal{X} \to \mathcal{Y}$ be random functions satisfying $\forall x \not\in S~[G(x) = H(x)]$.
Let $z$ be a random classical bit string or quantum state.
($S,G,H,z$ may have an arbitrary joint distribution.)
Let $\cA$ be an oracle-aided quantum algorithm that makes at most $q$ quantum queries. 
Let $\cB$ be an algorithm that on input $z$ chooses $i\lrun [q]$, runs $\cA^H(z)$, measures $\cA$'s $i$-th query, and outputs the measurement outcome.
Then we have
\[
\abs{\Pr[\cA^{G}(z)=1]-\Pr[\cA^H(z)=1]}\leq 2q\sqrt{\Pr[\cB^{H}(z)\in S]}.
\]
\end{lemma}
\begin{remark}
In \cite{C:AmbHamUnr19}, $z$ is assumed to be classical. However, as observed in \cite{AC:HhaXagYam19}, 
the lemma holds even if $z$ is quantum since any quantum state can be described by an exponentially large classical string, and there is no restriction on the size of $z$ in the lemma.
\end{remark}

\fi

\ifnum\noclassic=1
\else
\fi


\section{Public Key Encryption with Certified Deletion}\label{sec:pk_cd}

In this section, we define the notion of PKE with certified deletion, which is a natural extension of SKE with certified deletion and present how to achieve PKE with certified deletion from OT-CD secure SKE and IND-CPA secure (standard) PKE.

\subsection{Definition of PKE with Certified Deletion}\label{sec:pk_cd_def}
The definition of PKE with certified deletion is an extension of SKE with certified deletion. Note that a verification key for verifying a certificate is generated in the encryption algorithm.
\begin{definition}[PKE with Certified Deletion (Syntax)]\label{def:pk_cert_del}
A PKE with certified deletion is a tuple of QPT algorithms $(\keygen,\Enc,\Dec,\Delete,\Vrfy)$ with plaintext space $\Ms$.
\begin{description}
    \item[$\keygen (1^\secp) \ra (\pk,\sk)$:] The key generation algorithm takes as input the security parameter $1^\secp$ and outputs a classical key pair $(\pk,\sk)$.
    \item[$\Enc(\pk,m) \ra (\vk,\ct)$:] The encryption algorithm takes as input the public key $\pk$ and a plaintext $m\in\cM$ and outputs a classical verification key $\vk$ and a quantum ciphertext $\ct$. 
    \item[$\Dec(\sk,\ct) \ra m^\prime \mbox{ or } \bot$:] The decryption algorithm takes as input the secret key $\sk$ and the ciphertext $\ct$, and outputs a classical plaintext $m^\prime$ or $\bot$.
    \item[$\Delete(\ct) \ra \cert$:] The deletion algorithm takes as input the ciphertext $\ct$ and outputs a classical certificate $\cert$.
    \item[$\Vrfy(\vk,\cert)\ra \top \mbox{ or }\bot$:] The verification algorithm takes the verification key $\vk$
    and the certificate $\cert$, and outputs $\top$ or $\bot$.
\end{description}
\end{definition}

\begin{definition}[Correctness for PKE with Certified Deletion]\label{def:pk_cd_correctness}
There are two types of correctness. One is decryption correctness and the other is verification correctness.
\begin{description}
\item[Decryption correctness:] For any $\secp\in \N$, $m\in\Ms$, 
\begin{align}
\Pr\left[
\Dec(\sk,\ct)\ne m
\ \middle |
\begin{array}{ll}
(\pk,\sk)\lrun \keygen(1^\secp)\\
(\vk,\ct) \lrun \Enc(\pk,m)
\end{array}
\right] 
\le\negl(\secp).
\end{align}

\item[Verification correctness:] For any $\secp\in \N$, $m\in\Ms$, 
\begin{align}
\Pr\left[
\Vrfy(\vk,\cert)=\bot
\ \middle |
\begin{array}{ll}
(\pk,\sk)\lrun \keygen(1^\secp)\\
(\vk,\ct) \lrun \Enc(\pk,m)\\
\cert \lrun \Delete(\ct)
\end{array}
\right] 
\le\negl(\secp).
\end{align}

\end{description}
\end{definition}

\begin{definition}[Certified Deletion Security for PKE]\label{def:pk_certified_del}
Let $\Sigma=(\keygen, \Enc, \Dec, \Delete, \Vrfy)$ be a PKE with certified deletion scheme.
We consider the following security experiment $\expb{\Sigma,\cA}{pk}{cert}{del}(\secp,b)$.

\begin{enumerate}
    \item The challenger computes $(\pk,\sk) \la \keygen(1^\secp)$ and sends $\pk$ to $\cA$.
    \item $\cA$ sends $(m_0,m_1)\in \Ms^2$ to the challenger.
    \item The challenger computes $(\vk_b,\ct_b) \la \Enc(\pk,m_b)$ and sends $\ct_b$ to $\cA$.
    \item At some point, $\cA$ sends $\cert$ to the challenger.
    \item The challenger computes $\Vrfy(\vk_b,\cert)$. If the output is $\bot$, it sends $\bot$ to $\cA$.
    If the output is $\top$, it sends $\sk$ to $\cA$.
    \item $\cA$ outputs its guess $b'\in \bit$.
\end{enumerate}
Let $\advc{\Sigma,\cA}{pk}{cert}{del}(\secp)$ be the advantage of the experiment above.
We say that the $\Sigma$ is IND-CPA-CD secure  
if for any QPT adversary $\cA$, it holds that
\begin{align}
\advc{\cE,\cA}{pk}{cert}{del}(\secp)\seteq \abs{\Pr[ \expb{\Sigma,\cA}{pk}{cert}{del}(\secp, 0)=1] - \Pr[ \expb{\Sigma,\cA}{pk}{cert}{del}(\secp, 1)=1] }\leq \negl(\secp).
\end{align}
\end{definition}

\subsection{PKE with Certified Deletion from PKE and SKE with Certified Deletion}\label{sec:const_pk_cd_from_sk}
In this section, we present how to construct a PKE scheme with certified deletion from an SKE scheme with certified deletion and an NCE scheme, which can be constructed from standard IND-CPA PKE schemes.

\paragraph{Our PKE Scheme.}
We construct $\Sigma_{\pkcd} =(\keygen,\Enc,\Dec,\Delete,\Vrfy)$ with plaintext space $\Ms$ from an SKE with certified deletion scheme $\Sigma_{\skcd}=\SKE.(\Gen,\Enc,\Dec,\Delete,\Vrfy)$ with plaintext space $\Ms$ and key space $\Ks$ and a public key NCE scheme $\Sigma_{\nce}=\NCE.(\keygen,\Enc,\Dec,\Fake,\Reveal)$ with plaintext space $\Ks$.

\begin{description}
\item[$\keygen(1^\secp)$:] $ $
\begin{itemize}
\item Generate $(\nce.\pk,\nce.\sk,\nce.\aux)\lrun \NCE.\keygen(1^\secp)$ and output $(\pk,\sk) \seteq (\nce.\pk,\nce.\sk)$.
\end{itemize}
\item[$\Enc(\pk,m)$:] $ $
\begin{itemize}
	\item Parse $\pk = \nce.\pk$.
\item Generate $\ske.\sk \lrun \SKE.\Gen(1^\secp)$.
\item Compute $\nce.\ct \lrun \NCE.\Enc(\nce.\pk,\ske.\sk)$ and $\ske.\ct \lrun \SKE.\Enc(\ske.\sk,m)$.
\item Output $\ct \seteq (\nce.\ct,\ske.\ct)$ and $\vk \seteq \ske.\sk$.
\end{itemize}
\item[$\Dec(\sk,\ct)$:] $ $
\begin{itemize}
\item Parse $\sk = \nce.\sk$ and $\ct = (\nce.\ct,\ske.\ct)$.
\item Compute $\sk^\prime \lrun \NCE.\Dec(\nce.\sk,\nce.\ct)$.
\item Compute and output $m^\prime \lrun \SKE.\Dec(\sk^\prime,\ske.\ct)$.
\end{itemize}
\item[$\Delete(\ct)$:] $ $
\begin{itemize}
\item Parse $\ct= (\nce.\ct,\ske.\ct)$.
\item Generate $\ske.\cert \lrun \SKE.\Delete(\ske.\ct)$.
\item Output $\cert \seteq \ske.\cert$.
\end{itemize}
\item[$\Vrfy(\vk,\cert)$:] $ $
\begin{itemize}
\item Parse $\vk = \ske.\sk$ and $\cert = \ske.\cert$.
\item Output $b \lrun \SKE.\Vrfy(\ske.\sk,\ske.\cert)$.
\end{itemize}
\end{description}

\paragraph{Correctness.}
The decryption and verification correctness easily follow from the correctness of $\Sigma_{\nce}$ and $\Sigma_{\skcd}$.
\paragraph{Security.}
We prove the following theorem.
\begin{theorem}\label{thm:pke_cd_from_sk_cd_and_pke}
If $\Sigma_{\nce}$ is RNC secure and $\Sigma_{\skcd}$ is OT-CD secure, $\Sigma_{\pkcd}$ is IND-CPA-CD secure.
\end{theorem}
\begin{proof}
Let $\cA$ be a QPT adversary and $b\in \bit$ be a bit. 
We define the following hybrid game $\sfhyb{}{}(b)$.
\begin{description}
\item[$\sfhyb{}{}(b)$:] This is the same as $\expb{\Sigma_{\pkcd,\cA}}{pk}{cert}{del}(\secp,b)$ except that the challenger generate the target ciphertext as follows. It generates $\ske.\sk \lrun \SKE.\Gen(1^\secp)$ and computes $\nce.\ct^\ast \lrun \NCE.\Fake(\nce.\pk,\nce.\sk,\nce.\aux)$ and $\ske.\ct^\ast \lrun \SKE.\Enc(\ske.\sk,m_b)$. The target ciphertext is $\ct^\ast \seteq (\nce.\ct^\ast,\ske.\ct^\ast)$. In addition, we reveal $\tlsk \lrun \Reveal(\nce.\pk,\nce.\sk,\nce.\aux,\nce.\ct^\ast,\ske.\sk)$ instead of $\nce.\sk$.
\end{description}

\begin{proposition}\label{prop:pkecd_hyb_one}
If $\Sigma_{\nce}$ is RNC secure, $\abs{\Pr[\expb{\Sigma_{\pkcd},\cA}{pk}{cert}{del}(\secp, b)=1] - \Pr[\sfhyb{}{}(b)=1]} \le \negl(\secp)$.
\end{proposition}
\begin{proof}
We construct an adversary $\cB_\nce$  that breaks the RNC security of $\Sigma_{\nce}$ by assuming that $\cA$ distinguishes these two experiments.
First, $\cB_\nce$ is given $\nce.\pk$ from the challenger of $\expa{\Sigma_\nce,\cB_\nce}{rec}{nc}(\secp, b')$ for $b'\in \bit$. $\cB_\nce$ generates $\ske.\sk \lrun \SKE.\Gen(1^\secp)$ and sends $\nce.\pk$ to $\cA$.
When $\cA$ sends $(m_0,m_1)$, $\cB_\nce$ sends $\ske.\sk$ to the challenger of $\expa{\Sigma_\nce,\cB_\nce}{rec}{nc}(\secp, b')$, receives $(\nce.\ct^\ast,\tlsk)$, and generates $\ske.\ct \lrun \SKE.\Enc(\ske.\sk,m_b)$. $\cB_\nce$ sends $(\nce.\ct^\ast,\ske.\ct)$ to $\cA$ as the challenge ciphertext.
At some point, $\cA$ outputs $\cert$. If $\SKE.\Vrfy(\ske.\sk,\cert) = \top$, $\cB_\nce$ sends $\tlsk$ to $\cA$.
Otherwise, $\cB_\nce$ sends $\bot$ to $\cA$.
Finally, $\cB_\nce$ outputs whatever $\cA$ outputs.
\begin{itemize}
\item If $b'=0$, i.e., $(\nce.\ct^\ast,\tlsk) = (\NCE.\Enc(\nce.\pk,\ske.\sk),\nce.\sk)$, $\cB_\nce$ perfectly simulates $\expb{\Sigma_{\pkcd},\cA}{pk}{cert}{del}(\secp, b)$.
\item If $b'=1$, i.e., $(\nce.\ct^\ast,\tlsk) = (\NCE.\Fake(\nce.\pk,\nce.\sk,\nce.\aux),\NCE.\Reveal(\nce.\pk,\nce.\sk,\nce.\aux,\nce.\ct^\ast,\allowbreak\ske.\sk))$, $\cB_\nce$ perfectly simulates $\sfhyb{}{}(b)$.
\end{itemize}

Thus, if $\cA$ distinguishes the two experiments, $\cB_\nce$ breaks the RNC security of $\Sigma_\nce$. This completes the proof.
\end{proof}
\begin{proposition}\label{prop:pkecd_hyb_end}
If $\Sigma_{\skcd}$ is OT-CD secure, $\abs{\Pr[\sfhyb{}{}(0)=1] - \Pr[\sfhyb{}{}(1)=1] } \le \negl(\secp)$.
\end{proposition}
\begin{proof}
We construct an adversary $\cB_\skcd$  that breaks the OT-CD security of $\Sigma_{\skcd}$ assuming that $\cA$ distinguishes these two experiments.
$\cB_\skcd$ plays the experiment $\expb{\Sigma_\skcd,\cB_\skcd}{otsk}{cert}{del}(\secp,b')$ for some $b'\in \bit$.
First, $\cB_\skcd$ generates $(\nce.\pk,\nce.\sk,\nce.\aux) \lrun \NCE.\keygen(1^\secp)$ and sends $\nce.\pk$ to $\cA$.
When $\cA$ sends $(m_0,m_1)$, $\cB_\skcd$ sends $(m_0,m_1)$ to the challenger of $\expb{\Sigma_\skcd,\cB_\skcd}{otsk}{cert}{del}(\secp,b')$, receives $\ske.\ct^\ast$, and generates $\nce.\tlct \lrun \NCE.\Fake(\nce.\pk,\nce.\sk,\nce.\aux)$. $\cB_\skcd$ sends $(\nce.\tlct,\ske.\ct^\ast)$ to $\cA$ as the challenge ciphertext.
At some point, $\cA$ outputs $\cert$. $\cB_\skcd$ passes $\cert$ to the challenger of OT-CD SKE. If the challenger returns $\ske.\sk$, $\cB_\skcd$ generates $\tlsk \lrun \NCE.\Reveal(\nce.\pk,\nce.\sk,\nce.\aux,\allowbreak \nce.\tlct,\ske.\sk)$ and sends $\tlsk$ to $\cA$.
Otherwise, $\cB_\skcd$ sends $\bot$ to $\cA$. 
Finally, $\cB_\skcd$ outputs whatever $\cA$ outputs.

\begin{itemize}
\item If $b'=0$, i.e., $\ske.\ct^\ast = \SKE.\Enc(\ske.\sk,m_0)$, $\cB_\skcd$ perfectly simulates $\sfhyb{}{}(0)$.
\item If $b'=1$, i.e., $\ske.\ct^\ast = \SKE.\Enc(\ske.\sk,m_1)$, $\cB_\skcd$ perfectly simulates $\sfhyb{}{}(1)$.
\end{itemize}
Thus, if $\cA$ distinguishes the two experiments, $\cB_\skcd$ breaks the OT-CD security. This completes the proof.
\end{proof}
By~\cref{prop:pkecd_hyb_one,prop:pkecd_hyb_end}, we immediately obtain~\cref{thm:pke_cd_from_sk_cd_and_pke}.
\end{proof}

By~\cref{thm:ske_cert_del_no_assumption,thm:indcpa-pke_to_rnc-pke,thm:pke_cd_from_sk_cd_and_pke}, we immediately obtain the following corollary.

\begin{corollary}
If there exists IND-CPA secure PKE against QPT adversaries, there exists IND-CPA-CD secure PKE with certified deletion.
\end{corollary}

\paragraph{Reusable SKE with certified deletion.}
We can construct a secret key variant of $\Sigma_{\mathsf{pkcd}}$ above (that is, reusable SKE with certified deletion) by replacing $\Sigma_\nce$ with a secret key NCE scheme. We omit the proof since it is almost the same as that of~\cref{thm:pke_cd_from_sk_cd_and_pke}. By~\cref{thm:indcpa-pke_to_rnc-pke} and the fact that OWFs imply (reusable) SKE~\cite{SIAMCOMP:HILL99,JACM:GolGolMic86}, we also obtain the following theorem.

\begin{theorem}\label{thm:reusable_SKE_cd}
If there exists OWF against QPT adversaries, there exists IND-CPA-CD secure SKE with certified deletion.
\end{theorem}

See the definition and construction of reusable SKE with certified deletion in~\cref{sec:reusable_SKE_cd}.

\section{Attribute-Based Encryption with Certified Deletion}\label{sec:abe_cd}
In this section, we define the notion of attribute-based encryption (ABE) with certified deletion, which is a natural extension of ABE and PKE with certified deletion and present how to achieve ABE with certified deletion from OT-CD secure SKE, IO, and OWFs. 
In~\cref{sec:abe_cd_def}, we present the definition of ABE with certified deletion and non-committing ABE (NCABE), which is a crucial tool to achieve ABE with certified deletion. In~\cref{sec:NCABE_from_IO}, we present how to achieve NCABE from IO and standard ABE.
In~\cref{sec:const_abe_cd_from_sk}, we present how to achieve ABE with certified deletion from NCABE and OT-CD secure SKE with certified deletion.

\subsection{Definition of ABE with Certified Deletion}\label{sec:abe_cd_def}
The definition of ABE with certified deletion is a natural combination of ABE and PKE with certified deletion.
\begin{definition}[Attribute-Based Encryption with Certified Deletion (Syntax)]\label{def:abe_cert_del}
An ABE scheme with certified deletion is a tuple of QPT algorithms $(\Setup,\keygen,\Enc,\Dec,\Delete,\Vrfy)$ with plaintext space $\Ms$, attribute space $\Xs$, and policy space $\Ps$.
\begin{description}
    \item[$\Setup(1^\secp)\ra (\pk,\msk)$:] The setup algorithm takes as input the security parameter $1^\secp$ and outputs a public key $\pk$ and a master secret key $\msk$.
    \item[$\keygen (\msk,P) \ra \sk_P$:] The key generation algorithm takes as input $\msk$ and a policy $P\in \Ps$, and outputs a secret key $\sk_P$.
    \item[$\Enc(\pk,X,m) \ra (\vk,\ct_X)$:] The encryption algorithm takes as input $\pk$, an attribute $X\in\Xs$, and a plaintext $m \in \Ms$, and outputs a verification key $\vk$ and ciphertext $\ct_X$.
    \item[$\Dec(\sk_P,\ct_X) \ra m^\prime \mbox{ or } \bot$:] The decryption algorithm takes as input $\sk_P$ and $\ct_X$, and outputs a plaintext $m^\prime \in \Ms$ or $\bot$.
    \item[$\Delete(\ct_X) \ra \cert$:] The deletion algorithm takes as input $\ct_X$ and outputs a certification $\cert$.
    \item[$\Vrfy(\vk,\cert)\ra \top \mbox{ or }\bot$:] The verification algorithm takes as input $\vk$ and $\cert$, and outputs $\top$ or $\bot$. 
\end{description}
\end{definition}

\begin{definition}[Correctness for ABE with Certified Deletion]\label{def:abe_cd_correctness}
There are two types of correctness. One is decryption correctness and the other is verification correctness.
\begin{description}
\item[Decryption correctness:] For any $\secp\in \N$, $m\in\Ms$, $P\in\Ps$, and $X\in\Xs$ such that $P(X)=\top$,
\begin{align}
\Pr\left[
\Dec(\sk_P,\ct_X)\ne m
\ \middle |
\begin{array}{ll}
(\pk,\msk)\lrun \Setup(1^\secp)\\
\sk_P \lrun \keygen(\msk,P)\\
(\vk,\ct_X) \lrun \Enc(\pk,X,m)
\end{array}
\right] 
\le\negl(\secp).
\end{align}

\item[Verification correctness:] For any $\secp\in \N$, $P\in\Ps$, $X\in\Xs$, $m\in\Ms$, 
\begin{align}
\Pr\left[
\Vrfy(\vk,\cert)=\bot
\ \middle |
\begin{array}{ll}
(\pk,\msk)\lrun \Setup(1^\secp)\\
(\vk,\ct_X) \lrun \Enc(\pk,X,m)\\
\cert \lrun \Delete(\ct_X)
\end{array}
\right] 
\le\negl(\secp).
\end{align}
\end{description}
\end{definition}

\begin{definition}[ABE Certified Deletion Security]\label{def:abe_certified_del}
Let $\Sigma=(\Setup, \keygen, \Enc, \Dec, \Delete, \Vrfy)$ be an ABE with certified deletion.
We consider the following security experiment $\expc{\Sigma,\cA}{ind}{cpa}{cd}(\secp,b)$.

\begin{enumerate}
    \item The challenger computes $(\pk,\msk) \lrun \Setup(1^\secp)$ and sends $\pk$ to $\cA$.
    \item $\cA$ sends a key query $P \in\Ps$ to the challenger and it returns $\sk_{P} \lrun \keygen(\msk,P)$ to $\cA$. This process can be repeated polynomially many times.
    \item $\cA$ sends $X^\ast \in \Xs$ and $(m_0,m_1) \in \Ms^2$ to the challenger where $X^\ast$ must satisfy $P(X^\ast)=\bot$ for all key queries $P$ sent so far.
    \item The challenger computes $(\vk_b,\ct_b) \la \Enc(\pk,X^\ast,m_b)$ and sends $\ct_b$ to $\cA$.
    \item Again, $\cA$ can send key queries $P$ that must satisfy $P(X^\ast)=\bot$.
    \item $\cA$ computes $\cert \la \Delete(\ct_b)$ and sends $\cert$ to the challenger.
    \item The challenger computes $\Vrfy(\vk_b,\cert)$. If the output is $\bot$, the challenger sends $\bot$ to $\cA$. If the output is $\top$, the challenger sends $\msk$ to $\cA$.
    \item Again, $\cA$ can send key queries $P$ that must satisfy $P(X^\ast)=\bot$.\footnote{Such queries are useless if $\cA$ obtains $\msk$ in the previous item, but may be useful if the challenger returns $\bot$ there.}
    \item $\cA$ outputs its guess $b'\in \bit$.
\end{enumerate}
We say that the $\Sigma$ is IND-CPA-CD secure if for any QPT adversary $\cA$, it holds that
\begin{align}
\advd{\Sigma,\cA}{ind}{cpa}{cd}(\secp)\seteq \abs{\Pr[ \expc{\Sigma,\cA}{ind}{cpa}{cd}(\secp, 0)=1] - \Pr[ \expc{\Sigma,\cA}{ind}{cpa}{cd}(\secp, 1)=1] }\le \negl(\secp).
\end{align}
\end{definition}

Next, we define receiver non-committing ABE, which is a receiver non-committing encryption version of ABE.
\begin{definition}[Receiver Non-Committing Attribute-Based Encryption (Syntax)]\label{def:ncabe_syntax}
A receiver non-committing (key policy) attributed-based encryption (NCABE)  
is a tuple of PPT algorithms $(\Setup,\keygen,\Enc,\Dec,\FakeSetup,\FakeSK,\allowbreak \FakeCT,\Reveal)$ with plaintext space $\Ms$, attribute space $\Xs$, and policy space $\Ps$.
\begin{description}
    \item [$\Setup(1^\secp)\ra (\pk,\msk)$:] The setup algorithm takes as input the security parameter $1^\secp$ and outputs a public key $\pk$ and a master secret key $\msk$.
    \item [$\keygen(\msk, P)\ra \sk_P$:] The key generation algorithm takes as input $\msk$ and a policy $P\in \Ps$, and outputs a secret key $\sk_P$.
    \item [$\Enc(\pk,X,m)\ra \ct$:] The encryption algorithm takes as input $\pk$, an attribute $X\in \Xs$, and a plaintext $m\in\Ms$, and outputs a ciphertext $\ct$.
    \item [$\Dec(\sk,\ct)\ra m^\prime \mbox{ or }\bot$:] The decryption algorithm takes as input $\sk$ and $\ct$ and outputs a plaintext $m^\prime\in\Ms$ or $\bot$.
    \item [$\FakeSetup(1^\secp)\ra (\pk,\aux)$:] The fake setup algorithm takes as input the security parameter $1^\secp$, and outputs a public key $\pk$ and an auxiliary information $\aux$. 
    \item [$\FakeCT(\pk,\aux,X)\ra \tlct$:] The fake encryption algorithm takes $\pk$, $\aux$, and $X\in\Xs$, and outputs a fake ciphertext $\tlct$.
    \item [$\FakeSK(\pk,\aux,P)\ra \tlsk$:] The fake key generation algorithm takes $\pk$, $\aux$, and $P\in \Ps$, and outputs a fake secret key $\tlsk$.
    \item [$\Reveal(\pk,\aux,\tlct,m)\ra \tlmsk$:] The reveal algorithm takes $\pk,\aux$, a fake ciphertext $\tlct$, and a plaintext $m\in\Ms$, and outputs a fake master secret key $\tlmsk$.
\end{description}
\end{definition}

Correctness is the same as that of ABE.
\begin{definition}[RNC Security for ABE]\label{def:ncabe_security}
An NCABE scheme is RNC secure if it satisfies the following.
Let $\Sigma=(\Setup,\keygen, \Enc, \Dec, \Fake\Setup,\Fake\ct,\Fake\SK,\Reveal)$ be an NCABE scheme.   
We consider the following  security experiment $\expa{\Sigma,\cA}{rnc}{cpa}(\secp,b)$.

\begin{enumerate}
    \item The challenger does the following. 
        \begin{itemize}
     \item If $b=0$, the challenger computes $(\pk,\msk) \lrun \Setup(1^\secp)$ and sends $\pk$ to $\cA$.
     \item If $b=1$, the challenger computes $(\pk,\aux) \lrun \FakeSetup(1^\secp)$ and sends $\pk$ to $\cA$.
     \end{itemize} 
    \item $\cA$ sends a query $P_i \in \Ps$ to the challenger. 
    \begin{itemize}
     \item If $b=0$, the challenger returns a secret key $\sk_{P} \lrun \keygen(\msk,P)$.
     \item If $b=1$, the challenger returns a fake secret key $\tlsk_{P} \lrun \FakeSK(\pk,\aux,P)$.
     \end{itemize} 
      $\cA$ can send polynomially many key queries.
    \item At some point, $\cA$ sends the target attribute
    $X^\ast\in \Xs$ and message
    $m \in \Ms$ to the challenger where $X^\ast$ must satisfy $P(X^\ast)=\bot$ for all key queries $P$ sent so far. The challenger does the following.
    \begin{itemize}
    \item If $b =0$, the challenger generates $\ct^\ast \lrun \Enc(\pk,X^\ast,m)$ and returns $(\ct^*,\msk)$ to $\cA$.
    \item If $b=1$, the challenger generates $\tlct^\ast \lrun \FakeCT(\pk,\aux,X^\ast)$ and $\tlmsk \lrun \Reveal(\pk,\aux,\tlct,m)$ and returns $(\tlct,\tlmsk)$ to $\cA$.
    \end{itemize}
    \item Again, $\cA$ can send key queries $P$ that must satisfy  $P(X^\ast)=\bot$. 
    \item $\cA$ outputs $b'\in \bit$.
\end{enumerate}
We say that $\Sigma$ is RNC secure if for any QPT adversary, it holds that
\begin{align}
\advb{\Sigma,\cA}{rnc}{cpa}(\secp)\seteq \abs{\Pr[ \expa{\Sigma,\cA}{rnc}{cpa}(\secp, 0)=1] - \Pr[ \expa{\Sigma,\cA}{rnc}{cpa}(\secp, 1)=1] }\leq \negl(\secp).
\end{align} 
\end{definition}

\subsection{Non-Committing ABE from IO}\label{sec:NCABE_from_IO}

In this section, we construct NCABE scheme with plaintext space $\bit^{\ell_m}$, attribute space $\Xs$ where $\ell_m$ are some polynomials, and policy space $\Ps$ from IO for $\Ppoly$ and ABE scheme with plaintext space $\bit$, attribute space $\Xs$, and policy space $\Ps$.
\paragraph{Our NCABE scheme.}
Let $\Sigma_{\abe}=\ABE.(\Setup,\keygen,\Enc,\Dec)$ be an IND-CPA secure ABE scheme on the message space $\bit$ and $\Pi_{\nizk}$ be a NIZK proof for the $\NP$ language $\Lang$ corresponding to the following relation $\cR$.
\begin{align}
\cR &\seteq \setbk{((\pk,\setbk{\ct_{i,0},\ct_{i,1}}_{i\in [\ell_m]},X),\setbk{(m[i],r_{i,0},r_{i,1})}_{i\in[\ell_m]}) \mid \forall i\forall b\ \ct_{i,b}=\ABE.\Enc(\abe.\pk_{i,b},X,m[i];r_{i,b})}
\end{align}
where $\pk= \setbk{\abe.\pk_{i,0},\abe.\pk_{i,1}}_{i\in \ell_m}$.

We construct an NCABE scheme $\Sigma_{\nce}=(\Setup,\keygen,\Enc,\Dec,\FakeSetup,\FakeCT,\FakeSK,\Reveal)$ as follows.
\begin{description}
\item[$\Setup(1^\secp):$]$ $
\begin{enumerate}
\item Generate $(\abe.\pk_{i,b},\abe.\msk_{i,b}) \lrun \ABE.\Setup(1^\secp)$ for every $i\in[\ell_m]$ and $b\in\bit$.
\item Choose $z \chosen \bit^{\ell_m}$.
\item Computes $\crs \lrun \NIZK.\Setup(1^\secp)$.
\item Output $\pk \seteq (\setbk{\abe.\pk_{i,b}}_{i\in[\ell_m],b\in\bit},\crs)$ and $\msk\seteq (\pk,\setbk{\abe.\msk_{i,z[i]}}_{i\in[\ell_m]},z)$.
\end{enumerate}
\item[$\keygen(\msk,P)$:]$ $
\begin{enumerate}
\item Parse $\msk = (\pk,\setbk{\abe.\msk_{i,z[i]}}_{i\in[\ell_m]},z)$.
\item Generate $\sk_{i}\la \ABE.\keygen(\abe.\msk_{i,z[i]},P)$ for every $i\in[\ell_m]$.  
\item Generate and output $\sk_P \seteq \iO(\sfD[\crs,\setbk{\sk_{i}}_{i\in[\ell_m]},z])$, where circuit $\sfD$ is described in~\cref{fig:LR_dec_circuit}.
\end{enumerate}
\item[$\Enc(\pk,X,m):$]$ $
\begin{enumerate}
\item Parse $\pk = (\setbk{\abe.\pk_{i,b}}_{i\in[\ell_m],b\in\bit},\crs)$.
\item Generate $\ct_{i,b}\lrun \ABE.\Enc(\abe.\pk_{i,b},X,m[i];r_{i,b})$ for every $i\in[\ell_m]$ and $b\in\bit$ where $r_{i,b}$ is uniformly chosen from the randomness space for $\ABE.\Enc$.
\item Generate $\pi \lrun \NIZK.\Prove(\crs,d,w)$ where $d = (\setbk{(\abe.\pk_{i,0},\abe.\pk_{i,1},\ct_{i,0},\ct_{i,1})}_{i\in[\ell_m]},X)$ and $w = (m,\setbk{r_{i,0},r_{i,1}}_{i\in[\ell_m]})$.
\item Output $\ct_X \seteq (\setbk{\ct_{i,0},\ct_{i,1}}_{i\in[\ell_m]},\pi)$.
\end{enumerate}
\item[$\Dec(\sk_P,\ct_X):$]$ $
\begin{enumerate}
\item Parse $\sk_P = \widetilde{\sfD}$.
\item Compute and output $m\seteq \widetilde{\sfD}(\ct_X)$.
\end{enumerate}
\item[$\FakeSetup(1^\secp):$]$ $
\begin{enumerate}
\item Generate $(\abe.\pk_{i,b},\abe.\msk_{i,b}) \lrun \ABE.\Setup(1^\secp)$ for every $i\in[\ell_m]$ and $b\in\bit$.
\item Choose $z^\ast \chosen \bit^{\ell_m}$.
\item Computes $(\tlcrs,\td) \lrun \Sim_1(1^\secp)$.
\item Output $\pk \seteq (\setbk{\abe.\pk_{i,b}}_{i\in[\ell_m],b\in\bit},\tlcrs)$ and $\aux\seteq (\pk,\td, \setbk{\abe.\msk_{i,b}}_{i\in[\ell_m],b\in\bit},z^\ast)$.
\end{enumerate}
\item[$\FakeSK(\pk,\aux,P):$]$ $
\begin{enumerate}
\item Parse $\aux=(\pk,\td, \setbk{\abe.\msk_{i,b}}_{i\in[\ell_m],b\in\bit},z^\ast)$.
\item Generate $\sk_i^{0} \lrun \ABE.\keygen(\abe.\msk_{i,0},P)$ for every $i \in [\ell_m]$ and set $\sk_P^0 \seteq \setbk{\sk_i^{0}}_{i\in[\ell_m]}$.
\item Generate and output $\tlsk \seteq \iO(\sfD_0[\tlcrs,\sk_P^0])$, where circuit $\sfD_0$ is described in~\cref{fig:L_dec_circuit}.
\end{enumerate}
\item[$\FakeCT(\pk,\aux,X):$]$ $
\begin{enumerate}
\item Parse $\pk= (\setbk{\abe.\pk_{i,b}}_{i\in[\ell_m],b\in\bit},\tlcrs)$ and $\aux = (\pk,\td, \setbk{\abe.\msk_{i,b}}_{i\in[\ell_m],b\in\bit},z^\ast)$.
\item Compute $\ct^\ast_{i,z^\ast [i]}\lrun \ABE.\Enc(\abe.\pk_{i,z^\ast [i]},X,0)$ and $\ct^\ast_{i,1-z^\ast [i]}\lrun \ABE.\Enc(\pk_{i,1-z^\ast [i]},X,1)$ for every $i\in[\ell_m]$.
\item Compute $\tlpi \lrun \Sim_2(\tlcrs,\td,d^\ast)$ where $d^\ast = (\setbk{(\abe.\pk_{i,0},\abe.\pk_{i,1},\ct^\ast_{i,0},\ct^\ast_{i,1})}_{i\in[\ell_m]},X)$.
\item Outputs $\tlct_X \seteq (\{\ct^\ast_{i,b}\}_{i\in[\ell_m],b\in\bit},\tlpi)$.
\end{enumerate}
\item[$\Reveal(\pk,\aux,\tlct_X,m):$]$ $
\begin{enumerate}
\item Parse $\aux =(\pk,\td, \setbk{\abe.\msk_{i,b}}_{i\in[\ell_m],b\in\bit},z^\ast)$.
\item Outputs $\tlmsk \seteq (\pk,\setbk{\abe.\msk_{i,z^*[i]\xor m[i]}}_{i\in[\ell_m]},z^\ast \xor m)$.
\end{enumerate}
\end{description}

\protocol{Left-or-Right Decryption Circuit $\sfD$
}
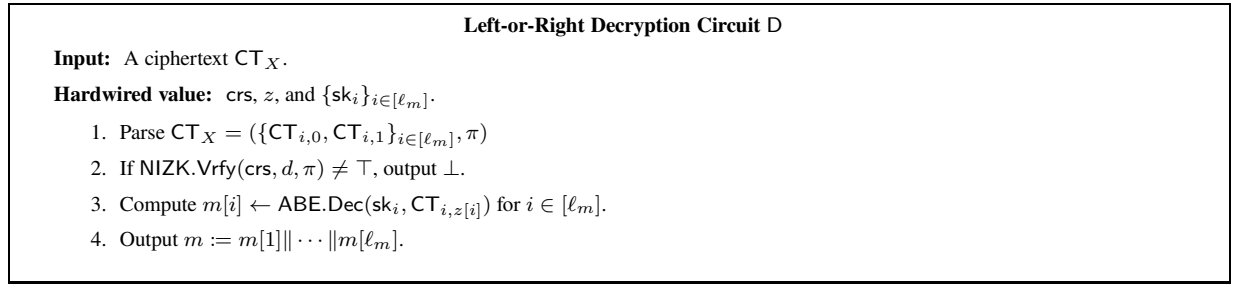
{The description of the left-or-right decryption circuit}
{fig:LR_dec_circuit}
{
\begin{description}
\item[Input:] A ciphertext $\ct_X $.
\item[Hardwired value:] $\crs$, $z$, and $\setbk{\sk_i}_{i\in[\ell_m]}$.
\end{description}
\begin{enumerate}
	\item Parse $\ct_X = (\setbk{\ct_{i,0},\ct_{i,1}}_{i\in[\ell_m]},\pi)$
\item If $\NIZK.\Vrfy(\crs,d,\pi)\ne \top$, output $\bot$.
\item Compute $m[i] \lrun \ABE.\Dec(\sk_{i},\ct_{i,z[i]})$ for $i\in [\ell_m]$.
\item Output $m \seteq m[1]\concat \cdots \concat m[\ell_m]$.
\end{enumerate}
}

\protocol{Left Decryption Circuit $\sfD_0$
}
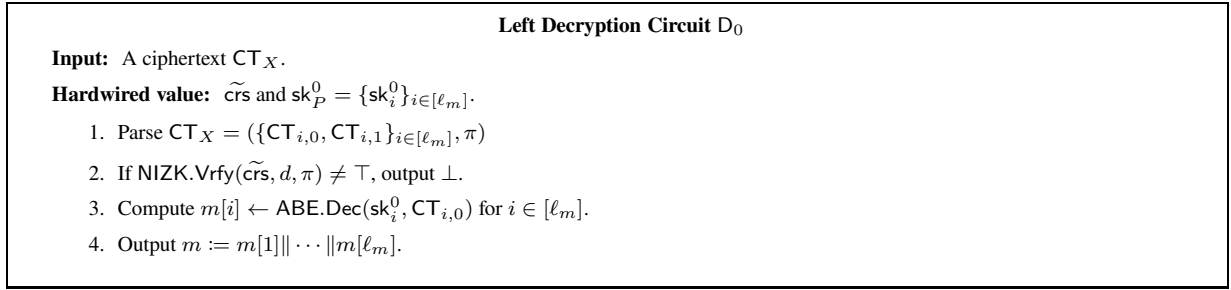
{The description of the left decryption circuit}
{fig:L_dec_circuit}
{
\begin{description}
\item[Input:] A ciphertext $\ct_X $.
\item[Hardwired value:] $\tlcrs$  and $\sk_P^{0}= \setbk{\sk_{i}^{0}}_{i\in[\ell_m]}$.
\end{description}
\begin{enumerate}
	\item Parse $\ct_X = (\setbk{\ct_{i,0},\ct_{i,1}}_{i\in[\ell_m]},\pi)$
\item If $\NIZK.\Vrfy(\tlcrs,d,\pi)\ne \top$, output $\bot$.
\item Compute $m[i] \lrun \ABE.\Dec(\sk_{i}^{0},\ct_{i,0})$ for $i\in [\ell_m]$.
\item Output $m \seteq m[1]\concat \cdots \concat m[\ell_m]$.
\end{enumerate}
}

\paragraph{Correctness.}
Correctness of $\Sigma_{\nce}$ easily follows from correctness of $\Sigma_{\abe}$ and completeness of $\Pi_\nizk$.

\paragraph{Security.}
We prove the following theorem.
\begin{theorem}\label{thm:ncabe_from_abe_io}
If $\Sigma_{\abe}$ is perfectly correct and IND-CPA secure, $\iO$ is secure IO for $\Ppoly$, and $\Pi_{\nizk}$ is a NIZK proof system for $\NP$, $\Sigma_{\nce}$ is RNC secure NCABE.
\end{theorem}

\begin{proof}
Let $\cA$ be a QPT adversary. 
We define the following sequence of hybrid games.
\begin{itemize}
\item $\hybi{0}$: This is the same as $\expa{\Sigma_\nce,\cA}{rnc}{cpa}(\secp,0)$. 
Let $X^\ast$  and $m$ be the target attribute and message, respectively, as in~\cref{def:ncabe_security}. 
\item $\hybi{1}$: This is the same as $\hybi{0}$ except that the challenger uses the circuit $\sfD_0[\tlcrs,\sk_P^0]$  instead of $\sfD[\crs,\setbk{\sk_i}_{i\in[\ell_m]},z]$ to generate secret keys for key queries. That is, it returns $\tlsk = \iO(\sfD_0[\tlcrs,\sk_P^0])$ instead of $\sk =\iO (\sfD[\crs,\setbk{\sk_i}_{i\in[\ell_m]},z])$. This change is indistinguishable by the IO security and statistical soundness of $\Pi_{\nizk}$. Note that secret keys do not depend on $z$ in this game. 
\item $\hybi{2}$: This is the same as $\hybi{1}$ except that the challenger 
generates 
a common reference string and proof of the NIZK by using simulators.       
That is, it generates 
$(\tlcrs,\td)\lrun \Sim_1(1^\secp)$ and
$\tlpi \lrun \Sim_2(\tlcrs,\td,d^\ast)$ where  $d^\ast=(\setbk{\abe.\pk_{i,0},\abe.\pk_{i,1},\ct^\ast_{i,0},\ct^\ast_{i,1}}_{i\in[\ell_m]},X^\ast)$. 
This change is indistinguishable by the computational zero-knowledge property of $\Pi_{\nizk}$.
\item $\hybi{3}$: This is the same as $\hybi{2}$ except that the challenger generates an inconsistent target ciphertext. That is, it generates $\ct^\ast_{i,1-z[i]}\lrun \ABE.\Enc(\abe.\pk_{i,1-z[i]},X^\ast,1-m[i])$ and $\ct^\ast_{i,z[i]}\lrun \ABE.\Enc(\abe.\pk_{i,z[i]},X^\ast,m[i])$ instead of double encryption of $m[i]$ for all $i$. Note that the NIZK proof in the target ciphertext is generated by the simulator in this game.
\item $\hybi{4}$: This is the same as $\hybi{3}$ except that the challenger chooses $z^*\chosen \zo{\ell_m}$, computes $\ct_{i,z[i]^\ast}^\ast \lrun \ABE.\Enc(\abe.\pk_{i,z^\ast[i]},X^\ast,0)$ and  $\ct_{i,1-z^\ast [i]}^\ast \lrun \ABE.\Enc(\abe.\pk_{i,1-z^\ast[i]},X^\ast,1)$, and sets $\tlmsk \seteq (z^\ast \xor m,\allowbreak\setbk{\msk_{i,z[i]^\ast \xor m[i]}}_{i\in[\ell_m]})$ as a master secret key. 
\end{itemize}
We prove~\cref{prop:ncabe_nizk_czk,prop:ncabe_hybrid_sk,prop:ncabe_abe_ind,prop:ncabe_conceptual}.
\begin{proposition}\label{prop:ncabe_hybrid_sk}
If 
$\Pi_{\abe}$ is perfectly correct, 
$\Pi_{\nizk}$ is statistically sound, and $\iO$ is secure, $\abs{\Pr[\hybi{0}=1] - \Pr[\hybi{1}=1]}\le  \negl(\secp)$.
\end{proposition}
\begin{proof}
We define more hybrid games. Let $q$ be the total number of key queries.
\begin{description}
\item [$\hybij{0}{j}$:] This is the same as $\hybi{0}$ except that
\begin{itemize}
\item for $j < k \le q$, the challenger generates $\sk = \iO(\sfD[\crs,\setbk{\sk_i}_{i\in[\ell_m]},z])$ for the $k$-th key query.
\item for $1\le k \le j$, the challenger generates $\tlsk =\iO(\sfD_0[\tlcrs,\sk_P^0])$ for the $k$-th key query.
\end{itemize}
Clearly, $\hybij{0}{0} = \hybi{0}$ and $\hybij{0}{q}=\hybi{1}$.
\end{description}
Let $\sfinvalid$ be an event that there exists $d^\dagger\notin \Lang$ and $\pi^\dagger$ such that $\NIZK.\Vrfy(\crs,d^\dagger,\pi^\dagger)=\top$.
By statistical soundness of $\Pi_{\nizk}$, this happens with negligible probability.
 If $\sfinvalid$ does not occur,
 by the definitions of $\sfD$ and $\sfD_0$ and perfect correctness of $\Pi_{\abe}$, their functionalities are equivalent for all inputs 
 Therefore, adversary's distinguishing advantage of obfuscation of these two circuits is negligible by the iO security.
The difference between $\hybij{1}{j-1}$ and $\hybij{1}{j}$ is that the $j$-th key query answer is generated by $\sfD_0$ instead of $\sfD$.
Therefore we have $\abs{\Pr[\hybij{i}{j-1}=1] - \Pr[\hybij{i}{j}=1]} \le \negl(\secp)$.
By a standard hybrid argument, we obtain~\cref{prop:ncabe_hybrid_sk}.
\end{proof}
\begin{proposition}\label{prop:ncabe_nizk_czk}
If $\Pi_{\nizk}$ is computationally zero-knowledge, 
$\abs{\Pr[\hybi{1}=1] - \Pr[\hybi{2}=1]}\le \negl(\secp)$.  
\end{proposition}
\begin{proof}   
The only difference of these two games is how to generate the common reference string and  proof of NIZK.
Thus, there is a straightforward reduction to the zero-knowledge property of $\Pi_{\nizk}$.
Specifically, 
we assume that $\abs{\Pr[\hybi{1}=1] - \Pr[\hybi{2}=1]}$ is non-negligible and 
construct an adversary $\cB_\nizk$ that breaks the zero-knowledge property of $\Pi_\nizk$.

Let $d = (\setbk{(\abe.\pk_{i,0},\abe.\pk_{i,1},\ct_{i,0},\ct_{i,1})}_{i\in[\ell_m]},X^\ast)$ and $w = (m,\setbk{r_{i,0},r_{i,1}}_{i\in[\ell_m]})$ be as in $\hybi{1}$ (or $\hybi{2}$, equivalently). 
$\cB_\nizk$ is given $(\crs^\ast,\pi^\ast)$ which is generated by the real setup and proving algorithms or simulators. 
$\cB_\nizk$ runs $\hybi{1}$ for $\cA$ while embedding   $(\crs^\ast,\pi^\ast)$ in the appropriate part. 
Finally, $\cB_\nizk$ outputs whatever $\cA$ outputs. 

\begin{itemize}
\item If $(\crs^\ast,\pi^\ast)$ is the real one, i.e., it is generated by $\crs^\ast \lrun \NIZK.\Setup(1^\secp)$ and $\pi^\ast \lrun \NIZK.\Prove(\crs,d,w)$, $\cB_\nizk$ perfectly simulates $\hybi{1}$. 
\item If $(\crs^\ast,\pi^\ast)$ is the simulated one, i.e., it is generated by $(\crs^\ast,\td) \lrun \Sim_1(1^\secp)$ and $\pi^\ast \lrun \Sim_2(\crs^\ast,\td,d)$, $\cB_\nizk$ perfectly simulates $\hybi{1}$.
\end{itemize}
Thus, if $\cA$ distinguishes these two hybrids, $\cB_\nizk$ breaks the computational zero-knowledge property of $\Pi_{\nizk}$. This completes the proof.
\end{proof}

\begin{proposition}\label{prop:ncabe_abe_ind}
If $\Pi_\abe$ is IND-CPA secure, 
$\abs{\Pr[\hybi{2}=1] - \Pr[\hybi{3}=1]}\le \negl(\secp)$.
\end{proposition}
\begin{proof}
We define more hybrid games. Recall $\ell_m$ is the length of plaintexts.
\begin{description}
\item [$\hybij{2}{j}$:] This is the same as $\hybi{2}$ except that
\begin{itemize}
\item for $j < i \le \ell_m$, the challenger generates $\ct_{i,b}^\ast \lrun \ABE.\Enc(\abe.\pk_{i,b},X^\ast,m[i])$ for $b\in\zo{}$.
\item for $1\le i \le j$, the challenger generates $\ct_{i,1- z[i]}^\ast \lrun \ABE.\Enc(\abe.\pk_{i,1-z[i]},X^\ast,1-m[i])$ and $\ct^\ast_{i,z[i]}\lrun \ABE.\Enc(\abe.\pk_{i,z[i]},m[i])$.
\end{itemize}
Clearly, $\hybij{2}{0} = \hybi{2}$ and $\hybij{2}{\ell_m} = \hybi{3}$.
\end{description}
The difference between $\hybij{2}{j}$ and $\hybij{2}{j-1}$ is the $j$-th component of the target ciphertext is valid or invalid.
We can show that this is indistinguishable 
by observing that the master secret key is set to be $\msk\seteq (\setbk{\abe.\msk_{i,z[i]}}_{i\in[\ell_m]},z)$ in these games and  $\setbk{\msk_{j,1-z[i]}}_{i\in [\ell_m]}$ is never revealed to the adversary.
Specifically, we can construct an adversary $\cB_\abe$ that breaks IND-CPA security of $\Sigma_{\abe}$ under key $\abe.\pk_{j,1-z[j]}$ assuming that $\cA$ distinguishes these two games.

$\cB_\abe$  receives $\abe.\pk$ from the challenger of $\expa{\Sigma_\abe,\cB_\abe}{ind}{cpa}(\secp,b')$ for $b'\in \bit$, and sets $\abe.\pk_{j,1-z[j]}\seteq \abe.\pk$.
For other public keys (that is, $\setbk{\abe.\pk_{i,b}}_{i\in[\ell_m],b\in\bit}\setminus \setbk{\abe.\pk_{j,1-z[j]}}$) and $\crs$, $\cB_\abe$ generates them by itself. $\cB_\abe$ sends $\pk \seteq (\setbk{\abe.\pk_{i,b}}_{i\in[\ell_m],b\in\bit},\crs)$ to $\cA$. When the distinguisher sends a key query $P$, $\cB_\abe$ passes $P$ to the challenger and receives $\sk_{j,1-z[j]} \lrun \ABE.\keygen(\msk_{j,1-z[j]},P)$.\footnote{In fact, $\cB_\abe$ need not query the challenger when $z[j]=0$ since $\sk_{j,1}$ is not needed for generating $\sfD_0[\tlcrs,\sk_P^0]$.} For other secret keys for $P$ (that is, $\setbk{\sk_{i,b}\lrun \ABE.\keygen(\msk_{i,b},P)}_{i\in[\ell_m],b\in\bit}\setminus \setbk{\sk_{j,1-z[j]}}$), $\cB_\abe$ generates them by itself since it has $\setbk{\msk_{i,b}}_{i\in[\ell_m],b\in\bit}$ except for $\msk_{j,1-z[j]}$. Thus, $\cB_\abe$ can compute $\tlsk = \iO(\sfD_0[\tlcrs,\sk_P^0])$.
At some point, 
$\cA$ declares target attribute $X^\ast$ and message $m$. 
$\cB_\abe$ sends $(m[j],1-m[j])$ to the challenger and receives $\ct_{j,1-z[j]}^\ast$. For $(i,b)\in[\ell_m]\times \zo{}\setminus (j,1-z[j])$, $\cB_\abe$ generates $\ct_{j,z[j]}^\ast \lrun \ABE.\Enc(\abe.\pk_{j,z[j]},X^\ast,m[j])$ and $\setbk{\ct_{i,b}}_{i \in [\ell_m]\setminus\setbk{j},b\in\zo{}}$ as in $\hybij{2}{j}$ and $\hybij{2}{j-1}$. Note that the difference between two games is the $j$-th component (and in particular $(j,1-z[j])$ part) of the target ciphertext. Again, $\cB_\abe$ simulates answers for secret key queries as above. $\cB_\abe$ outputs whatever $\cA$ outputs.
\begin{itemize}
\item If $b'=0$, i.e., $\ct_{j,1-z[j]}^\ast \lrun \ABE.\Enc(\abe.\pk_{j,1-z[j]},X^\ast,m^\ast[j])$, $\cB_\abe$ perfectly simulates $\hybij{i}{j-1}$.
\item If $b'=1$, i.e., $\ct_{j,1-z[j]}^\ast \lrun \ABE.\Enc(\abe.\pk_{j,1-z[j]},X^\ast,1-m^\ast[j])$, $\cB_\abe$ perfectly simulates $\hybij{i}{j}$.
\end{itemize}
Thus, if $\cA$ distinguishes these two games, $\cB_\abe$ breaks IND-CPA security of $\Sigma_{\abe}$. This completes the proof.
\end{proof}
\begin{proposition}\label{prop:ncabe_conceptual}
$\Pr[\hybi{3}=1]=\Pr[\hybi{4}=1]$.
\end{proposition}
\begin{proof}
This is a conceptual change. The advantage of distinguishing these two games is $0$ since we can see that these two games are identical if we set $z\seteq z^\ast \xor m
^\ast$. Note that secret keys do not depend on $z$ in these games.
\end{proof}
Clearly, $\hybi{4}=\expa{\Sigma_\nce,\cA}{rnc}{cpa}(\secp,1)$. Therefore, we complete the proof by~\cref{prop:ncabe_nizk_czk,prop:ncabe_hybrid_sk,prop:ncabe_abe_ind,prop:ncabe_conceptual}.
\end{proof}

\subsection{ABE with Certified Deletion from NCABE and SKE with Certified Deletion}\label{sec:const_abe_cd_from_sk}
In this section, we construct ABE with certified deletion from NCABE and OT-CD secure SKE with certified deletion.
\paragraph{Our ABE with certified deletion scheme.}
We construct an ABE with certified deletion scheme $\Sigma_{\mathsf{cd}} =(\Setup,\keygen,\allowbreak\Enc,\Dec,\Delete,\Vrfy)$ with plaintext space $\Ms$, attribute space $\Xs$, and policy space $\Ps$ from an NCABE scheme $\Sigma_{\nce}=\NCE.(\Setup,\keygen,\Enc,\Dec,\FakeSetup,\FakeSK,\FakeCT,\Reveal)$ with plaintext space $\bit^{\ell}$, attribute space $\Xs$, and policy space $\Ps$ and an SKE with certified deletion scheme $\Sigma_{\skcd}=\SKE.(\Gen,\Enc,\Dec,\Delete,\Vrfy)$ with plaintext space $\Ms$ and key space $\bit^\ell$.

\begin{description}
	\item[$\Setup(1^\secp)$:]$ $
	\begin{itemize}
	\item Generate $(\nce.\pk,\nce.\msk)\lrun \NCE.\Setup(1^\secp)$.
	\item Output $(\pk,\msk) \seteq (\nce.\pk,\nce.\msk)$.
	\end{itemize}
\item[$\keygen(\msk,P)$:] $ $
\begin{itemize}
\item Generate $\nce.\sk_P \lrun \NCE.\keygen(\nce.\msk,P)$ and output $\sk_P \seteq \nce.\sk_P$.
\end{itemize}
\item[$\Enc(\pk,X,m)$:] $ $
\begin{itemize}
	\item Parse $\pk = \nce.\pk$.
\item Generate $\ske.\sk \lrun \SKE.\Gen(1^\secp)$.
\item Compute $\nce.\ct_X \lrun \NCE.\Enc(\nce.\pk,X,\ske.\sk)$ and $\ske.\ct \lrun \SKE.\Enc(\ske.\sk,m)$.
\item Output $\ct_X \seteq (\nce.\ct_X,\ske.\ct)$ and $\vk \seteq \ske.\sk$.
\end{itemize}
\item[$\Dec(\sk_P,\ct_X)$:] $ $
\begin{itemize}
\item Parse $\sk_P = \nce.\sk_P$ and $\ct_X = (\nce.\ct_X,\ske.\ct)$.
\item Compute $\sk^\prime \lrun \NCE.\Dec(\nce.\sk_P,\nce.\ct_X)$.
\item Compute and output $m^\prime \lrun \SKE.\Dec(\sk^\prime,\ske.\ct)$.
\end{itemize}
\item[$\Delete(\ct)$:] $ $
\begin{itemize}
\item Parse $\ct_X = (\nce.\ct_X,\ske.\ct)$.
\item Generate $\ske.\cert \lrun \SKE.\Delete(\ske.\ct)$.
\item Output $\cert \seteq \ske.\cert$.
\end{itemize}
\item[$\Vrfy(\vk,\cert)$:] $ $
\begin{itemize}
\item Parse $\vk = \ske.\sk$ and $\cert = \ske.\cert$.
\item Output $b \lrun \SKE.\Vrfy(\ske.\sk,\ske.\cert)$.
\end{itemize}
\end{description}

\paragraph{Correctness.} Correctness easily follows from correctness of $\Sigma_\skcd$ and $\Sigma_\nce$.
\begin{theorem}\label{thm:abe_cd_from_sk_cd_and_ncabe}
If $\Sigma_{\nce}$ is RNC secure ABE and $\Sigma_{\skcd}$ is OT-CD secure, $\Sigma_{\mathsf{cd}}$ is IND-CPA-CD secure ABE.
\end{theorem}

\begin{proof}
Let $\cA$ be a QPT adversary and $b$ be a bit. 
We define the following hybrid game $\hybi{}(b)$.
\begin{description}
\item[$\hybi{}(b)$:] This is the same as $\expb{\Sigma_{\mathsf{cd},\cA}}{ind}{cpa}{cd}(\secp,b)$ except for the following  differences: 
\begin{enumerate}
\item The challenger generates the public key as $(\nce.\pk,\nce.\aux)\lrun \NCE.\FakeSetup(1^\secp)$. 
\item The challenger generates the challenge ciphertext as follows. It generates $\ske.\sk \lrun \SKE.\Gen(1^\secp)$, $\nce.\ct_{X^\ast}^\ast \lrun \NCE.\Fake(\nce.\pk,\nce.\aux,X^\ast)$, and $\ske.\ct^\ast \lrun \SKE.\Enc(\ske.\sk,m_b)$. The challenge ciphertext is $\ct_{X^\ast}^\ast \seteq (\nce.\ct_{X^\ast}^\ast,\ske.\ct^\ast)$. 
\item The challenger generates secret keys as  $\nce.\tlsk_P \lrun \NCE.\FakeSK(\nce.\pk,\allowbreak\nce.\aux,P)$. 
\item The challenger reveals $\nce.\tlmsk \lrun \Reveal(\nce.\pk,\nce.\aux,\nce.\ct_X^\ast,\ske.\sk)$ instead of $\nce.\msk$.
\end{enumerate}
\end{description}

\begin{proposition}\label{prop:abecd_hyb_one}
If $\Sigma_{\nce}$ is RNC secure,
$\abs{\Pr[\expb{\Sigma_{\mathsf{cd}},\cA}{ind}{cpa}{cd}(\secp, b)=1] - \Pr[\sfhyb{}{}(b)=1]} \le \negl(\secp)$.
\end{proposition}
\begin{proof}
We construct an adversary $\cB_\nce$ that breaks the RNC security of $\Sigma_{\nce}$ by using $\cA$ that distinguishes them.

$\cB_\nce$ receives $\nce.\pk$ from the challenger of $\expa{\Sigma_\nce,\cB_\nce}{rnc}{cpa}(\secp, b')$ for $b'\in \bit$ and sends $\nce.\pk$ to $\cA$. 
When $\cA$ makes a key query $P$, $\cB_\nce$  passes $P$ to the challenger, receives $\sk_P$, and passes it to $\cA$.
At some point, $\cA$ sends the target attribute $X^\ast$ and messages $(m_0,m_1)$.
$\cB_\nce$ generates $\ske.\sk \lrun \SKE.\Gen(1^\secp)$, sends the target attribute $X^\ast$ and message $\ske.\sk$ to the challenger, and receives $\nce.\ct_{X^\ast}^\ast$ and $\nce.\msk^\ast$ from the challenger.
$\cB_\nce$ generates $\ske.\ct \lrun \SKE.\Enc(\ske.\sk,m_b)$
and 
sends $\ct_{X^\ast}^\ast\seteq (\nce.\ct_{X^\ast}^\ast,\ske.\ct)$  to $\cA$.
Again, when $\cA$ makes a key query $P$, $\cB_\nce$ responds similarly as above. 
At some point, $\cA$ sends a certificate $\cert$. 
If $\SKE.\Vrfy(\ske.\sk,\cert) = \top$, $\cB_\nce$ sends $\nce.\msk^\ast$ to $\cA$.
Otherwise, $\cB_\nce$ sends $\bot$ to $\cA$. 
Finally, $\cB_\nce$ outputs whatever $\cA$ outputs.

We can see that $\cB_\nce$ perfectly simulates $\expb{\Sigma_{\mathsf{cd}},\cA}{ind}{cpa}{cd}(\secp, b)$ if $b'=0$ and  $\sfhyb{}{}(b)$ if $b'=1$. 
Thus, if $\cA$ distinguishes the two hybrids, $\cB_\nce$ breaks the RNC security. This completes the proof.
\end{proof}
\begin{proposition}\label{prop:abecd_hyb_end}
If $\Sigma_{\skcd}$ is OT-CD secure,  
$\abs{\Pr[\sfhyb{}{}(0)=1] - \Pr[\sfhyb{}{}(1)=1] } \le \negl(\secp)$.
\end{proposition}
\begin{proof}
We construct an adversary $\cB_\skcd$ that breaks the OT-CD security of $\Sigma_{\skcd}$ assuming that $\cA$ distinguishes these two hybrids.

$\cB_\skcd$ plays the experiment $\expb{\Sigma_\skcd,\cB_\skcd}{otsk}{cert}{del}(\secp,b')$ for some $b'\in \bit$.
$\cB_\skcd$ generates $(\nce.\pk,\nce.\aux) \lrun \NCE.\Setup(1^\secp)$ and sends $\nce.\pk$ to $\cA$.
When $\cA$ sends a key query $P$, $\cB_\skcd$ generates $\tlsk_P \lrun \NCE.\FakeSK(\nce.\pk,\nce.\aux,P)$ and returns it to $\cA$.
When $\cA$ sends 
the target attribute $X^\ast$ and messages
$(m_0,m_1)$, $\cB_\skcd$ sends $(m_0,m_1)$ to the challenger of $\expb{\Sigma_\skcd,\cB_\skcd}{otsk}{cert}{del}(\secp,b')$, receives $\ske.\ct^\ast$ and generates $\nce.\tlct_X \lrun \NCE.\FakeCT(\nce.\pk,\nce.\aux,X^\ast)$. $\cB_\skcd$ sends $(\nce.\tlct_X,\ske.\ct^\ast)$ to $\cA$ as the challenge ciphertext.
Again, when $\cA$ makes a key query $P$, $\cB_\nce$ responds similarly as above.   
At some point, $\cA$ outputs $\cert$. $\cB_\skcd$ passes $\cert$ to the challenger of $\expb{\Sigma_\skcd,\cB_\skcd}{otsk}{cert}{del}(\secp,b')$. If the challenger returns $\ske.\sk$, $\cB_\skcd$ generates $\tlmsk \lrun \NCE.\Reveal(\nce.\pk,\nce.\msk,\nce.\aux,\nce.\tlct_X,\ske.\sk)$ and sends $\tlmsk$ to $\cA$.
Otherwise, $\cB_\skcd$ sends $\bot$ to $\cA$.  Finally, $\cB_\skcd$ outputs whatever $\cA$ outputs. 

\begin{itemize}
	\item If 
	$b'=0$, i.e., 
	$\ske.\ct^\ast = \SKE.\Enc(\ske.\sk,m_0)$, $\cB_\skcd$ perfectly simulates $\hybi{}(0)$.
	\item If 
	$b'=1$, i.e., 
	$\ske.\ct^\ast = \SKE.\Enc(\ske.\sk,m_1)$, $\cB_\skcd$ perfectly simulates $\hybi{}(1)$.
\end{itemize}
Thus, if $\cA$ distinguishes the two hybrids, $\cB_\skcd$ breaks the OT-CD security of $\Sigma_\skcd$. This completes the proof.
\end{proof}

By~\cref{prop:abecd_hyb_one,prop:abecd_hyb_end}, we immediately obtain~\cref{thm:pke_cd_from_sk_cd_and_pke}.
\end{proof}

\paragraph{Summary of this section.}
Since IO and OWFs imply computational NIZK proof for $\NP$~\cite{TCC:BitPan15} and IND-CPA secure ABE for circuits, we immediately obtain the following corollary by using~\cref{thm:ncabe_from_abe_io,thm:abe_cd_from_sk_cd_and_ncabe,thm:ske_cert_del_no_assumption,thm:ABE_circuits_from_LWE_or_IO}.

\begin{corollary}\label{cor:existence_abe_cd}
If there exist secure IO for $\Ppoly$ and OWFs against QPT adversaries, there exists ABE with certified deletion for circuits.
\end{corollary}


\section{PKE with Certified Deletion and Classical Communication}\label{sec:classical_com}

In this section, we define the notion of public key encryption with certified deletion with classical
communication, and construct it from the LWE assumption \revise{in the QROM}.
In \cref{sec:pk_cd_def_classical_com}, we present the definition of the public key encryption
with certified deletion with classical communication. In \cref{sec:cut_and_choose_adaptive_hardcore}, we introduce what we
call the {\it cut-and-choose adaptive hardcore property}, which is used in the security proof
of the PKE with certified deletion with classical communication.
In \cref{sec:PKE_cd_cc_construction}, we construct a PKE with certified deletion with classical communication, and show
its security.

\subsection{Definition of PKE with Certified Deletion and Classical Communication}\label{sec:pk_cd_def_classical_com}
We define PKE with certified deletion with classical communication.
Note that the encryption algorithm of a PKE with certified deletion with classical communication is interactive unlike PKE with certified deletion (with quantum communication) as defined in \cref{def:pk_cert_del}.
It is easy to see that the interaction is necessary if we only allow classical communication.
\begin{definition}[PKE with Certified Deletion with Classical Communication (Syntax)]\label{def:pk_cert_del_classical_com}
A public key encryption scheme with certified deletion with classical communication 
is a tuple of quantum algorithms $(\keygen,\Enc,\Dec,\Delete,\Vrfy)$ with plaintext space $\Ms$.
\begin{description}
    \item[$\keygen (1^\secp) \ra (\pk,\sk)$:] The key generation algorithm takes as input the security parameter $1^\secp$ and outputs a classical key pair $(\pk,\sk)$.
    \item[$\Enc\langle \sen(\pk,m),\rec\rangle \ra (\vk,\ct)$:] 
    This is an interactive process between a classical sender $\sen$ with input $\pk$ and a plaintext $m\in\cM$, and a quantum receiver $\rec$ without input.
    After exchanging classical messages, 
    $\sen$ outputs a classical verification key $\vk$ and $\rec$ outputs a quantum ciphertext $\ct$. 
    \item[$\Dec(\sk,\ct) \ra m^\prime \mbox{ or } \bot$:] The decryption algorithm takes as input the secret key $\sk$ and the ciphertext $\ct$, and outputs a plaintext $m^\prime$ or $\bot$.
    \item[$\Delete(\ct) \ra \cert$:] The deletion algorithm takes as input the ciphertext $\ct$ and outputs a classical certificate $\cert$.
    \item[$\Vrfy(\vk,\cert)\ra \top \mbox{ or }\bot$:] The verification algorithm takes the verification key $\vk$ and the certificate $\ct$, and outputs $\top$ or $\bot$.
\end{description}
\end{definition}

\begin{definition}[Correctness for PKE with Certified Deletion with Classical Communication]\label{def:pk_cd_correctness_classical_com}
There are two types of correctness. One is decryption correctness and the other is verification correctness.
\begin{description}
\item[Decryption correctness:] For any $\secp\in \N$, $m\in\Ms$, 
\begin{align}
\Pr\left[
\Dec(\sk,\ct)\ne m
\ \middle |
\begin{array}{ll}
(\pk,\sk)\lrun \keygen(1^\secp)\\
(\vk,\ct) \lrun \Enc\langle \sen(\pk,m),\rec\rangle
\end{array}
\right] 
\le\negl(\secp).
\end{align}

\item[Verification correctness:] For any $\secp\in \N$, $m\in\Ms$, 
\begin{align}
\Pr\left[
\Vrfy(\vk,\cert)=\bot
\ \middle |
\begin{array}{ll}
(\pk,\sk)\lrun \keygen(1^\secp)\\
(\vk,\ct) \lrun \Enc\langle \sen(\pk,m),\rec\rangle\\
\cert \lrun \Delete(\ct)
\end{array}
\right] 
\le\negl(\secp).
\end{align}

\end{description}
\end{definition}

\begin{definition}[Certified Deletion Security for PKE with Classical Communication]\label{def:pk_certified_del_classical_com}
Let $\Sigma=(\keygen, \Enc, \Dec, \Delete, \Vrfy)$ be a PKE with certified deletion scheme with classical communication.
We consider the following security experiment $\expb{\Sigma,\cA}{ccpk}{cert}{del}(\secp,b)$.

\begin{enumerate}
    \item The challenger computes $(\pk,\sk) \la \keygen(1^\secp)$ and sends $\pk$ to $\cA$.
    \item $\cA$ sends $(m_0,m_1)\in \Ms^2$ to the challenger.
    \item The challenger and $\cA$ jointly execute $(\vk_b,\ct_b) \la \Enc\langle\sen(\pk,m_b),\cA(\pk)\rangle$ where the challenger plays the role of the sender and $\cA$ plays the role of the receiver.
    \item At some point, $\cA$ sends $\cert$ to the challenger.
    \item The challenger computes $\Vrfy(\vk_b,\cert)$. If the output is $\bot$, the challenger
    sends $\bot$ to $\cA$. If the output is $\top$, the challenger sends $\sk$ to $\cA$.
    \item $\cA$ outputs its guess $b'\in \bit$.
\end{enumerate}
Let $\advc{\Sigma,\cA}{ccpk}{cert}{del}(\secp)$ be the advantage of the experiment above.
We say that the $\Sigma$ is IND-CPA-CD secure if for any QPT adversary $\cA$, it holds that
\begin{align}
\advc{\Sigma,\cA}{ccpk}{cert}{del}(\secp)\seteq \abs{\Pr[ \expb{\Sigma,\cA}{ccpk}{cert}{del}(\secp, 0)=1] - \Pr[ \expb{\Sigma,\cA}{ccpk}{cert}{del}(\secp, 1)=1] }\leq \negl(\secp).
\end{align}
\end{definition}

\subsection{Preparation}
\label{sec:cut_and_choose_adaptive_hardcore}
We prove that any injective invariant NTCF family satisfies a property which we call the \emph{cut-and-choose adaptive hardcore property}, which is used in the security proof of our PKE with certified deletion with classical communication. 
To prove the cut-and-choose adaptive hardcore property, we first prove a simple combinatorial lemma and its corollary.
\begin{lemma}\label{lem:subset_probability}
Let $n,k$ be a positive integers and $T\subseteq [2n]$ be a subset. 
Let $S\subseteq [2n]$ be a uniformly random subset conditioned on that $|S|=n$. 
If $k\leq |T| $, we have
\[
\Pr\left[S\cap T=\emptyset\right]\leq \left(\frac{1}{2}\right)^{k}.
\]
If $|T|\leq  k\leq n$, we have 
\[
\Pr\left[S\cap 
 T=\emptyset\right]> \left(\frac{n-k}{2n-k}\right)^{k}.
\]
\end{lemma}
\begin{proof}
When $|T|\geq n+1$, we have $\Pr\left[S\cap T=\emptyset\right]=0$ by the pigeonhole principle and thus the first inequality trivially holds. Note that we need not  consider such a case for the second inequality since we assume $|T|\leq n$ for the second inequality. 
In the following, we assume $|T|\leq n$. 
Let $k'\seteq |T|$.
Then we have 
\begin{align}
\Pr\left[S\cap T=\emptyset\right]=\frac{\binom{2n-k'}{n}}{\binom{2n}{n}}
=\frac{(2n-k')!}{(n-k')!n!}\cdot\frac{n!n!}{(2n)!}=\prod_{i=0}^{k'-1}\frac{n-i}{2n-i}.
\end{align}
For all $0\leq i\leq k'-1$, we have 
\begin{align}
   \frac{n-k'}{2n-k'} <\frac{n-i}{2n-i}\leq \frac{1}{2}.
\end{align}
Therefore we have 
\begin{align}
\left(\frac{n-k'}{2n-k'}\right)^{k'}<\Pr\left[S\cap T=\emptyset\right]\leq \left(\frac{1}{2}\right)^{k'}.
\end{align}
Then, \cref{lem:subset_probability} follows from
$\left(\frac{1}{2}\right)^{k'}\leq \left(\frac{1}{2}\right)^{k}$ for $k\leq k'$ and 
$\left(\frac{n-k}{2n-k}\right)^{k}\leq \left(\frac{n-k'}{2n-k'}\right)^{k'}$ for $k'\leq k \leq n$.
\end{proof}
\begin{corollary}\label{cor:subset_probability_two}
Let $n,k$ be a positive integers and $T\subseteq [4n]$ be a subset.
Let $S\subseteq [4n]$ be a uniformly random subset conditioned on that $|S|=2n$  
and $U\subseteq S$ be a uniformly random subset conditioned on that $|U|=n$. 
If $k\leq |T| $, we have
\[
\Pr\left[S\cap T=\emptyset\right]\leq \left(\frac{1}{2}\right)^{k}.
\]
For any $S^*\subseteq [4n]$ such that $|S^*|=2n$, 
if $|T|<k<n$, we have
\[
\Pr\left[U\cap T=\emptyset|S=S^*\right]> \left(\frac{n-k}{2n-k}\right)^{k}.
\]
\end{corollary}
\begin{proof}
The former part immediately follows from \cref{lem:subset_probability} by considering $2n$ as $n$ in  \cref{lem:subset_probability}.
For the latter part, suppose that  $S$ is fixed to be $S^*$. 
We have $U\cap T=U\cap (S^*\cap T)$ and $|S^*\cap T|\leq |T|<k<n$. 
By considering a one-to-one map from $S^*$ to $[2n]$, we can apply the latter statement of \cref{lem:subset_probability},  
where we think of the images of $U$ and $S^*\cap T$ as $S$ and $T$ in \cref{lem:subset_probability}, respectively, 
to obtain \cref{cor:subset_probability_two}. 
\end{proof}

Then we define the cut-and-choose adaptive hardcore property and prove that any injective invariant NTCF family satisfies it.
\begin{lemma}[Cut-and-Choose  Adaptive Hardcore Property]\label{lem:cut_and_choose_adaptive_hardcore}
Let $\mathcal{F}$ be an injective invariant NTCF family and $\mathcal{G}$ be the corresponding trapdoor injective family. 
Then $\mathcal{F}$ and $\mathcal{G}$ satisfy what we call the \emph{cut-and-choose  adaptive hardcore property} defined below.
For a QPT adversary $\cA$ and a positive integer $n$, 
we consider the following  experiment $\expb{(\mathcal{F},\mathcal{G}),\cA}{cut}{and}{choose}(\secp,n)$.

\begin{enumerate}
    \item The challenger chooses a uniform subset $S\subseteq [4n]$ such that $|S|=2n$.\footnote{We can also take $S\subseteq [2n]$ such that $|S|=n$, but we do as above just for convenience  in the proof.} 
    \item The challenger generates 
    $(\fk_i,\td_i)\lrun \GenG(1^\secp)$ for all $i\in S$ and 
     $(\fk_i,\td_i)\lrun \GenF(1^\secp)$ for all $i\in \overline{S}$ and
    sends $\{\fk_i\}_{i\in[4n]}$ to $\cA$. 
    \item $\cA$ sends 
    $\{y_i,d_i,e_i\}_{i\in[4n]}$ to the challenger. 
    \item The challenger
    computes $x_{i,\beta}\lrun \InvF(\td_i,\beta,y_i)$ for all $(i,\beta)\in \overline{S}\times \bit$ and
    checks if 
     $d_i\in G_{\fk_i,0,x_{i,0}} \cap G_{\fk_i,1,x_{i,1}}$ and
    $e_i=d_i\cdot (J(x_{i,0})\oplus J(x_{i,1}))$ hold for all $i\in \overline{S}$.
    If they do not hold for some $i\in \overline{S}$, the challenger immediately aborts and the experiment returns $0$. 
    \item  \label{step:reveal_S} 
    The challenger sends $S$ to $\cA$.
    \item $\cA$ sends $\{b_i,x_i\}_{i\in S}$ to the challenger. 
    \item 
    The challenger checks if $\Chk_{\cG}(\fk_i,b_i,x_i,y_i)=1$ holds for all $i\in S$.
    If this holds for all $i\in S$, the experiment returns $1$. Otherwise, it returns $0$.
\end{enumerate}
Then for any $n$ such that
$n\leq \poly(\secp)$ and 
$n=\omega(\log \secp)$, it holds that 
\begin{align}
\advc{(\mathcal{F},\mathcal{G}),\cA}{cut}{and}{choose}(\secp,n)\seteq \Pr[ \expb{(\mathcal{F},\mathcal{G}),\cA}{cut}{and}{choose}(\secp,n)=1] \leq \negl(\secp).
\end{align}
\end{lemma}
\begin{proof}
We consider a modified experiment $\wtexpb{\mathcal{F},\cA}{cut}{and}{choose}(\secp,n)$ that works similarly to $\expb{(\mathcal{F},\mathcal{G}),\cA}{cut}{and}{choose}(\secp,n)$ except that the challenger generates $(\fk_i,\td_i)$ by $\GenF$ for all $i\in[4n]$. 
Since the challenger in these experiments does not use $\td_i$ for $i\in S$ at all, the injective invariance implies 
\begin{align}
 \abs
 {\Pr[ \wtexpb{\mathcal{F},\cA}{cut}{and}{choose}(\secp,n)=1]
 -\Pr[ \expb{(\mathcal{F},\mathcal{G}),\cA}{cut}{and}{choose}(\secp,n)=1]} \leq \negl(\secp)
\end{align}
by a straightforward hybrid argument.
Therefore, it suffices to prove 
\begin{align}
\Pr[ \wtexpb{\mathcal{F},\cA}{cut}{and}{choose}(\secp,n)=1] \leq \negl(\secp).
\end{align}
In the following, we reduce this to the amplified adaptive hardcore property (\cref{lem:amplified_adaptive_hardcore}).
For the sake of contradiction, we assume that $\Pr[ \wtexpb{\mathcal{F},\cA}{cut}{and}{choose}(\secp,n)=1]$ is non-negligible. Then there exists a polynomial $p$ such that 
\begin{align}
\Pr[ \wtexpb{\mathcal{F},\cA}{cut}{and}{choose}(\secp,n)=1] \geq \frac{1}{p(\secp)}
\end{align}
for infinitely many $\secp\in \mathbb{N}$.
Let $k$ be the minimal integer such that 
\[
\left(\frac{1}{2}\right)^{k}\leq \frac{1}{2p(\secp)}.
\]
In $ \wtexpb{\mathcal{F},\cA}{cut}{and}{choose}(\secp,n)$, let $T\subseteq [4n]$ be the subset consisting of $i\in [4n]$  such that 
$d_i\notin G_{\fk_i,0,x_{i,0}} \cap G_{\fk_i,1,x_{i,1}}$ or
$e_i\neq d_i\cdot (J(x_{i,0})\oplus J(x_{i,1}))$  where 
$x_{i,\beta}\lrun \InvF(\td_i,\beta,y_i)$ for $\beta \in \bit$.\footnote{Note that $x_{i,\beta}$ was defined only for
$(i,\beta)\in \overline{S}\times \bit$ in $\wtexpb{\mathcal{F},\cA}{cut}{and}{choose}(\secp,n)$, but we naturally generalize it to all $(i,\beta)\in [4n]\times \bit$.
We remark that this is well-defined since the challenger  uses $\GenF$ on all positions in $\wtexpb{\mathcal{F},\cA}{cut}{and}{choose}(\secp,n)$  unlike in $\expb{(\mathcal{F},\mathcal{G}),\cA}{cut}{and}{choose}(\secp,n)$.} Let $\Bad$ be the event such that $|T|\geq k$. 
When $\wtexpb{\mathcal{F},\cA}{cut}{and}{choose}(\secp,n)$ returns $1$, we must have 
$\overline{S}\cap T= \emptyset$.
Therefore, by \cref{cor:subset_probability_two}, we have 
\begin{align}
&\Pr[\wtexpb{\mathcal{F},\cA}{cut}{and}{choose}(\secp,n)=1 \land \Bad]\\
&\leq  \Pr[\overline{S}\cap T= \emptyset \land \Bad]\\
&\leq \Pr[\overline{S}\cap T= \emptyset | \Bad]\\
&\leq
\left(\frac{1}{2}\right)^k\\
&\leq \frac{1}{2p(\secp)}. 
\end{align}
We remark that we can apply \cref{cor:subset_probability_two} to get the third inequality above since 
$|T|\geq k$ when $\Bad$ occurs and
no information of $S$ is given to $\A$ before Step \ref{step:reveal_S} in $\wtexpb{\mathcal{F},\cA}{cut}{and}{choose}(\secp,n)$ and thus $\overline{S}$ is independent of  $\Bad$.
On the other hand, we have 
\begin{align}
\Pr[ \wtexpb{\mathcal{F},\cA}{cut}{and}{choose}(\secp,n)=1] 
=\Pr[ \wtexpb{\mathcal{F},\cA}{cut}{and}{choose}(\secp,n)=1 \land \Bad] +\Pr[ \wtexpb{\mathcal{F},\cA}{cut}{and}{choose}(\secp,n)=1 \land \overline{\Bad}] 
\geq \frac{1}{p(\secp)}
\end{align}
for infinitely many $\secp$.
Therefore we have 
\begin{align}
\Pr[\wtexpb{\mathcal{F},\cA}{cut}{and}{choose}(\secp,n)=1 \land \overline{\Bad}]\geq  \frac{1}{2p(\secp)}
\end{align}
for infinitely many $\secp$. 
This naturally gives an adversary $\cB$ that breaks the amplified adaptive hardcore property that is described below.
\begin{description}
\item [$\cB(\{\fk^*_j\}_{j\in [n]})$:] Given a problem instance $\{\fk^*_j\}_{j\in [n]}$, it works as follows.
\begin{enumerate}
    \item Choose a uniform subset $S\subseteq [4n]$ such that $|S|=2n$.
    \item Choose a uniform subset $U\subseteq S$ such that $|U|=n$. 
    Let $i_1,...,i_n$ be the elements of $U$. 
    \item Set $\fk_{i_j}:=\fk^*_{j}$ for all $j\in [n]$ and generate $(\fk_i,\td_i)\lrun \GenF(1^\secp)$ for all $i\in \overline{U}$.
    \item Send $\{\fk_i\}_{i\in [4n]}$ to $\A$ and receives the response $\{y_i,d_i,e_i\}_{i\in[4n]}$ from $\A$. 
    \item Send $S$ to $\A$ and receives the response $\{b_i,x_i\}_{i\in S}$ from $\A$. 
    \item Output $\{b_i,x_i,y_i,d_i,e_i\}_{i\in U}$. 
\end{enumerate}
\end{description}
We can see that $\cB$ perfectly simulates $\wtexpb{\mathcal{F},\cA}{cut}{and}{choose}(\secp,n)$ for $\cA$, and $\cB$ wins the amplified adaptive hardcore game whenever we have $\wtexpb{\mathcal{F},\cA}{cut}{and}{choose}(\secp,n)=1$ and $U\cap T=\emptyset$
in the simulated experiment. 

We have  $k=O(\log \secp)$ by the definition of $k$ and $n=\omega(\log \secp)$ by the assumption. 
Therefore we have $k<\frac{n}{2}$ for sufficiently large $\secp$. 
Then, for sufficiently large $\secp$,  we have 
\begin{align}
    &\Pr[U\cap T=\emptyset|\wtexpb{\mathcal{F},\cA}{cut}{and}{choose}(\secp,n)=1 \land\overline{\Bad}]\\
    &=\frac{\Pr[U\cap T=\emptyset \land \wtexpb{\mathcal{F},\cA}{cut}{and}{choose}(\secp,n)=1 \land\overline{\Bad}]}{\Pr[\wtexpb{\mathcal{F},\cA}{cut}{and}{choose}(\secp,n)=1 \land\overline{\Bad}]}\\
    &=\sum_{S^*\subseteq [4n]\text{~s.t.~}|S^*|=2n}\sum_{T^*\subseteq [4n]\text{~s.t.~}|T^*|<k}\frac{\Pr[U\cap T^*=\emptyset \land \wtexpb{\mathcal{F},\cA}{cut}{and}{choose}(\secp,n)=1 \land\overline{\Bad}\land S=S^* \land T=T^*]}{\Pr[\wtexpb{\mathcal{F},\cA}{cut}{and}{choose}(\secp,n)=1 \land\overline{\Bad}]}\\
    &=\sum_{S^*\subseteq [4n]\text{~s.t.~}|S^*|=2n}\sum_{T^*\subseteq [4n]\text{~s.t.~}|T^*|<k}\frac{\Pr[U\cap T^*=\emptyset]\cdot\Pr[\wtexpb{\mathcal{F},\cA}{cut}{and}{choose}(\secp,n)=1 \land\overline{\Bad}\land S=S^* \land T=T^*]}{\Pr[\wtexpb{\mathcal{F},\cA}{cut}{and}{choose}(\secp,n)=1 \land\overline{\Bad}]}\\
    &> \sum_{S^*\subseteq [4n]\text{~s.t.~}|S^*|=2n}\sum_{T^*\subseteq [4n]\text{~s.t.~}|T^*|<k}\left(\frac{n-k}{2n-k}\right)^{k}\cdot \frac{\Pr[\wtexpb{\mathcal{F},\cA}{cut}{and}{choose}(\secp,n)=1 \land\overline{\Bad}\land S=S^* \land T=T^*]}{\Pr[\wtexpb{\mathcal{F},\cA}{cut}{and}{choose}(\secp,n)=1 \land\overline{\Bad}]}\\
    &= \left(\frac{n-k}{2n-k}\right)^{k}\\
    &\geq \frac{1}{\poly(\secp)}
\end{align}
where the first equality follows from the definition of conditional probability, 
the second and fourth equalities follow from the fact that $|T|<k$ when $\overline{\Bad}$ occurs, 
the third equality follows from the fact that $U$ is independent of the events 
$\wtexpb{\mathcal{F},\cA}{cut}{and}{choose}(\secp,n)=1$, $\overline{\Bad}$, and $T=T^*$ when we fix $S$,  
the first inequality follows  from \cref{cor:subset_probability_two}
and  $|T|<k<n/2<n$ for sufficiently large $\secp$,
and the second inequality follows from   $k<\frac{n}{2}$ for sufficiently large $\secp$, in which case we have $\frac{n-k}{2n-k}> \frac{1}{3}$, and $k=O(\log \secp)$.

Therefore we have 
\begin{align}
    \Pr[\cB\text{~wins}]
    &\geq \Pr[\wtexpb{\mathcal{F},\cA}{cut}{and}{choose}(\secp,n)=1 \land U\cap T=\emptyset]\\
    &\geq \Pr[\wtexpb{\mathcal{F},\cA}{cut}{and}{choose}(\secp,n)=1 \land U\cap T=\emptyset \land \overline{\Bad}]\\
    &=\Pr[\wtexpb{\mathcal{F},\cA}{cut}{and}{choose}(\secp,n)=1 \land \overline{\Bad}]\cdot \Pr[
    U\cap T=\emptyset
    |\wtexpb{\mathcal{F},\cA}{cut}{and}{choose}(\secp,n)=1 \land \overline{\Bad}]\\
    &\geq \frac{1}{\poly(\secp)}
\end{align}
for infinitely many $\secp$.
This contradicts the amplified adaptive hardcore property (\cref{lem:amplified_adaptive_hardcore}).
Therefore, our assumption that $\Pr[ \wtexpb{\mathcal{F},\cA}{cut}{and}{choose}(\secp,n)=1]$ is non-negligible is false, which completes the proof of \cref{lem:cut_and_choose_adaptive_hardcore}.
\end{proof}

\subsection{Construction}
\label{sec:PKE_cd_cc_construction}
We construct a PKE scheme with certified deletion with classical communication $\Sigma_{\cccd} =(\keygen,\Enc,\Dec,\Delete,\Vrfy)$ with plaintext space $\Ms=\bit^\ell$ from 
an NTCF family $\mathcal{F}$ with the corresponding trapdoor injective family $\mathcal{G}$ for which we use notations given in~\cref{sec:NTCF},
a public key NCE scheme $\Sigma_{\nce}=\NCE.(\keygen,\Enc,\Dec,\Fake,\Reveal)$ with plaintext space $\{S\subseteq [4n]: |S|=2n\}$  
where $n$ is a positive integer such that $n\leq \poly(\secp)$ and $n=\omega(\log \secp)$ and
we just write $S$ to mean the description of the set $S$ by abuse of notation,  
a OW-CPA secure PKE scheme $\Sigma_{\ow}=\OW.(\keygen,\Enc,\Dec)$ with plaintext space $\bit^\secp$,
and a hash function $H$ 
from $\bit^{\secp} \times (\bit\times \mathcal{X})^{2n}$ to $\bit^\ell$
modeled as a quantumly-accessible random oracle.
\begin{description}
\item[$\keygen(1^\secp)$:] $ $
\begin{itemize}
\item Generate $(\nce.\pk,\nce.\sk,\nce.\aux)\lrun \NCE.\keygen(1^\secp)$ 
and $(\ow.\pk,\ow.\sk)\lrun \OW.\keygen(1^\secp)$ 
and output $(\pk,\sk) \seteq ((\nce.\pk,\ow.\pk),(\nce.\sk,\ow.\sk))$.
\end{itemize}
\item[$\Enc\langle\sen(\pk,m),\rec \rangle$:]  This is an interactive protocol between a sender $\sen$ with input $(\pk,m)$ and a receiver $\rec$ without input that works as follows. 
\begin{itemize}
\item $\sen$ parses $\pk=(\nce.\pk,\ow.\pk)$. 
\item $\sen$ chooses a uniform subset $S\subseteq [4n]$ such that $|S|=2n$,
    generates 
\begin{align}
  (\fk_i,\td_i)\lrun
  \begin{cases}
   \GenG(1^\secp)~~~&i \in S\\
   \GenF(1^\secp)~~~~&i \in \overline{S}
  \end{cases}
\end{align}    
for $i\in [4n]$, 
     and
    sends $\{\fk_i\}_{i\in[4n]}$ to $\rec$.
    \item 
    For $i\in [4n]$, 
    $\rec$ generates a quantum state
\begin{align}
\ket{\psi'_{i}} =
\begin{cases}
 \frac{1}{\sqrt{|\mathcal{X}|}}\sum_{x\in\mathcal{X},
y\in\mathcal{Y},b\in\bit}\sqrt{(g_{\fk_i,b}(x))(y)}|b,x\rangle|y\rangle ~~~&(i\in S)\\
\frac{1}{\sqrt{\abs{\cX}}}\sum_{x\in \cX, y\in\cY, b\in\zo{}}\sqrt{(f'_{\fk_i,b}(x))(y)}\ket{b,x}\ket{y} & (i \in \overline{S})
\end{cases}
\end{align}
by using $\Samp$, measure the last register to obtain $y_i\in\cY$, and let $\ket{\phi'_i}$ be the post-measurement state where the measured register is discarded.
Note that this can be done without knowing $S$ since $\Samp_{\mathcal{F}}=\Samp_{\mathcal{G}}$, which is just denoted by $\Samp$, as required in \cref{def:injective_invariance}.
By \cref{def:NTCF,def:TIF}, we can see that for all $i\in [4n]$, $\ket{\phi'_i}$ has a negligible trace distance from the following state:
\begin{align}
\ket{\phi_{i}} =
\begin{cases}
\ket{b_i}\ket{x_{i}} ~~~&(i\in S)\\
\frac{1}{\sqrt{2}}\left(\ket{0}\ket{x_{i,0}}+\ket{1}\ket{x_{i,1}}\right)  & (i \in \overline{S})
\end{cases}
\end{align}
where $(x_{i},b_i)\lrun \InvG(\td_i,y_i)$ for $i\in S$ and $x_{i,\beta}\lrun \InvF(\td_i,\beta,y_i)$ for $(i,\beta)\in \overline{S}\times \bit$.\footnote{Indeed, $\ket{\psi'_i}=\ket{\psi_i}$ for $i\in S$.}
    $\rec$ sends 
    $\{y_i\}_{i\in[4n]}$ to $\sen$ and keeps $\{\ket{\phi'_i}\}_{i\in[4n]}$.
\item $\sen$ chooses $K\lrun \bit^\secp$ and computes 
$(b_i,x_i)\lrun \InvG(\td_i,y_i)$ for all $i\in S$.
If $\Chk_{\cG}(\fk_i,b_i,x_i,y_i)=0$ for some $i\in S$, $\sen$ returns $\bot$ to $\rec$. 
Otherwise, let $i_1,...,i_{2n}$ be the elements of $S$ in the ascending order. 
$\sen$ 
sets $Z\seteq (K,(b_{i_1}, x_{i_1} ),(b_{i_2},x_{i_2}) ,...,(b_{i_{2n}},x_{i_{2n}}))$, 
computes 
\begin{align}
&\nce.\ct\lrun \NCE.\Enc(\nce.\pk,S),\\
&\ow.\ct\lrun \OW.\Enc(\ow.\pk,K),\\
&\ctmsg\seteq m\xor H(Z),
\end{align}
and sends $(\nce.\ct,\ow.\ct,\ctmsg)$ to $\rec$.
\item $\sen$ outputs $\vk\seteq \{\td_i,y_i\}_{i\in \overline{S}}$ 
and $\rec$ outputs $\ct\seteq (\{\ket{\phi'_i}\}_{i\in[4n]},\nce.\ct,\ow.\ct,\ctmsg)$.
\end{itemize}
\item[$\Dec(\sk,\ct)$:] $ $
\begin{itemize}
\item Parse $\sk = (\nce.\sk,\ow.\sk)$ and $\ct = (\{\ket{\phi'_i}\}_{i\in[4n]},\nce.\ct,\ow.\ct,\ctmsg)$.
\item Compute $S^\prime \lrun \NCE.\Dec(\nce.\sk,\nce.\ct)$.
\item Compute $K^\prime \lrun \OW.\Dec(\ow.\sk,\ow.\ct)$.
\item For all $i\in S'$, measure $\ket{\phi'_i}$ in the computational basis and let $(b^\prime_i,x^\prime_i)$ be the outcome. 
\item Compute and output $m^\prime \seteq \ctmsg \xor H(K^\prime,(b_{i_1}^\prime, x_{i_1}^\prime),(b_{i_2}^\prime, x_{i_2}^\prime),...,(b_{i_{2n}}^\prime, x_{i_{2n}}^\prime))$
where $i_1,...,i_{2n}$ are the elements of $S'$ in the ascending order.\footnote{If $S'=\bot$ or $K'=\bot$, output $\bot$.} 
\end{itemize}
\item[$\Delete(\ct)$:] $ $
\begin{itemize}
\item Parse $\ct = (\{\ket{\phi'_i}\}_{i\in[4n]},\nce.\ct,\ow.\ct,\ctmsg)$.
\item For all $i\in [4n]$, evaluate the function $J$ on the second register of $\ket{\phi'_i}$. That is, apply an isometry that maps $\ket{b,x}$ to $\ket{b,J(x)}$ to $\ket{\phi'_i}$. (Note that this can be done efficiently since $J$ is injective and efficiently invertible.) 
Let $\ket{\phi''_i}$ be the resulting state.
\item For all $i\in [4n]$, measure $\ket{\phi''_i}$ in the Hadamard basis and let $(e_i,d_i)$ be the outcome.
\item Output $\cert \seteq \{(e_i,d_i)\}_{i\in [4n]}$.
\end{itemize}
\item[$\Vrfy(\vk,\cert)$:] $ $
\begin{itemize}
\item Parse $\vk = \{\td_i,y_i\}_{i\in \overline{S}}$ 
and $\cert =\{(e_i,d_i)\}_{i\in [4n]}$.
   \item 
   Compute $x_{i,\beta}\lrun \InvF(\td_i,\beta,y_i)$ for all $(i,\beta)\in \overline{S}\times \bit$.
   \item Output $\top$ if 
     $d_i\in G_{\fk_i,0,x_{i,0}} \cap G_{\fk_i,1,x_{i,1}}$ and
    $e_i=d_i\cdot (J(x_{i,0})\oplus J(x_{i,1}))$ hold for all $i\in \overline{S}$ and output $\bot$ otherwise.
\end{itemize}
\end{description}

\paragraph{Correctness.}
As observed in the description, $\ket{\phi'_i}$ in the ciphertext has a negligible trace distance from $\ket{\phi_i}$.
Therefore, it suffices to prove correctness assuming that $\ket{\phi'_i}$ is replaced with $\ket{\phi_i}$.
After this replacement, decryption correctness clearly holds assuming correctness of $\Sigma_{\nce}$ and $\Sigma_{\ow}$.

We prove verification correctness below.
For $i\in \overline{S}$, if we apply $J$ to the second register of $\ket{\phi_i}$ and then apply Hadamard transform for both registers as in $\Delete$, then the resulting state can be written as 
\begin{align}
&2^{- \frac{w+2}{2}}\sum_{d,b,e} (-1)^{d\cdot J(x_{i,b})\xor e b }\ket{e}\ket{d}\\
& = 2^{-\frac{w}{2}}\sum_{d \in \zo{w}}(-1)^{d\cdot J(x_{i,0})}\ket{d\cdot (J(x_{i,0})\xor J(x_{i,1}))}\ket{d}.
\end{align}
Therefore, the measurement result is $(e_i,d_i)$ such that 
$e_i=d_i\cdot (J(x_{i,0})\xor J(x_{i,1}))$ 
for a uniform $d_i\lrun \bit^w$. 
By the first item of the adaptive hardcore property in~\cref{def:NTCF}, it holds that $d_i \in G_{\fk_i,0,x_{i,0}}\cap G_{\fk_i,1,x_{i,1}}$ 
except for a negligible probability.
Therefore, the certificate $\cert=\{(e_i,d_i)\}_{i\in[4n]}$ passes the verification by  $\Vrfy$ with overwhelming probability. 
\paragraph{Security.}
We prove the following theorem.
\begin{theorem}\label{thm:pke_cccd}
If $\Sigma_{\nce}$ is RNC secure,  $\Sigma_{\ow}$ is OW-CPA secure,
and $\mathcal{F}$ is an injective invariant NTCF family with the corresponding injective trapdoor family $\mathcal{G}$, 
$\Sigma_{\cccd}$ is IND-CPA-CD secure in the QROM where $H$ is modeled as a quantumly-accessible random oracle.
\end{theorem}
\begin{proof}
What we need to prove is that for any QPT adversary $\cA$, it holds that
\begin{align}
\advc{\Sigma_{\cccd},\cA}{ccpk}{cert}{del}(\secp)\seteq \abs{\Pr[ \expb{\Sigma_{\cccd},\cA}{ccpk}{cert}{del}(\secp, 0)=1] - \Pr[ \expb{\Sigma_{\cccd},\cA}{ccpk}{cert}{del}(\secp, 1)=1] }\leq \negl(\secp).
\end{align}
Let $q=\poly(\secp)$ be the maximum number of $\cA$'s random oracle queries.   
For clarity, we describe how  $\expb{\Sigma_{\cccd},\cA}{ccpk}{cert}{del}(\secp, b)$ works below. 
\begin{enumerate}
   \item 
   A uniformly random function $H$ from $\bit^{\secp} \times (\bit\times \mathcal{X})^{2n}$ to $\bit^\ell$ is chosen, and 
   $\cA$ can make arbitrarily many quantum queries to $H$ at any time in the experiment. 
    \item The challenger generates  $(\nce.\pk,\nce.\sk,\nce.\aux)\lrun \NCE.\keygen(1^\secp)$ 
and $(\ow.\pk,\ow.\sk)\lrun \OW.\keygen(1^\secp)$ 
and sends $\pk \seteq (\nce.\pk,\ow.\pk)$ to $\cA$.
\item $\cA$ sends $(m_0,m_1)\in \Ms^2$ to the challenger.
    \item  
    The challenger chooses a uniform subset $S\subseteq [4n]$ such that $|S|=2n$,
    generates 
\begin{align}
  (\fk_i,\td_i)\lrun
  \begin{cases}
   \GenG(1^\secp)~~~&i \in S\\
   \GenF(1^\secp)~~~~&i \in \overline{S}
  \end{cases}
\end{align}    
for $i\in [4n]$, 
     and
    sends $\{\fk_i\}_{i\in[4n]}$ to $\cA$.
\item $\cA$ sends $\{y_i\}_{i\in [4n]}$ to the challenger. 
\item \label{step:send_ct}
The challenger chooses $K\lrun \bit^\secp$ and computes 
$(b_i,x_i)\lrun \InvG(\td_i,y_i)$ for all $i\in S$. 
If $\Chk_{\cG}(\fk_i,b_i,x_i,y_i)=0$ for some $i\in S$, the challenger sets $Z\seteq\sfnull$ and returns $\bot$ to $\cA$ where $\sfnull$ is a special symbol indicating that $Z$ is undefined.  
Otherwise, let $i_1,...,i_{2n}$ be the elements of $S$ in the ascending order. 
The challenger 
sets $Z\seteq (K,(b_{i_1}, x_{i_1} ),(b_{i_2},x_{i_2}) ,...,(b_{i_{2n}},x_{i_{2n}}))$,  
computes 
\begin{align}
&\nce.\ct\lrun \NCE.\Enc(\nce.\pk,S),\\
&\ow.\ct\lrun \OW.\Enc(\ow.\pk,K),\\
&\ctmsg\seteq m_b\xor H(Z),
\end{align}
and sends $(\nce.\ct,\ow.\ct,\ctmsg)$ to $\cA$.
\item  \label{step:send_cert}
$\cA$ sends $\cert=\{(e_i,d_i)\}_{i\in [4n]}$ to the challenger.
    \item  \label{step:reveal_sk}
  The challenger
   computes $x_{i,\beta}\lrun \InvF(\td_i,\beta,y_i)$ for all $(i,\beta)\in \overline{S}\times \bit$.
   If 
     $d_i\in G_{\fk_i,0,x_{i,0}} \cap G_{\fk_i,1,x_{i,1}}$ and
    $e_i=d_i\cdot (J(x_{i,0})\oplus J(x_{i,1}))$ hold for all $i\in \overline{S}$, sends $\sk\seteq(\nce.\sk,\ow.\sk)$ to $\cA$, and otherwise sends $\bot$ to $\cA$.
    \item  $\cA$ outputs $b'$. The output of the experiment is $b'$.
\end{enumerate}

We define the following sequence of hybrids.
\begin{description}
\item[$\mathsf{Hyb}_1(b)$:]
Let $\revealsk$ be the event that the challenger sends $\sk$ in Step~\ref{step:reveal_sk}.  
$\mathsf{Hyb}_1(b)$ is identical to $\expb{\Sigma_{\cccd},\cA}{ccpk}{cert}{del}(\secp, b)$ except that 
$K$ is chosen at the beginning  and 
the oracle given to $\cA$ before $\revealsk$ occurs is replaced with $H_{K\concat\ast \ra H'}$, which is $H$ reprogrammed according to $H'$ on inputs whose first entry is $K$ where $H'$ is another independent random function.  
More formally, 
 $H_{K\concat\ast \ra H'}$ is defined by 
\begin{align}
H_{K\concat\ast \ra H'}(K',(b_1,x_1),...,(b_{2n}, x_{2n}))\seteq
\begin{cases}
H(K',(b_1,x_1),...,(b_{2n}, x_{2n}))~~~&(K'\neq K)\\
H'(K',(b_1,x_1),...,(b_{2n}, x_{2n})) ~~~&(K'= K)
\end{cases}.
\end{align}
We note that 
 the challenger still uses $H$ to generate $\ctmsg$ and
the oracle after $\revealsk$ occurs is still
$H$   
similarly to the real experiment. 
On the other hand, if $\revealsk$ does not occur, the oracle $H_{K\concat\ast \ra H'}$ is used throughout the experiment except for the generation of $\ctmsg$.  
\item[$\mathsf{Hyb}_2(b)$:]
This is identical to $\mathsf{Hyb}_1(b)$ except that $\nce.\ct$ and $\nce.\sk$ that may be sent to $\cA$ in Step~\ref{step:send_ct}~and~\ref{step:reveal_sk} 
are replaced by
\begin{align}
&\nce.\tlct\lrun \NCE.\Fake(\nce.\pk,\nce.\sk,\nce.\aux),\\
&\nce.\tlsk\lrun \NCE.\Reveal(\nce.\pk,\nce.\sk,\nce.\aux,\nce.\tlct,S).
\end{align}
\item[$\mathsf{Hyb}_3(b)$:]
This is identical to $\mathsf{Hyb}_2(b)$ except that the oracle given to $\cA$ after  $\revealsk$ occurs is replaced with $H_{Z \ra r}$, which is $H$ reprogrammed to output $r$ on input $Z=(K,(b_{i_1},x_{i_1}),...,(b_{i_{2n}},x_{i_{2n}}))$ where $r$ is an independently random $\ell$-bit string. 
More formally, 
 $H_{Z \ra r}$ is defined by 
\begin{align}
H_{Z \ra r}(Z')\seteq
\begin{cases}
H(Z')~~~&(Z'\neq Z)\\
r ~~~&(Z'= Z)
\end{cases}.
\end{align}
 Note that we have $H_{Z \rightarrow r}= H$ if $Z=\sfnull$, i.e., if  $\Chk_{\cG}(\fk_i,b_i,x_i,y_i)=0$ for some $i\in S$ in Step \ref{step:send_ct}. 
\end{description}

\begin{proposition}\label{prop:pkecccd_hyb_one}
If $\Sigma_{\ow}$ is OW-CPA secure, 
$\abs{\Pr[\expb{\Sigma_{\cccd},\cA}{ccpk}{cert}{del}(\secp, b)=1] - \Pr[\mathsf{Hyb}_1(b)=1]} \le \negl(\secp)$.
\end{proposition}
\begin{proof}
To show this, we assume that
$
\abs{\Pr[\expb{\Sigma_{\cccd},\cA}{ccpk}{cert}{del}(\secp, b)=1] - \Pr[\mathsf{Hyb}_1(b)=1]} 
$
is non-negligible, and construct an adversary $\cB_\ow$ that breaks the OW-CPA security of $\Sigma_\ow$.
For notational simplicity,
we denote $\expb{\Sigma_{\cccd},\cA}{ccpk}{cert}{del}(\secp, b)$ by $\mathsf{Hyb}_0(b)$.
We consider an algorithm $\widetilde{\cA}$ that works as follows. $\widetilde{\cA}$ is given an oracle $\mathcal{O}$, which is either $H$ or $H_{K\concat \ast \rightarrow H'}$, and an input $z$ that consists of $K$ and the whole truth table of $H$, where 
$K\lrun \bit^\secp$,  
$H$ and $H'$ are uniformly random functions. 
$\widetilde{\cA}$ runs $\mathsf{Hyb}_0(b)$ except that 
it uses its own oracle $\mathcal{O}$ to simulate $\cA$'s random oracle queries before $\revealsk$ occurs. 
On the other hand, $\widetilde{\cA}$ uses $H$ to simulate $\ctmsg$ and $\cA$'s random oracle queries after $\revealsk$ occurs regardless of $\mathcal{O}$, which is possible because the truth table of $H$ is included in the input $z$.    
By definition, we clearly have 
\begin{align}
    \Pr[\mathsf{Hyb}_0(b)=1]=\Pr[\widetilde{\cA}^{H}(K,H)=1]
\end{align}
and 
\begin{align}
    \Pr[\mathsf{Hyb}_1(b)=1]=\Pr[\widetilde{\cA}^{H_{K\concat \ast \rightarrow H'}}(K,H)=1]
\end{align}
where $H$ in the input means the truth table of $H$.
We apply the one-way to hiding lemma (\cref{lem:o2h}) to the above $\widetilde{\cA}$. 
Note that $\widetilde{\cA}$ is inefficient, but the one-way to hiding lemma is applicable to inefficient algorithms.
Then if we let $\widetilde{\cB}$ be the algorithm that measures uniformly chosen query of $\widetilde{\cA}$ and $S_K\subseteq\bit^{\secp} \times (\bit\times \mathcal{X})^{2n}$ be the subset of all elements whose first entry is $K$, we have 
\begin{align}
   \abs{\Pr[\widetilde{\cA}^{H}(K,H)=1]-\Pr[\widetilde{\cA}^{H_{K\concat \ast \rightarrow H'}}(K,H)=1]}\leq 2q\sqrt{\Pr[\widetilde{\cB}^{H_{K\concat \ast \rightarrow H'}}(K,H)\in S_K]}.  
\end{align}
By the assumption, the LHS is non-negligible, and thus $\Pr[\widetilde{\cB}^{H_{K\concat \ast \rightarrow H'}}(K,H)\in S_K]$ is non-negligible. 
We remark that $\widetilde{\cB}$ uses the truth table of $H$ only for generating $\ctmsg:=m_b\oplus H(Z)$ since it halts before $\revealsk$ occurs, and thus it needs not to simulate the random oracle for $\cA$ after  $\revealsk$ occurs. 
Here, we observe that the oracle $H_{K\concat \ast \rightarrow H'}$ reveals no information about $H(Z)$ since the first entry of $Z$ is $K$, and thus $\ctmsg$ is  independently and uniformly random. 
Therefore, if we let $\widetilde{\cB'}$ be the same as  $\widetilde{\cB}$ except that it does not take the truth table of $H$ as input, and sets $\ctmsg$ to be a uniformly random string instead of setting $\ctmsg:=m_b\oplus H(Z)$,  we have 
\begin{align}
    \Pr[\widetilde{\cB}^{H_{K\concat \ast \rightarrow H'}}(K,H)\in S_K]=\Pr[{\widetilde{\cB'}}^{H_{K\concat \ast \rightarrow H'}}(K)\in S_K].
\end{align}
Moreover, for any fixed $K$, when $H$ and $H'$ are uniformly random, $H_{K\concat \ast \rightarrow H'}$ is also a uniformly random function, and thus we have 
\begin{align}
\Pr[{\widetilde{\cB'}}^{H_{K\concat \ast \rightarrow H'}}(K)\in S_K]=\Pr[{\widetilde{\cB'}}^{H}(K)\in S_K].
\end{align}
Since $\Pr[\widetilde{\cB}^{H_{K\concat \ast \rightarrow H'}}(K,H)\in S_K]$ is non-negligible,  $\Pr[{\widetilde{\cB'}}^{H}(K)\in S_K]$ is also non-negligible.
Recall that $\widetilde{\cB'}^{H}$ is an algorithm that simulates $\mathsf{Hyb}_0(b)$
  with the modification that $\ctmsg$ is set to be uniformly random and measures randomly chosen $\A$'s query before $\revealsk$ occurs. 
Then it is straightforward to construct an adversary $\cB_\ow$ that breaks the OW-CPA security of $\Sigma_\ow$ by using $\widetilde{\cB'}$.
For clarity, we give the description of $\cB_\ow$ below.

$\cB_\ow$ is given  $(\ow.\pk,\ow.\ct)$ from the challenger of $\expa{\Sigma,\cB_\ow}{ow}{cpa}(\secp)$.
$\cB_\ow$ chooses $i\lrun [q]$ and runs  $\hybi{1}(b)$ until $\cA$ makes $i$-th random oracle query or $\revealsk$ occurs where  $\cB_\ow$ embeds the problem instance $(\ow.\pk,\ow.\ct)$ into those sent to $\cA$ instead of generating them by itself. 
If $\revealsk$ occurs before $\cA$ makes $i$-th random oracle query, $\cB_\ow$ just aborts. 
Otherwise, $\cB_\ow$ measures the $i$-th random oracle query by $\cA$, and lets $Z'$ be the measurement outcome. 
$\cB_\ow$ outputs the first entry of $Z'$.

We note that $\cB_\ow$ can efficiently simulate the random oracle $H$ by Zhandry's compressed oracle technique \cite{C:Zhandry19}.
We also note that $\cB_\ow$ needs not to know $\ow.\sk$ since it aborts as soon as $\revealsk$ occurs. 
We can see that the probability that $Z'\in S_K$ where $K$ is the underlying message behind $\ow.\ct$  is exactly $\Pr[{\widetilde{\cB'}}^{H}(K)\in S_K]$, which is non-negligible. 
When $Z'\in S_K$, $\cB_\ow$ outputs $K$.
Therefore, $\cB_\ow$ succeeds in predicting $K$ with non-negligible probability.
This contradicts the OW-CPA security of $\Sigma_\ow$.
Therefore $\abs{\Pr[\hybi{0}(b)]-\Pr[\hybi{1}(b)]}$ is negligible.

\end{proof}

\begin{proposition}\label{prop:pkecccd_hyb_two}
If $\Sigma_{\nce}$ is RNC secure, 
$\abs{\Pr[\mathsf{Hyb}_1(b)=1]-\Pr[\mathsf{Hyb}_2(b)=1]} \le \negl(\secp)$.
\end{proposition}
\begin{proof}

The proof is very similar to \cref{prop:pkecd_hyb_one},
but for the convenience of readers, we give the proof here.
To show the proposition, we assume
that
$\abs{\Pr[\mathsf{Hyb}_1(b)=1]-\Pr[\mathsf{Hyb}_2(b)=1]}$ is non-negligible,
and construct an adversary $\cB_\nce$ that breaks the RNC security of $\Sigma_{\nce}$.
Let $\cA$ be the distinguisher for $\hybi{1}(b)$ and $\hybi{2}(b)$. 
First, $\cB_\nce$ receives $\nce.\pk$ from the challenger of $\expa{\Sigma,\cB_\nce}{rec}{nc}(\secp,b')$ where $b'$ is $0$ or $1$.
 $\cB_\nce$ generates $K\lrun \bit^\secp$ and simulates the random oracle for $\cA$ as in $\hybi{1}(b)$ and $\hybi{2}(b)$.
That is, it uses 
$H_{K\concat \ast \rightarrow H'}$ before $\revealsk$ occurs
and 
$H$ after $\revealsk$ occurs,
where $H$ and $H'$ are random functions. 
Note that quantum access to random functions $H$ and $H'$ can be efficiently simulated by using Zhandry's compressed oracle technique \cite{C:Zhandry19}.  
$\cB_\nce$ generates $(\ow.\pk,\ow.\sk)\leftarrow\OW.\keygen(1^\secp)$ and sends $(\nce.\pk,\ow.\pk)$ to $\cA$ 
and receives $(m_0,m_1)$ from $\cA$.
$\cB_\nce$ generates 
$S$ and 
$\{\fk_i,\td_i\}_{i\in[4n]}$ as in $\hybi{1}(b)$ and $\hybi{2}(b)$,  
sends $\{\fk_i\}_{i\in[4n]}$ to $\cA$, 
and receives $\{y_i\}_{i\in[4n]}$ from $\cA$.
$\cB_\nce$ computes $\ow.\ct$ and $\ctmsg$ as in $\hybi{1}(b)$ and $\hybi{2}(b)$, 
 sends $S$ to the challenger of $\expa{\Sigma,\cB_\nce}{rec}{nc}(\secp,b')$, 
and receives $(\nce.\CT^*,\nce.\sk^*)$ from the challenger. 
Here,
\begin{align}
(\nce.\ct^*,\nce.\sk^*)= 
(\NCE.\Enc (\nce.\pk,S),\nce.\sk)
\end{align}
if $b'=0$, and
\begin{align}
(\nce.\ct^*,\nce.\sk^*)= 
(\NCE.\Fake (\nce.\pk,\nce.\sk,\nce.\aux),\NCE.\Reveal(\nce.\pk,\nce.\sk,\nce.\aux,\nce.\ct^*,S))
\end{align}
if $b'=1$.
$\cB_\nce$ sends $(\nce.\ct^*,\ow.\ct,\ct_\msg)$ to $\cA$  
and receives $\cert$ from $\cA$.  
If it is valid one (i.e., $\vrfy(\vk,\cert)=\top$), $\cB_\nce$ sends $(\nce.\sk^*,\ow.\sk)$
to $\cA$. Otherwise, it sends $\bot$ to $\cA$.
 It is easy to see that $\cB_\nce$ perfectly simulates  $\hybi{1}(b)$ if $b'=0$ and $\hybi{2}(b)$ otherwise.  
By assumption, $\cA$ can distinguish $\hybi{1}(b)$
and $\hybi{2}(b)$, and therefore $\cB_\nce$ can break the RNC security of $\Sigma_\nce$.

\end{proof}

\begin{proposition}\label{prop:pkecccd_hyb_three}
If $\mathcal{F}$ and $\mathcal{G}$ satisfy the cut-and-choose adaptive hardcore property (\cref{lem:cut_and_choose_adaptive_hardcore}),  
$\abs{\Pr[\mathsf{Hyb}_2(b)=1]-\Pr[\mathsf{Hyb}_3(b)=1]} \le \negl(\secp)$.
\end{proposition}
\begin{proof}

To show this, we assume that 
$\abs{\Pr[\mathsf{Hyb}_2(b)=1]-\Pr[\mathsf{Hyb}_3(b)=1]}$ is non-negligible,
and show that the cut-and-choose adaptive hardcore property (\cref{lem:cut_and_choose_adaptive_hardcore})   
for $\mathcal{F}$ and $\mathcal{G}$ is broken.
For simplicity, we first prove this assuming that $\abs{\Pr[\mathsf{Hyb}_2(b)=1]-\Pr[\mathsf{Hyb}_3(b)=1]}$ is \emph{noticeable}, i.e., there exists a polynomial $p$ such that   $\abs{\Pr[\mathsf{Hyb}_2(b)=1]-\Pr[\mathsf{Hyb}_3(b)=1]}\geq 1/p(\secp)$, and then we explain how to extend it to the non-negligible case.\footnote{Looking ahead, we first consider the noticeable case because the fact that  a product of two noticeable functions is also noticeable simplifies the proof. Note that a similar statement for non-negligible functions does not hold in general.} 
First, we remark that $\mathsf{Hyb}_2(b)$ and $\mathsf{Hyb}_3(b)$ are identical unless $\revealsk$ occurs. 
Therefore, we have
\begin{align}
\abs{\Pr[\mathsf{Hyb}_2(b)=1]-\Pr[\mathsf{Hyb}_3(b)=1]}&=\abs{\Pr[\mathsf{Hyb}_2(b)=1\land \revealsk]-\Pr[\mathsf{Hyb}_3(b)=1\land \revealsk]}\\
&=\Pr[\revealsk]\cdot \abs{\Pr[\mathsf{Hyb}_2(b)=1\mid \revealsk]-\Pr[\mathsf{Hyb}_3(b)=1\mid \revealsk]}
\end{align}
where we remark that $\Pr[\revealsk]$ is identical in $\mathsf{Hyb}_2(b)$ and $\mathsf{Hyb}_3(b)$ and thus we need not specify which hybrid is considered when we write  $\Pr[\revealsk]$. 
Therefore, both $\Pr[\revealsk]$
and $\abs{\Pr[\mathsf{Hyb}_2(b)=1|\revealsk]-\Pr[\mathsf{Hyb}_3(b)=1|\revealsk]}$ are noticeable.

We consider an algorithm $\widetilde{\cA}$ that works as follows. $\widetilde{\cA}$ is given an oracle $\mathcal{O}$, which is either $H$ or $H_{Z \rightarrow r}$ and an input $z\seteq (\{\fk_i,\td_i\}_{i\in[4n]},\{y_i,e_i,d_i\}_{i\in[4n]},\rho_{\cA})$ that are sampled according to the distribution  $\mathcal{D}_{\revealsk}$ defined below.
Let $\mathcal{D}$ be the distribution of $(H,H_{Z \rightarrow r},(\{\fk_i,\td_i\}_{i\in[4n]},\{y_i,e_i,d_i\}_{i\in[4n]},\rho_{\cA}))$ sampled as in $\mathsf{Hyb}_3(b)$ where $\rho_{\cA}$ is $\cA$'s internal state just after sending $\cert$ in Step~\ref{step:send_cert}. 
We assume that $\cA$ does nothing after sending $\cert$ until it receives $\sk$ or $\bot$ in Step~\ref{step:reveal_sk} without loss of generality. Therefore, $\rho_{\cA}$ is also the state just before receiving $\sk$ or $\bot$ in Step~\ref{step:reveal_sk}.  
$\mathcal{D}_{\revealsk}$ is defined to be the conditional distribution of $\mathcal{D}$ conditioned on that $\revealsk$ occurs. 
Given an oracle $\mathcal{O}$, which is either $H$ or $H_{Z \rightarrow r}$ and an input $z\seteq (\{\fk_i,\td_i\}_{i\in[4n]},\{y_i,e_i,d_i\}_{i\in[4n]},\rho_{\cA})$ sampled from $\mathcal{D}_{\revealsk}$, 
$\widetilde{\cA}$ runs $\hybi{2}(b)$ starting from 
Step~\ref{step:reveal_sk} where 
the internal state of $\cA$ is initialized to be $\rho_{\cA}$ and 
$\widetilde{\cA}$ uses its own oracle $\mathcal{O}$ to simulate $\cA$'s random oracle queries.
By definition, we clearly have 
\begin{align}
    \Pr[\mathsf{Hyb}_2(b)=1\mid \revealsk]=\Pr[\widetilde{\cA}^{H}(z)=1]
\end{align}
and 
\begin{align}
    \Pr[\mathsf{Hyb}_3(b)=1\mid \revealsk]=\Pr[\widetilde{\cA}^{H_{Z \rightarrow r}}(z)=1]
\end{align}
where $(H,H_{Z\rightarrow r},z) \lrun \mathcal{D}_{\revealsk}$. 
We apply the one-way to hiding lemma (\cref{lem:o2h}) to the above $\widetilde{\cA}$. 
Let $\widetilde{\cB}$ be the algorithm that measures uniformly chosen query of $\widetilde{\cA}$.
Then we have 
\begin{align}
   \abs{\Pr[\widetilde{\cA}^{H}(z)=1]-\Pr[\widetilde{\cA}^{H_{Z \rightarrow r}}(z)=1]} \leq 2q\sqrt{\Pr[\widetilde{\cB}^{H_{Z \rightarrow r}}(z)= Z]}
\end{align}
where 
$(H,H_{Z\rightarrow r},z) \lrun \mathcal{D}_{\revealsk}$ and when $Z=\sfnull$, we regard that  $\widetilde{\cB}^{H_{Z \rightarrow r}}(z)= Z$ does not occur regardless of the output of $\widetilde{\cB}$.    
The LHS is equal to   $\abs{\Pr[\mathsf{Hyb}_2(b)=1|\revealsk]-\Pr[\mathsf{Hyb}_3(b)=1|\revealsk]}$, which is noticeable.
Therefore $\Pr_{(H,H_{Z\rightarrow r},z) \lrun \mathcal{D}_{\revealsk}}[\widetilde{\cB}^{H_{Z \rightarrow r}}(z)=Z]$ is noticeable.      
Here, we observe that $H(Z)$ is used only for generating $\ctmsg\seteq m_b\oplus H(Z)$ in the execution of $(H,H_{Z\rightarrow r},z) \lrun \mathcal{D}_{\revealsk}$ and $\widetilde{\cB}^{H_{Z \rightarrow r}}(z)$. 
Therefore, the output distribution of $\widetilde{\cB}$ does not change even if we set $\ctmsg$ to be a uniformly random string in the sampling procedure of $\mathcal{D}_{\revealsk}$.
Moreover, 
for any fixed $Z$, 
the joint distribution of  $(H_{K\concat \ast\rightarrow H'},H_{Z \rightarrow r})$ is identical to the joint distribution of $(H_{K\concat \ast\rightarrow H'},H)$   
when $H$, $H'$, and $r$ are uniformly random.
Then, if we denote by $\mathcal{D}'_{\revealsk}$ the distribution that is the same as  $\mathcal{D}_{\revealsk}$ except that $\ctmsg$ is set to be a uniformly random string in the sampling procedure  and the second output $H_{Z \rightarrow r}$ is omitted, 
we have 
\begin{align}
    \Pr_{(H,H_{Z\rightarrow r},z) \lrun \mathcal{D}_{\revealsk}}[\widetilde{\cB}^{H_{Z \rightarrow r}}(z)=Z]=\Pr_{(H,z) \lrun \mathcal{D}'_{\revealsk}}[\widetilde{\cB}^{H}(z)=Z].
\end{align}
Since the LHS is noticeable, the RHS is noticeable.
We note that the sequential execution of $(H,z) \lrun \mathcal{D}'_{\revealsk}$ and $\widetilde{\cB}^{H}(z)$ is the same as executing $\hybi{2}(b)$ with the modification that $\ctmsg$ is set to be uniformly random and measuring a uniformly chosen $\A$'s query after Step~\ref{step:reveal_sk} conditioned on that $\revealsk$ occurs.
This naturally gives an adversary $\cB_{\mathsf{cac}}$ that breaks 
the cut-and-choose adaptive hardcore property (\cref{lem:cut_and_choose_adaptive_hardcore})   
for $\mathcal{F}$ and $\mathcal{G}$, which embeds the problem instance of $\expb{(\mathcal{F},\mathcal{G}),\cB_{\mathsf{cac}}}{cut}{and}{choose}(\secp,n)$ into the execution of $(H,z) \lrun \mathcal{D}'_{\revealsk}$ and $\widetilde{\cB}^{H}(z)$.
For clarity, we give the description of $\cB_{\mathsf{cac}}$ below.

$\cB_{\mathsf{cac}}$ receives $\{\fk_i\}_{i\in[4n]}$ from the challenger of $\expb{(\mathcal{F},\mathcal{G}),\cB_{\mathsf{cac}}}{cut}{and}{choose}(\secp,n)$.
$\cB_{\mathsf{cac}}$ chooses $K\lrun \bit^\secp$ and simulates $\A$'s random oracle queries before $\revealsk$ occurs by $H_{K\concat\ast \rightarrow H'}$ and those after $\revealsk$ occurs by $H$, which can be done efficiently by  Zhandry's compressed oracle technique \cite{C:Zhandry19}. 
$\cB_{\mathsf{cac}}$ generates $(\nce.\pk,\nce.\sk,\nce.\aux)\lrun \NCE.\keygen(1^\secp)$ 
and $(\ow.\pk,\ow.\sk)\lrun \OW.\keygen(1^\secp)$, sends $\pk \seteq (\nce.\pk,\ow.\pk)$ to $\cA$, and receives $(m_0,m_1)$ from $\cA$. 
$\cB_{\mathsf{cac}}$ sends $\{\fk_i\}_{i\in[4n]}$ to $\cA$ and receives $\{y_i\}_{i\in [4n]}$ from $\cA$. $\cB_{\mathsf{cac}}$ generates 
$\nce.\tlct \lrun \NCE.\Fake(\nce.\pk,\nce.\sk,\allowbreak\nce.\aux)$, 
$\ow.\ct\lrun \OW.\Enc(\ow.\pk,K)$, and 
$\ctmsg\lrun \bit^\ell$, sends $(\nce.\tlct,\ow.\ct,\ctmsg)$ to $\cA$, and receives $\cert=\{(e_i,d_i)\}_{i\in [4n]}$ from $\cA$.  
$\cB_{\mathsf{cac}}$ sends $\{(y_i,e_i,d_i)\}_{i\in [4n]}$ to the challenger of $\expb{(\mathcal{F},\mathcal{G}),\cB_{\mathsf{cac}}}{cut}{and}{choose}(\secp,n)$.
If the challenger aborts, $\cB_{\mathsf{cac}}$ aborts.
Otherwise, $\cB_{\mathsf{cac}}$ receives $S$ from the challenger, generates $\nce.\tlsk\lrun \NCE.\Reveal(\nce.\pk,\nce.\sk,\allowbreak \nce.\aux,\nce.\tlct,S)$, and sends  $\sk\seteq(\nce.\tlsk,\ow.\sk)$ to $\cA$.
$\cB_{\mathsf{cac}}$ chooses $i\lrun [q]$,  measures $\A$'s $i$-th query after Step~\ref{step:reveal_sk}, and lets $Z^\prime=(K^\prime, (b_{i_1}^\prime, x_{i_1}^\prime), (b_{i_2}^\prime, x_{i_2}^\prime),...,(b_{i_{2n}}^\prime, x_{i_{2n}}^\prime))$ be the measurement outcome where $i_1,...,i_{2n}$ are the elements of $S$ in the ascending order. 
$\cB_{\mathsf{cac}}$ sends $\{b_i^\prime,x_i^\prime\}_{i\in S}$ to the challenger.

We can see that $\cB_{\mathsf{cac}}$ perfectly simulates the sequential execution of $(H,z) \lrun \mathcal{D}'_{\revealsk}$ and $\widetilde{\cB}^{H}(z)$ conditioned on that it does not abort, which corresponds to the event ${\revealsk}$ in the simulated experiment for $\cA$. 
Moreover, when $Z^\prime=Z\neq \sfnull$ where $Z$ is defined as in $\hybi{3}(b)$, we have $\Chk_{\cG}(\fk_i,b_i^\prime,x_i^\prime,y_i)=1$ for all $i\in S$ by the definition of $Z$.
Therefore, we have 
\begin{align}
    \Pr[\expb{(\mathcal{F},\mathcal{G}),\cB_{\mathsf{cac}}}{cut}{and}{choose}(\secp,n)=1]\geq \Pr[\revealsk]\cdot \Pr_{(H,z) \lrun \mathcal{D}'_{\revealsk}}[\widetilde{\cB}^{H}(z)=Z].
\end{align}
Since both $\Pr[\revealsk]$ and $\Pr_{(H,z) \lrun \mathcal{D}'_{\revealsk}}[\widetilde{\cB}^{H}(z)=Z]$ are noticeable as already proven, $\Pr[\expb{(\mathcal{F},\mathcal{G}),\cB_{\mathsf{cac}}}{cut}{and}{choose}(\secp,n)=1]$ is noticeable, which means that $\cB_{\mathsf{cac}}$ breaks the cut-and-choose adaptive hardcore property. 
This completes the proof for the case where $\abs{\Pr[\mathsf{Hyb}_2(b)=1]-\Pr[\mathsf{Hyb}_3(b)=1]}$ is noticeable.

When it is non-negligible rather than noticeable, the reason that the above proof does not immediately works is that   $\Pr[\revealsk]\cdot \Pr_{(H,z) \lrun \mathcal{D}'_{\revealsk}}[\widetilde{\cB}^{H}(z)=Z]$ may be negligible even if both $\Pr[\revealsk]$ and $\Pr_{(H,z) \lrun \mathcal{D}'_{\revealsk}}[\widetilde{\cB}^{H}(z)=Z]$ are non-negligible in general.
Intuitively, we overcome this by observing that $\Pr[\revealsk]$ and $\Pr_{(H,z) \lrun \mathcal{D}'_{\revealsk}}[\widetilde{\cB}^{H}(z)=Z]$ take ``noticeable values" on the same infinite subset of security parameters, in which case the product also takes ``noticeable values" on the same subset. 
More precisely, since we assume that $\abs{\Pr[\mathsf{Hyb}_2(b)=1]-\Pr[\mathsf{Hyb}_3(b)=1]}$ is non-negligible, there exists an infinite subset $I\subseteq \mathbb{N}$ and a polynomial $p$ such that  $\abs{\Pr[\mathsf{Hyb}_2(b)=1]-\Pr[\mathsf{Hyb}_3(b)=1]}\geq 1/p(\secp)$ for all $\secp \in I$. 
By similar arguments as above, we can show that there exists a polynomial $p'$ such that
$\Pr[\revealsk]\geq 1/p'(\secp)$ and $\Pr_{(H,z) \lrun \mathcal{D}'_{\revealsk}}[\widetilde{\cB}^{H}(z)=Z]\geq 1/p'(\secp)$ for all $\secp \in I$.
Then we have $\Pr[\expb{(\mathcal{F},\mathcal{G}),\cB_{\mathsf{cac}}}{cut}{and}{choose}(\secp,n)=1]\geq (1/p'(\secp))^2$ for all $\secp\in I$, in which case it is non-negligible. 
This completes the proof.

\end{proof}

\begin{proposition}\label{prop:pkecccd_hyb_final}
It holds that   
$\Pr[\mathsf{Hyb}_3(0)=1]=\Pr[\mathsf{Hyb}_3(1)=1]$.
\end{proposition}
\begin{proof}
In $\hybi{3}$, the challenger queries $H$ while the adversary queries
$H_{K\|*\rightarrow H'}$ or $H_{Z\rightarrow r}$. 
Therefore, $H(Z)$ is used only for generating $\ctmsg$ in $\mathsf{Hyb_3}$  and thus 
$\ctmsg$ is an independently uniform string regardless of $b$ from the view of the adversary.
Therefore \cref{prop:pkecccd_hyb_final} holds.
\end{proof}

By combining
\cref{prop:pkecccd_hyb_one,prop:pkecccd_hyb_two,prop:pkecccd_hyb_three,prop:pkecccd_hyb_final}
\cref{thm:pke_cccd} is proven.
\end{proof}

\section{PKE with Publicly Verifiable Certified Deletion and Classical Communication}\label{sec:WE_plus_OSS}

In this section, we define the notion of public key encryption with publicly verifiable
certified deletion with classical
communication, and construct it from the witness encryption and the one-shot signature.
In \cref{sec:pk_pv_cd_def_classical_com}, we present the definition of the public key encryption
with publicly verifiable certified deletion with classical communication. 
In \cref{sec:pk_pv_cd_cc_construction}, we give a construction 
and show
its security.

\subsection{Definition of PKE with Publicly Verifiable Certified Deletion with Classical Communication}
\label{sec:pk_pv_cd_def_classical_com}

In this section, we consider a PKE with publicly verifiable certified deletion with classical communication.
It is a publicly verifiable version of the one given in
(\cref{def:pk_cert_del_classical_com}).
The construction of \cref{sec:PKE_cd_cc_construction}
is not publicly verifiable, because the verification key $\vk$ 
(which is the trapdoor $\{\td_i\}_i$) should be
secret to the adversary. On the other hand, in \cref{sec:pk_pv_cd_cc_construction} we construct a publicly verifiable one,
which means that the security is kept even if the verification key $\vk$ (which is the public key $\oss.\pk$ of the one-shot signature in our construction)
is given to the adversary.

The definition (syntax) is the same as that of
the non-publicly-verifiable one (\cref{def:pk_cert_del_classical_com}).
Furthermore, its correctness, i.e., 
the decryption correctness and the verification correctness,
are also the same as those of the non-publicly-verifiable one (\cref{def:pk_cd_correctness_classical_com}).
Regarding the certified deletion security, it is the same as
that of the non-publicly-verifiable one (\cref{def:pk_certified_del_classical_com})
except that the challenger sends $\vk$ to the adversary (which is $\oss.\pk$ in our construction).

\subsection{Construction}
\label{sec:pk_pv_cd_cc_construction}
We construct a PKE scheme with publicly verifiable certified deletion with classical communication
$\Sigma_{\mathsf{pvcccd}}=(\KeyGen,\Enc,\Dec,\Delete,\Vrfy)$ 
from a public key NCE scheme $\Sigma_\nce=\NCE.(\keygen,\Enc,\Dec,\Fake,\Reveal)$,
the witness encryption scheme $\Sigma_\we=\WE.(\Enc,\Dec)$, and 
the one-shot signature scheme $\Sigma_\oss=\OSS.(\Setup,\keygen,\Sign,\Vrfy)$.

\begin{description}
\item[$\keygen(1^\lambda)$:] $ $
\begin{itemize}
\item Generate $(\nce.\pk,\nce.\sk,\nce.\aux)\leftarrow \NCE.\keygen(1^\lambda)$.
\item Output $(\pk,\sk)=(\nce.\pk,\nce.\sk)$.
\end{itemize}

\item[$\Enc\langle \cS(\pk,m),\cR\rangle$:] This is an interactive protocol between
a sender $\cS$ with input $(\pk,m)$ and a receiver $\cR$ without input that works as follows.
\begin{itemize}
\item $\cS$ parses $\pk=\nce.\pk$.
\item $\cS$ generates $\crs\leftarrow\OSS.\Setup(1^\secp)$, and sends $\crs$ to $\cR$.
\item $\cR$ generates $(\oss.\pk,\oss.\sk)\lrun \OSS.\keygen(\crs)$, sends $\oss.\pk$
to $\cS$, and keeps $\oss.\sk$.
    \item $\cS$ computes $\we.\ct\leftarrow \WE.\Enc(1^\lambda,x,m)$ with the statement $x$ that
    ``$\exists \sigma$ s.t. $\OSS.\Vrfy(\crs,\oss.\pk,\sigma,0)=\top$".
\item $\cS$ computes $\nce.\ct\leftarrow\NCE.\Enc(\nce.\pk,\we.\ct)$,
    and sends $\nce.\ct$ to $\cR$.
    \item $\cS$ outputs $\vk=(\crs,\oss.\pk)$. $\cR$ outputs $\ct=(\nce.\ct,\oss.\sk)$.
\end{itemize}

\item[$\Dec(\sk,\ct)$:]  $ $
\begin{itemize}
\item Parse $\sk=\nce.\sk$ and $\ct=(\nce.\ct ,\oss.\sk)$.
    \item Compute $\sigma\leftarrow \OSS.\Sign(\oss.\sk,0)$.
    \item Compute $m'\leftarrow \NCE.\Dec(\nce.\sk,\nce.\ct)$
    \item Compute $m\leftarrow \WE.\Dec(m',\sigma)$.
    \item Output $m$.
\end{itemize}

\item[$\Delete(\ct)$:]  $ $
\begin{itemize}
\item Parse $\ct=(\nce.\ct,\oss.\sk)$.
    \item Compute $\sigma\leftarrow \OSS.\Sign(\oss.\sk,1)$.
    \item Output $\cert=\sigma$.
\end{itemize}

\item[$\Vrfy(\vk,\cert)$:]  $ $
\begin{itemize}
\item Parse $\vk=(\crs,\oss.\pk)$ and $\cert=\sigma$.
    \item Compute $b\leftarrow \OSS.\Vrfy(\crs,\oss.\pk,\sigma,1)$.
    \item Output $b$.
\end{itemize}
\end{description}

\paragraph{Correctness.}
The decryption and verification correctness easily follow from the correctness of $\Sigma_\we$
and $\Sigma_\oss$.

\paragraph{Security.}
We show the following theorem.

\begin{theorem}\label{theorem:WEOSS}
If $\Sigma_\nce$ is RNC secure, $\Sigma_\we$ has the extractable security and $\Sigma_\oss$ is secure, then
$\Sigma_{\mathsf{pvcccd}}$ is IND-CPA-CD secure.
\end{theorem}

\begin{proof}
For clarity, we describe how  $\expb{\Sigma_{\mathsf{pvcccd}},\cA}{pvccpk}{cert}{del}(\secp, b)$ 
(which we call $\hybi{0}(b)$ for simplicity) works below. 
\begin{enumerate}
\item
The challenger generates $(\nce.\pk,\nce.\sk,\nce.\aux)\leftarrow\NCE.\keygen(1^\lambda)$,
and sends $\nce.\pk$ to $\cA$.
\item
$\cA$ sends $(m_0,m_1)\in \cM^2$ to the challenger.
\item
The challenger generates $\crs\leftarrow\OSS.\Setup(1^\secp)$, and
sends $\crs$ to $\cA$.
\item
$\cA$ sends $\oss.\pk$ to the challenger.
\item
The challenger computes $\we.\ct\leftarrow \WE.\Enc(1^\lambda,x,m_b)$,
where $x$ is the statement that ``$\exists\sigma$ s.t. $\OSS.\Vrfy(\crs,\oss.\pk,\sigma,0)=\top$.
The challenger computes $\nce.\ct\leftarrow \NCE.\Enc(\nce.\pk,\we.\ct)$.
The challenger sends $\nce.\ct$ to $\cA$.
\item
$\cA$ sends $\cert=\sigma$ to the challenger.
\item
The challenger computes $b\leftarrow\OSS.\Vrfy(\crs,\oss.\pk,\sigma,1)$.
If the output is $b=\bot$, the challenger sends $\bot$ to $\cA$.
If the output is $b=\top$, the challenger sends $\nce.\sk$ to $\cA$.
\item
$\cA$ outputs $b'$. The output of the experiment is $b'$.
\end{enumerate}

We define the following hybrid.
\begin{description}
\item[$\mathsf{Hyb}_1(b)$:]
This is identical to $\hybi{0}(b)$ except that 
$\nce.\ct$ and $\nce.\sk$ are generated as
\begin{align}
\nce.\ct&\leftarrow \NCE.\Fake(\nce.\pk,\nce.\sk,\nce.\aux),\\
\nce.\sk&\leftarrow \NCE.\Reveal(\nce.\pk,\nce.\sk,\nce.\aux,\nce.\CT,\we.\ct).
\end{align}
 \end{description}
 
\begin{proposition}\label{prop:WEOSS_hyb_1}
If $\Sigma_{\nce}$ is RNC secure, 
$\abs{\Pr[\hybi{0}(b)=1] - \Pr[\mathsf{Hyb}_1(b)=1]} \le \negl(\secp)$.
\end{proposition}

\begin{proof}
To show this, we assume that
$\abs{\Pr[\hybi{0}(b)=1] - \Pr[\mathsf{Hyb}_1(b)=1]}$ is non-negligible, 
and construct an adversary $\cB_\nce$ that breaks the RNC security of $\Sigma_{\nce}$.
Let $\cA$ 
be the distinguisher for $\hybi{0}(b)$ and $\hybi{1}(b)$. 
First, $\cB_\nce$ receives $\nce.\pk$ from the challenger of the RNC security game.
$\cB_\nce$ then sends $\nce.\pk$ to $\cA$.
$\cB_\nce$ receives $(m_0,m_1)$ from $\cA$.
$\cB_\nce$ generates $\crs\lrun \OSS.\Setup(1^\secp)$ and sends $\crs$ to $\cA$.
$\cB_\nce$ receives $\oss.\pk$ from $\cA$.
$\cB_\nce$ generates $\we.\ct\leftarrow\WE.\Enc(1^\secp,x,m_b)$. 
$\cB_\nce$ sends $\we.\ct$ to the challenger of the RNC security game, and receives
$(\nce.\CT^*,\nce.\sk^*)$ from the challenger of the RNC security game.
Here,
\begin{align}
(\nce.\ct^*,\nce.\sk^*)= 
(\NCE.\Enc (\nce.\pk,\we.\ct),\nce.\sk)
\end{align}
if the challenger's bit is 0, and
\begin{align}
(\nce.\ct^*,\nce.\sk^*)= 
(\NCE.\Fake (\nce.\pk,\nce.\sk,\nce.\aux),\NCE.\Reveal(\nce.\pk,\nce.\sk,\nce.\aux,\nce.\ct^*,\we.\ct))
\end{align}
if the challenger's bit is 1.
$\cB_\nce$ sends $\nce.\ct^*$ to $\cA$. 
$\cB_\nce$ receives $\cert=\sigma$ from $\cA$.
If $\OSS.\Vrfy(\crs,\oss.\pk,\sigma,1)=\bot$, $\cB_\nce$ sends $\bot$ to $\cA$, 
and otherwise sends $\nce.\sk^*$ to $\cA$. 
By assumption, $\cA$ can distinguish $\hybi{0}(b)$
and $\hybi{1}(b)$, and therefore $\cB_\nce$ can output the bit of the challenger of the RNC security game
with non-negligible probability,
which breaks the RNC security of $\Sigma_\nce$.
\end{proof}

\begin{proposition}\label{prop:WEOSS_hyb_2}
If $\Sigma_{\oss}$ is secure and $\Sigma_{\we}$ has the extractable security, 
$\abs{\Pr[\mathsf{Hyb}_1(0)=1] - \Pr[\mathsf{Hyb}_1(1)=1]} \le \negl(\secp)$.
\end{proposition}
\begin{proof}
Let $\mathsf{good}$ (resp. $\mathsf{bad}$) be the event that the adversary 
in $\hybi{1}(b)$
sends a valid (resp. invalid) $\cert$. 
Then,
\begin{align}
|\Pr[\hybi{1}(0)=1]-\Pr[\hybi{1}(1)=1]|
&=
|\Pr[\hybi{1}(0)=1\wedge\mathsf{good}]+\Pr[\hybi{1}(0)=1\wedge\mathsf{bad}]\\
&-\Pr[\hybi{1}(1)=1\wedge\mathsf{good}]-\Pr[\hybi{1}(1)=1\wedge\mathsf{bad}]|\\
&\le
|\Pr[\hybi{1}(0)=1\wedge\mathsf{good}]
-\Pr[\hybi{1}(1)=1\wedge\mathsf{good}]|\\
&+|\Pr[\hybi{1}(0)=1\wedge\mathsf{bad}]
-\Pr[\hybi{1}(1)=1\wedge\mathsf{bad}]|\\
&=
|\Pr[\hybi{1}(0)=1\wedge\mathsf{good}]
-\Pr[\hybi{1}(1)=1\wedge\mathsf{good}]|
\end{align}
Assume that
$|\Pr[\hybi{1}(0)=1]-\Pr[\hybi{1}(1)=1]|$ is non-negligible,
which means that there exists an infinite subset $I\subseteq{\mathbb N}$ and a polynomial $p$
such that
$|\Pr[\hybi{1}(0)=1\wedge\mathsf{good}]-\Pr[\hybi{1}(1)=1\wedge\mathsf{good}]|\ge1/p(\secp)$ for all $\secp\in I$.
Because
\begin{align}
|\Pr[\hybi{1}(0)=1\wedge\mathsf{good}]-\Pr[\hybi{1}(1)=1\wedge\mathsf{good}]|=
|\Pr[\hybi{1}(0)=1|\mathsf{good}]-\Pr[\hybi{1}(1)=1|\mathsf{good}]|\Pr[\mathsf{good}],
\end{align}
It means 
$|\Pr[\hybi{1}(0)=1|\mathsf{good}]-\Pr[\hybi{1}(1)=1|\mathsf{good}]|\ge1/p(\secp)$ and
$\Pr[\mathsf{good}]\ge1/p(\secp)$ for all $\secp\in I$.
Let $\aux$ be the adversary's internal state after outputting $\cert$ conditioned on $\mathsf{good}$.  
Then, due to the extractability of the witness encryption,
there is a QPT extractor $\cE$ and polynomial $q$ such that the probability that
$\cE(1^\secp,x,\aux)$ outputs a valid $\sigma'$
such that $\OSS.\Vrfy(\crs,\oss.\pk,\sigma',0)=\top$ is at least $1/q(\secp)$ for all $\secp\in I$. 

We can construct an adversary $\cB_{\oss}$ that breaks the security of the one-shot signature scheme
as follows.
$\cB_{\oss}$ receives $\oss.\crs\leftarrow \OSS.\Setup(\secp)$.
It generates $(\nce.\pk,\nce.\sk,\nce.\aux)\leftarrow \NCE.\keygen(1^\secp)$,
and sends $\nce.\pk$ to the adversary of $\hybi{1}(b)$.
$\cB_{\oss}$ receives $m_0,m_1$ from the adversary of $\hybi{1}(b)$,
and returns $\oss.\crs$.
$\cB_{\oss}$ receives $\oss.\pk$ from the adversary of $\hybi{1}(b)$.
$\cB_{\oss}$ sends
$\nce.\ct^*\leftarrow\NCE.\Fake(\nce.\pk,\nce.\sk,\nce.\aux)$ to the adversary of $\hybi{1}(b)$,
and receives $\cert$.
$\cB_{\oss}$ simulates $\cE$ and gets its output $\sigma'$.
$\cB_{\oss}$ outputs $(\oss.\pk,1,\cert,0,\sigma')$.
Because
\begin{align}
&\Pr[m_0\neq m_1\wedge \Vrfy(\oss.\crs,\oss.\pk,0,\sigma')=\top\wedge
\Vrfy(\oss.\crs,\oss.\pk,1,\cert)=\top]\\
&=\Pr[\mathsf{good}]\Pr[\cE(1^\secp,x,\aux)\mbox{ outputs valid } \sigma']
\ge\frac{1}{p(\secp)}\frac{1}{q(\secp)}
\end{align}
for all $\secp\in I$,
$\cB_{\oss}$ breaks the security of the one-shot signature scheme.

\end{proof}
By \cref{prop:WEOSS_hyb_1} and \cref{prop:WEOSS_hyb_2}, we obtain \cref{theorem:WEOSS}.

\end{proof}

\ifnum\noclassic=1
\else
\section*{Acknowledgement}
TM is supported by the MEXT Q-LEAP, JST FOREST, JST PRESTO No.JPMJPR176A,
and the Grant-in-Aid for Scientific Research (B) No.JP19H04066 of JSPS.
\fi

\ifnum\llncs=1
\bibliographystyle{alpha} 
\bibliography{abbrev3,crypto,reference}
\else
\ifnum\arxiv=1
\newcommand{\etalchar}[1]{$^{#1}$}

\else
\bibliographystyle{alpha} 
\bibliography{abbrev3,crypto,reference}
\fi
\fi

\ifnum\cameraready=1
\else
\appendix

	\ifnum\llncs=1
	\newpage
	 	\setcounter{page}{1}
 	{
	\noindent
 	\begin{center}
	{\Large SUPPLEMENTAL MATERIALS}
	\end{center}
 	}
	\setcounter{tocdepth}{2}
	\else

\section{Reusable SKE with Certified Deletion}\label{sec:reusable_SKE_cd}
We present the definition and construction of reusable SKE with certified deletion in this section.
First, we recall the definition of secret key NCE.

\begin{definition}[Secret Key NCE (Syntax)]\label{def:sk_nce_syntax}
A secret key NCE scheme is a tuple of PPT algorithms $(\keygen,\Enc,\Dec,\allowbreak\Fake,\Reveal)$ with plaintext space $\Ms$.
\begin{description}
    \item [$\keygen(1^\secp)\ra (\ek,\dk,\aux)$:] The key generation algorithm takes as input the security 
    parameter $1^\secp$ and outputs a key pair $(\ek,\dk)$ and an auxiliary information $\aux$.
    \item [$\Enc(\ek,m)\ra \ct$:] The encryption algorithm takes as input $\ek$ and a plaintext $m\in\cM$ and outputs a ciphertext $\ct$.
    \item [$\Dec(\dk,\ct)\ra m^\prime \mbox{ or }\bot$:] The decryption algorithm takes as input $\dk$ and $\ct$ and outputs a plaintext $m^\prime$ or $\bot$.
    \item [$\Fake(\ek,\aux)\ra \tlct$:] The fake encryption algorithm takes $\dk$ and $\aux$, and outputs a fake ciphertext $\tlct$.
    \item [$\Reveal(\ek,\aux,\tlct,m)\ra \tldk $:] The reveal algorithm takes $\ek,\aux,\tlct$ and $m$, and outputs a fake secret key $\tldk$.
\end{description}  
\end{definition}

Correctness is similar to that of PKE, so we omit it.
\begin{definition}[Receiver Non-Committing (RNC) Security for SKE]\label{def:sk_nce_security}
A secret key NCE scheme is RNC secure if it satisfies the following.
Let $\Sigma=(\keygen, \Enc, \Dec, \Fake,\Reveal)$ be a secret key NCE scheme.
We consider the following security experiment $\expb{\Sigma,\cA}{sk}{rec}{nc}(\secp,b)$.

\begin{enumerate}
    \item The challenger computes $(\ek,\dk,\aux) \lrun \keygen(1^\secp)$ and sends $1^\secp$ to the adversary $\cA$.
    \item $\cA$ sends an encryption query $m$ to the challenger. The challenger computes and returns $\ct\lrun \Enc(\ek,m)$ to $\cA$. This process can be repeated polynomially many times.
    \item $\cA$ sends a query $m \in \Ms$ to the challenger.
    \item The challenger does the following.
    \begin{itemize}
    \item If $b =0$, the challenger generates $\ct \lrun \Enc(\ek,m)$ and returns $(\ct,\dk)$ to $\cA$.
    \item If $b=1$, the challenger generates $\tlct \lrun \Fake(\ek,\aux)$ and $\tldk \lrun \Reveal(\ek,\aux,\tlct,m)$ and returns $(\tlct,\tldk)$ to $\cA$.
    \end{itemize}
    \item Again $\cA$ can send encryption queries.
    \item $\cA$ outputs $b'\in \bit$.
\end{enumerate}
Let $\advc{\Sigma,\cA}{sk}{rec}{nc}(\secp)$ be the advantage of the experiment above.
We say that the $\Sigma$ is RNC secure if for any QPT adversary, it holds that
\begin{align}
\advc{\Sigma,\cA}{sk}{rec}{nc}(\secp)\seteq \abs{\Pr[ \expb{\Sigma,\cA}{sk}{rec}{nc}(\secp, 0)=1] - \Pr[ \expb{\Sigma,\cA}{sk}{rec}{nc}(\secp, 1)=1] }\leq \negl(\secp).
\end{align}
\end{definition}

\begin{definition}[Reusable SKE with Certified Deletion (Syntax)]\label{def:reusable_sk_cert_del}
A secret key encryption scheme with certified deletion is a tuple of quantum algorithms $(\keygen,\Enc,\Dec,\Delete,\Vrfy)$ with plaintext space $\Ms$ and key space $\Ks$.
\begin{description}
    \item[$\keygen (1^\secp) \ra \sk$:] The key generation algorithm takes as input the security parameter $1^\secp$ and outputs a secret key $\sk \in \Ks$.
    \item[$\Enc(\sk,m) \ra (\vk,\ct)$:] The encryption algorithm takes as input $\sk$ and a plaintext $m\in\Ms$ and outputs a verification key $\vk$ and a ciphertext $\ct$.
    \item[$\Dec(\sk,\ct) \ra m^\prime$:] The decryption algorithm takes as input $\sk$ and $\ct$ and outputs a plaintext $m^\prime \in \Ms$ or $\bot$.
    \item[$\Delete(\ct) \ra \cert$:] The deletion algorithm takes as input $\ct$ and outputs a certification $\cert$.
    \item[$\Vrfy(\vk,\cert)\ra \top \mbox{ or }\bot$:] The verification algorithm takes $\vk$ and $\cert$ and outputs $\top$ or $\bot$.
\end{description}
\end{definition}
A difference the definition by Broadbent and Islam and ours is that $\Enc$ outputs not only $\CT$ but also $\vk$, and $\vk$ is used in $\Vrfy$ istead of $\sk$.

\begin{definition}[Correctness for reusable SKE with Certified Deletion]\label{def:reusable_sk_cd_correctness}
There are two types of correctness. One is decryption correctness and the other is verification correctness.
\begin{description}
\item[Decryption correctness:] For any $\secp\in \N$, $m\in\Ms$, 
\begin{align}
\Pr\left[
\Dec(\sk,\ct)\ne m
\ \middle |
\begin{array}{ll}
\sk\lrun \keygen(1^\secp)\\
(\vk,\ct) \lrun \Enc(\sk,m)
\end{array}
\right] 
\le\negl(\secp).
\end{align}

\item[Verification correctness:] For any $\secp\in \N$, $m\in\Ms$, 
\begin{align}
\Pr\left[
\Vrfy(\vk,\cert)=\bot
\ \middle |
\begin{array}{ll}
\sk\lrun \keygen(1^\secp)\\
(\vk,\ct) \lrun \Enc(\sk,m)\\
\cert \lrun \Delete(\ct)
\end{array}
\right] 
\le\negl(\secp).
\end{align}
\end{description}
\end{definition}

\begin{definition}[Certified Deletion Security for Reusable SKE]\label{def:reusable_sk_certified_del}
Let $\Sigma=(\keygen, \Enc, \Dec, \Delete, \Vrfy)$ be a secret key encryption with certified deletion.
We consider the following security experiment $\expb{\Sigma,\cA}{sk}{cert}{del}(\secp,b)$.

\begin{enumerate}
    \item The challenger computes $\sk \la \keygen(1^\secp)$.
    \item $\cA$ sends an encryption query $m$ to the challenger. The challenger computes $(\vk,\ct)\lrun \Enc(\sk,m)$ to $\cA$ and returns $(\vk,\ct)$ to $\cA$. This process can be repeated polynomially many times.
    \item $\cA$ sends $(m_0,m_1)\in\cM^2$ to the challenger.
    \item The challenger computes $(\vk_b,\ct_b) \la \Enc(\sk,m_b)$ and sends $\ct_b$ to $\cA$.
    \item Again, $\cA$ can send encryption queries.
    \item At some point, $\cA$ sends $\cert$ to the challenger.
    \item The challenger computes $\Vrfy(\vk_b,\cert)$. If the output is $\bot$, the challenger sends $\bot$ to $\cA$.
    If the output is $\top$, the challenger sends $\sk$ to $\cA$.
    \item If the challenger sends $\bot$ in the previous item, $\cA$ can send encryption queries again.
    \item $\cA$ outputs $b'\in \bit$.
\end{enumerate}
Let $\advc{\Sigma,\cA}{sk}{cert}{del}(\secp)$ be the advantage of the experiment above.
We say that the $\Sigma$ is IND-CPA-CD secure if for any QPT $\cA$, it holds that
\begin{align}
\advc{\Sigma,\cA}{sk}{cert}{del}(\secp)\seteq \abs{\Pr[ \expb{\Sigma,\cA}{sk}{cert}{del}(\secp, 0)=1] - \Pr[ \expb{\Sigma,\cA}{sk}{cert}{del}(\secp, 1)=1] }\leq \negl(\secp).
\end{align}
\end{definition}

\paragraph{Our reusable SKE scheme.}
We construct $\Sigma_{\mathsf{r}\skcd} =(\keygen,\Enc,\Dec,\Delete,\Vrfy)$ with plaintext space $\Ms$ from an one-time SKE with certified deletion scheme $\Sigma_{\mathsf{o}\skcd}=\SKE.(\Gen,\Enc,\Dec,\Delete,\Vrfy)$ with plaintext space $\Ms$ and key space $\Ks$ and a secret key NCE scheme $\Sigma_{\nce}=\NCE.(\keygen,\Enc,\Dec,\Fake,\Reveal)$ with plaintext space $\Ks$.

\begin{description}
\item[$\keygen(1^\secp)$:] $ $
\begin{itemize}
\item Generate $(\nce.\ek,\nce.\dk,\nce.\aux)\lrun \NCE.\keygen(1^\secp)$ and output $\sk \seteq (\nce.\ek,\nce.\dk)$.
\end{itemize}
\item[$\Enc(\sk,m)$:] $ $
\begin{itemize}
	\item Parse $\sk = (\nce.\ek,\nce.\dk)$.
\item Generate $\oske.\sk \lrun \SKE.\Gen(1^\secp)$.
\item Compute $\nce.\ct \lrun \NCE.\Enc(\nce.\ek,\oske.\sk)$ and $\oske.\ct \lrun \SKE.\Enc(\oske.\sk,m)$.
\item Output $\ct \seteq (\nce.\ct,\oske.\ct)$ and $\vk \seteq \oske.\sk$.
\end{itemize}
\item[$\Dec(\sk,\ct)$:] $ $
\begin{itemize}
\item Parse $\sk = (\nce.\ek,\nce.\dk)$ and $\ct = (\nce.\ct,\oske.\ct)$.
\item Compute $\sk^\prime \lrun \NCE.\Dec(\nce.\dk,\nce.\ct)$.
\item Compute and output $m^\prime \lrun \SKE.\Dec(\sk^\prime,\oske.\ct)$.
\end{itemize}
\item[$\Delete(\ct)$:] $ $
\begin{itemize}
\item Parse $\ct= (\nce.\ct,\oske.\ct)$.
\item Generate $\oske.\cert \lrun \SKE.\Delete(\oske.\ct)$.
\item Output $\cert \seteq \oske.\cert$.
\end{itemize}
\item[$\Vrfy(\vk,\cert)$:] $ $
\begin{itemize}
\item Parse $\vk = \oske.\sk$ and $\cert = \oske.\cert$.
\item Output $b \lrun \SKE.\Vrfy(\oske.\sk,\oske.\cert)$.
\end{itemize}
\end{description}

\begin{theorem}\label{thm:reusable_ske_cd_from_sk_cd_and_ske}
If $\Sigma_{\nce}$ is RNC secure and $\Sigma_{\mathsf{o}\skcd}$ is OT-CD secure, $\Sigma_{\mathsf{r}\skcd}$ is IND-CPA-CD secure.
\end{theorem}
We omit the proof of this theorem since it is almost the same as the proof of~\cref{thm:pke_cd_from_sk_cd_and_pke}.

	\fi
\fi

\ifnum\cameraready=1
\else
\ifnum\submission=1
\newpage
\setcounter{tocdepth}{1}
\tableofcontents
\else
\fi
\fi

\end{document}